\documentclass[sigconf]{aamas}

\usepackage{balance} 

\usepackage{algorithm}
\usepackage{algpseudocode}
\usepackage{graphicx}
\usepackage{subcaption}
\usepackage{listings}

\newtheorem{definition}{Definition}
\newtheorem{theorem}{Theorem}
\newtheorem{lemma}{Lemma}

\newtheorem{assumption}{Assumption}
\newtheorem{corollary}{Corollary}
\newenvironment{proofsk}{\renewcommand{\proofname}{Proof sketch}\proof}{\endproof}

\setcopyright{none}
\acmConference[AAMAS '26]{ACM conference}{2026}
{TBD}{Under review}
\copyrightyear{2026}
\acmYear{2026}
\acmDOI{}
\acmPrice{}
\acmISBN{}

\acmSubmissionID{748}

\title[AAMAS-2026 Formatting Instructions]{Constant-Memory Strategies in Stochastic Games:\\ Best Responses and Equilibria}

\author{Fengming Zhu}
\affiliation{
  \institution{The Hong Kong University of Science and Technology}
  \city{Hong Kong SAR}
  \country{China}}
\email{fzhuae@connect.ust.hk}

\author{Fangzhen Lin}
\affiliation{
  \institution{The Hong Kong University of Science and Technology}
  \city{Hong Kong SAR}
  \country{China}}
\email{flin@cs.ust.hk}

\begin{abstract}
Stochastic games have become a prevalent framework for studying long-term multi-agent interactions, especially in the context of multi-agent reinforcement learning. In this work, we comprehensively investigate the concept of constant-memory strategies in stochastic games. We first establish some results on best responses and Nash equilibria for behavioral constant-memory strategies, followed by a discussion on the computational hardness of best responding to mixed constant-memory strategies. Those theoretic insights are later verified on several sequential decision-making testbeds, including the \textit{Iterated Prisoner's Dilemma}, the \textit{Iterated Traveler's Dilemma}, and the \textit{Pursuit} domain.
This work aims to enhance the understanding of theoretical issues in single-agent planning under multi-agent systems, and uncover the connection between decision models in single-agent and multi-agent contexts.
\textit{
The code is available at \texttt{https://github.com/Fernadoo/Const-Mem}.
}

\end{abstract}


\newcommand{\BibTeX}{\rm B\kern-.05em{\sc i\kern-.025em b}\kern-.08em\TeX}

\begin{document}


\pagestyle{fancy}
\fancyhead{}


\maketitle 


\section{Introduction}

Various real-world situations that involve long-term interactions among a group of participants can be modeled as stochastic games, such as negotiation between multiple stakeholders~\cite{baarslag2016learning, de2017negotiating, baarslag2024multi}, bidding and mechanism design in repeated auctions~\cite{iyer2014mean, balseiro2015repeated, guo2019learning, guo2024generative, su2024auctionnet, shen2020reinforcement}, multi-agent teamwork~\cite{sharon2015conflict, mirsky2022survey, stern2019multi-overview}, and even human-robot collaboration~\cite{zhang2023building,puig2023habitat}.
Stochastic games, also known as Markov games, model the interactions of these multi-agent systems as  a Markov chain over a set of states, where the transitions are triggered by joint actions and are potentially stochastic.

The formalization of stochastic games was first proposed in Shapley's seminal work~\cite{shapley1953stochastic}.
A perfectly rational agent in a stochastic game is supposed to make use of all past histories to determine the next action, and therefore, the notion of strategies, defined as mappings from all possible histories to actions, is inherently complex.
The fact that there are infinitely many strategies prohibits the direct application of Nash's theorem for establishing any existence result of equilibria.
However, the stationary transitions of stochastic games inevitably draw attention to a highly special subclass of time-independent and memoryless strategies that only consider the current states while discarding all past histories, termed \textit{stationary strategies}.
Indeed, the existence of equilibria formed by stationary strategies in 
$n$-player general-sum stochastic games was later proven by~\citeauthor{fink1964equilibrium}~\cite{fink1964equilibrium} and~\citeauthor{takahashi1964equilibrium}~\cite{takahashi1964equilibrium}, under mild assumptions.
Despite being highly restricted in terms of expressiveness, the notion of stationary strategies has enabled the community to practically investigate some complex real-world applications, particularly by resorting to multi-agent reinforcement learning (MARL) techniques, as advocated by~\citeauthor{littman1994markov}~\cite{littman1994markov} and implemented in a line of subsequent work~\cite{lowe2017multi, foerster2018counterfactual, rashid2020monotonic, yu2022the}.

Notably, one would naturally expect strategies in other less restricted forms that can encode a broader class of behavioral patterns, hoping to achieve better payoff outcomes.
For example, in the Iterated Prisoner's Dilemma (IPD), if only stationary strategies are considered, there is a unique Nash equilibrium where both players choose to \textit{defect} all the time, resulting in the lowest overall payoff.
However, even with the ability to remember only one past action played by the opponent, the well-known \textit{Tit-For-Tat} (TFT) strategy (start with \textit{cooperation}) can be devised.
One can easily see that if both players adopt the TFT strategy, they will follow a trajectory of both \textit{cooperating} throughout the game, resulting in a Nash equilibrium with the highest possible payoff.
Apart from other forms of representation, such as strategies represented as finite automata~\cite{rubinstein1986finite, ben1990complexity, zuo2015optimal} and even Turing machines~\cite{megiddo1986play, knoblauch1994computable, nachbar1996non, chen2017bounded},
we focus our main effort on investigating the notion of \textit{constant-memory strategies}, i.e., mappings from history segments of bounded lengths to actions, mainly because it directly relates to the concept of bounded rationality~\cite{simon1990bounded} in general,
and is highly implementable using function approximators like Recurrent Neural Networks~\cite{hochreiter1997long, cho2014learning} and Transformers~\cite{vaswani2017attention} in practice.
Note that this notion has been preliminarily investigated by ~\citeauthor{chen2017k}~\cite{chen2017k} and ~\citeauthor{wang2019pure}~\cite{wang2019pure}.
However, they only focus on behavioral strategy best responses for repeated games, without further discussion on either Nash equilibria or mixed strategies.


In this paper, we comprehensively study the theoretical properties associated with \textit{constant-memory strategies} in \textit{stochastic games}.
We begin by presenting the following two results:
\begin{enumerate}
	\item \textit{A Characterization of Best Responses:} Given a constant-memory strategy profile adopted by the opponents, there always exists a deterministic constant-memory strategy that makes use of the same length of memory acting as a pure strategy best response.
	\item \textit{An Existence Result of Equilibria:} Given any finite length of memory, there always exists a Nash equilibrium where all agents adopt constant-memory (but not necessarily deterministic) strategies using that length of memory.
\end{enumerate}
As a side benefit of using memories of constant lengths, any strategy that uses a shorter memory can always be implemented by one that uses a longer memory.
Therefore, the above two results directly imply that any NE formed by shorter-length-memory strategies can be transformed into an NE formed by longer-length-memory strategies, suggesting that the longer the memory used by the strategies, the richer the equilibria one can potentially expect.

Additionally, we provide further results about best responses against mixed constant-memory strategies, mathematically defined as those sampled from a set of support strategies with certain probabilities.
This is associated with broad applications in the domain of opponent modeling~\cite{carmel1998explore, albrecht2015game, albrecht2018autonomous, zhu2025single}, particularly for type-based methods~\cite{albrecht2015game, albrecht2019convergence, albrecht2016belief, zhu2025single}.
However, we demonstrate that:
\begin{enumerate}
	\item \textit{An Negative Result on Strategy Equivalence:} An opponent with a mixed constant-memory strategy may not correspond to an equivalent opponent with a single (behavioral) constant-memory strategy in terms of resulting in the same payoff.
	\item \textit{A Negative Result on Best Responses:} The best response against a mixed constant-memory strategy is not necessarily constant-memory, and computing such best responses is computationally hard, possibly even not computable.
\end{enumerate}
In spite of these negative results, we do provide a computational model for solving the best response against a mixed strategy, which 1) also serves as the evidence that those computational models proposed by~\citeauthor{zhu2025single}~\cite{zhu2025single} do not over-complicate the problem; and 2) can be carried over to the methods in~\cite{zhu2025single} which only assume stationary strategies.


\section{Preliminaries}

The whole system where the agents interact is modelled as a \textit{stochastic game} (SG, also known as Markov games)~\citep{shapley1953stochastic,solan2015stochastic},
which 
can be seen as an extension of both \textit{normal-form games} (to dynamic situations with stochastic transitions) and \textit{Markov decision processes} (to strategic situations with multiple agents).
A stochastic game
is a 5-tuple $\langle \mathcal{N},\mathcal{S}, \mathcal{A}, T, R \rangle$ given as follows,
\begin{enumerate}
	\item $\mathcal{N}$ is a finite set of $n$ agents.
	\item $\mathcal{S}$ is a finite set of (environmental) states.
	\item $\mathcal{A} = \mathcal{A}_1 \times \cdots \times \mathcal{A}_n$ is a set of joint actions, where $\mathcal{A}_i$ is the action set of agent $i$. In particular, we write $a_i$ as the action of agent $i$ and the one without any subscript $a = (a_i, a_{-i})$  as the joint action.
	\item $T: \mathcal{S} \times \mathcal{A}_1 \times \cdots \mathcal{A}_n \mapsto \Delta(\mathcal{S})$ defines stochastic transitions among states.
	\item $R_i: \mathcal{S} \times \mathcal{A}_1 \times \cdots \mathcal{A}_n \mapsto \mathbb{R}$ denotes the immediate rewards for agent $i$.
\end{enumerate}

To define best responses and hence equilibria, we need to first define strategies and objectives.

Assuming complete observability and perfect recall,
a perfectly rational agent should utilize the entire history,
while in memory-restricted cases, an agent can only devise strategies based on past memories of finite lengths.
We denote the space of all possible histories of length $K\in \mathbb{N}$ as $\mathcal{H}^K \triangleq (\mathcal{S} \times \mathcal{A})^K$. In particular, when $K = 0$, we have $\mathcal{H}^0 = \emptyset$ meaning that no history can be utilized.
Then, given any non-negative integer $K$, a $K$-memory strategy for agent $i$ is a mapping from all possible histories with lengths less then or equal to $K$ and the current states to (possibly randomized) actions, mathematically denoted as $\pi_i: \mathcal{H}^{\leq K} \times \mathcal{S} \mapsto \Delta(\mathcal{A}_i)$ where $\mathcal{H}^{\leq K} \triangleq \cup_{k=0}^K \mathcal{H}^k$.
Let $\Pi_i^{K}$
denote the set of all such $K$-memory strategies for agent $i$.
For convenience, we let $\mathcal{H}^\infty = (\mathcal{S} \times \mathcal{A})^*$ denote the set of complete histories that an agent with perfect recall can possibly memorize, and therefore, $\Pi_i^\infty$ is the set of all possible infinite-memory strategies for agent $i$ of the form $\pi_i: (\mathcal{S} \times \mathcal{A})^* \times \mathcal{S} \mapsto \Delta(\mathcal{A}_i)$.
A direct consequence is that $\Pi_i^{K} \subseteq \Pi_i^{K'} \subseteq \Pi_i^\infty$ for any non-negative $K \leq K'$.
Among them, one of the most popular class of strategies is $\Pi_i^{0}$, termed \textit{stationary strategies.}
\textbf{Note that an agent capable of performing infinite-memory strategies can deliberately adopt a constant-memory strategy.}
To be clear, we use the term \textit{constant-memory} to distinguish from those infinite-memory strategies, and use the term \textit{K-memory} when this specific $K$ needs to be emphasized.

The objective for each agent is to maximize its accumulated discounted rewards (a.k.a. the discounted-payoff scenario, as opposed to the average-payoff scenario).
We let $R_{i,t}$ denote the reward signaled to agent $i$ at step $t$, similarly for $S_t$ and $a_{i,t}$, then the overall utility under a policy profile $(\pi_i, \pi_{-i})$ starting from any arbitrary state $S\in \mathcal{S}$ is
\begin{equation}
u_i(S;\pi_i, \pi_{-i}) = \mathbb{E}_{(\pi_i, \pi_{-i})}\Big[ \sum_{t=0}^\infty \gamma^t R_{i,t} \Big| S_0=S \Big]
\label{eq:utility_def}
\end{equation}
$\pi_i$ is said to be the best response of $\pi_{-i}$, denoted as $\pi_i \in BR(\pi_{-i})$, if 
\begin{equation}
\forall S\in \mathcal{S}, \pi_i'\in \Pi_i^\infty, u_i(S;\pi_i, \pi_{-i}) \geq u_i(S;\pi_i', \pi_{-i})
\end{equation}
requiring that a $\pi_i$ must outperform any other in $\Pi_i^\infty$ to serve as the best response.
Note that, to compare the values of two strategy profiles, one must ensure that the limit of the right-hand side (RHS) in Equation~(\ref{eq:utility_def}) exists in the first place.
Also note that, some pairs of $\pi_i$ and $\pi_i'$ may not be comparable in the above sense, as as this comparison requires value dominance across all possible states.

\section{Best Responses and Nash Equilibria}
\label{sec:main_res}


One should be aware of the following fact for single-agent Markov Decision Processes (MDPs)~\cite{puterman2014markov} in the first place,  which will be considered as a lemma for the remainder of this paper.

\begin{lemma}
	For a (single-agent) MDP $\langle S, A, T, R, \gamma\rangle$, the following two are equivalent,
	\begin{enumerate}
		\item Searching for a policy $\pi_*: S \mapsto \Delta(A)$ that maximizes the accumulated rewards $\mathbb{E}_{\pi_*}\sum_{t=0}^\infty[\gamma^tR_t]$, for any initial $s\in S$.
		\item Solving the Bellman optimality equation below
		\[
		\forall s\in S, v_*(s) = \max_{a\in A}\Big[ R(s,a) + \gamma \sum_{s'\in S} T(s'|s,a)v_*(s') \Big],
		\]
		and then extracting the policy from the optimal value function
		\[
		\pi_*(s) \in \arg\max_{a\in A}\Big[ R(s,a) + \gamma \sum_{s'\in S} T(s'|s,a)v_*(s') \Big].
		\]
	\end{enumerate}
\end{lemma}

\begin{assumption}
	We assume that agents are independent of each other and rewards are bounded.
\end{assumption}

We first characterize the best response of an agent
when all the other opponents are equipped with constant-memory strategies with the same non-negative (and finite) memory length.

\begin{theorem}
Given $\pi_{j} \in \Pi_{j}^K$ with $K \in \mathbb{Z}$ for all $j\neq i$ , i.e., all the other agents are adopting constant-memory strategies with the same finite memory length $K$, it is sufficient for agent $i$ to best respond with a $K$-memory strategy as well.	
\label{thm:kmem_br}
\end{theorem}

\begin{proofsk}

As the full proof involves some fundamental (and probably tedious) derivations, we defer it to Appendix~\ref{app:proof:thm_kmem_br}.
The main issue is that, in general the decision process from the perspective of agent $i$ is an infinite MDP where states encompass infinitely many histories (of all possible lengths), and transitions/rewards are jointly controlled by the stochastic game itself as well as the other opponents.
Thus, the goal of this proof is to show that there indeed exists a finite MDP with the same effect.
Here, we present a proof sketch by construction, which is, in fact, a consequence of the full proof, and can be approached from a more direct perspective.

Given $\pi_{j} \in \Pi_{j}^K$ for all $j\neq i$, agent $i$ is then faced with an MDP with environmental states augmented by finite-length histories,
denoted as $\mathcal{M}^K(\pi_{-i}) = \langle \mathcal{H}^{\leq K}\times \mathcal{S}, \mathcal{A}_i, T^K_{\pi_{-i}}, R^K_{\pi_{-i}}, \gamma \rangle$,
	\begin{itemize}
		\item $\mathcal{A}_i$ and $\gamma$ are inherited from the previous setup,
		\item A state is now consisting of the past $K$ steps plus the current environmental state, resulting in a space of $\mathcal{H}^{\leq K}\times \mathcal{S}$,
		\item Transitions are now made among the augmented states, i.e., for every pair $(H', S')$, $(H, S) \in \mathcal{H}^{\leq K}\times \mathcal{S}$, and $a_i \in \mathcal{A}_i$,
		\[
		\begin{split}
		& T^K_{\pi_{-i}}(H', S'| H, S, a_i) \triangleq \\
		& 
		\begin{cases}
			T(S'|S,a)\pi_{-i}(a_{-i}|H, S),
			& \text{ if } H' = slide_K(H, S, (a_i, a_{-i})) \\
			0,
			& \text{ otherwise}\\
		\end{cases}	
		\end{split}
		\]
		where $slide_K(H, S, (a_i, a_{-i}))$ means to discard the earliest step if it is more then $K$ steps away, and append the latest state and action profile.
		\item The reward for each $(H, S) \in \mathcal{H}^{\leq K}\times \mathcal{S}$ and $a_i \in \mathcal{A}_i$ is
		\[
		R^K_{\pi_{-i}}(H, S, a_i) \triangleq \sum_{a_{-i} \in \mathcal{A}_{-i}}R_i(S,a)\pi_{-i}(a_{-i}|H, S)
		\]
	\end{itemize}
Among all the optimal solutions of $\mathcal{M}^K(\pi_{-i})$, there must exist a stationary (and deterministic) policy, i.e. $\pi_{*}: \mathcal{H}^{\leq K}\times \mathcal{S} \mapsto \mathcal{A}_i$, which corresponds to a $K$-memory (and pure) strategy best response of agent $i$ against $\pi_{-i} \in \Pi_{-i}^K$.

\textit{This proof is summarized in our code implementation.\footnote{Please refer to the notebook \texttt{code/kMemBR.ipynb} in the codebase.}}
\end{proofsk}




One can immediately see the following corollary where the opponents may use constant-memory strategies but with potentially different memory lengths.
The justification is straightforward: all the opponents can be jointly viewed as a ``super-agent'', and consequently this ``super-agent'' is adopting a $(\max\{K_j\}_{j\neq i})$-memory strategy.
Therefore, the best response for the pivotal agent shall also be of $(\max\{K_j\}_{j\neq i})$-memory.

\begin{corollary}
Given $\pi_{j} \in \Pi_{j}^{K_j}$ with each $K_j \in \mathbb{Z}$ for all $j\neq i$, i.e., all the other agents are adopting constant-memory strategies but with varying memory lengths, it is sufficient for agent $i$ to best respond with a $(\max\{K_j\}_{j\neq i})$-memory strategy.	
\end{corollary}

%



As best responses are well established, we will examine whether an equilibrium exists when everyone best responds to one another.

\begin{definition}[Nash Equilibrium]
	A strategy profile $\{\pi_i^*\}_{i \in \mathcal{N}}$ is a Nash equilibrium (NE) if
	\[
	\forall i \in \mathcal{N}, \pi_i^* \in BR(\pi_{-i}^*)
	\]
\end{definition}

We first need the following lemma. It is important to note that the following lemma only asserts the existence of a fixed point, but does not guarantee the presence of a contraction mapping.

\begin{lemma}[Brouwer's Fixed-Point Theorem~\cite{brouwer1911abbildung}]
	Let $\Delta = \prod_{l=1}^L \Delta_{m_l}$, where each $\Delta_{m_l}$ is a simplex in $\mathbb{R}^{m_l+1}$. If $f: \Delta \mapsto \Delta$ is a continuous mapping, then $f$ has a fixed point.
\label{lm:fixpoint}
\end{lemma}

\begin{theorem}
\label{thm:ne_beh}
	There exists an NE when the agents are all adopting $K$-memory strategies, for any arbitrary non-negative finite $K$.
\end{theorem}

\begin{proof}
\textit{As previously, we also summarize this proof into 
a ready-to-run code implementation.\footnote{Please refer to \texttt{code/kMemNE\_full.ipynb} in the codebase.}}

Given a non-negative finite $K$, to establish a Nash equilibrium we need to prove there is a solution to the system of equations defined by
\[
	\forall i \in \mathcal{N}, \pi_i \in \Pi_i^K \land \pi_i \in BR(\pi_{-i})
\]
More specifically, the following equations should be satisfied simultaneously for any $(H,S) \in \mathcal{H}^{\leq K}\times \mathcal{S}$, and for every $i \in \mathcal{N}$,
\begin{equation}
\begin{split}
v_i(H,S) = 
& \max_{a_i\in\mathcal{A}_i} \Big[ R^K_{\pi_{-i}}(H, S, a_i)\\
& + \gamma \sum_{H',S'} T^K_{\pi_{-i}}(H', S'| H, S, a_i) \cdot v_i(H',S') \Big] \\
\pi_i(H,S) \in
& \arg\max_{a_i\in\mathcal{A}_i} \Big[ R^K_{\pi_{-i}}(H, S, a_i) \\
& + \gamma \sum_{H',S'} T^K_{\pi_{-i}}(H', S'| H, S, a_i) \cdot v_i(H',S') \Big] \\
\end{split}
\label{eq:ne_boe}
\end{equation}
We will first construct a mapping to iteratively refine the strategies, and then show that there is a bijection between the fixed points of this mapping and the solutions to the above system of equations.

From each agent $i$'s perspective, with the opponent's strategies given as $\pi_{-i}$, it shall evaluate the value of its own strategy by the Bellman expectation equation, i.e.,
\begin{equation}
\begin{split}
v_i|_{\pi_i}^{\pi_{-i}}(H,S) & = \sum_{a_i\in\mathcal{A}_i} \pi_i(a_i|H, S) \cdot Q_i|_{\pi_i}^{\pi_{-i}}(H,S, a_i) \\
Q_i|_{\pi_i}^{\pi_{-i}}(H,S, a_i) & = 
 R^K_{\pi_{-i}}(H, S, a_i)  \\
& + \gamma \sum_{H',S'} T^K_{\pi_{-i}}(H', S'| H, S, a_i) \cdot v_i|_{\pi_i}^{\pi_{-i}}(H',S')	
\end{split}
\label{eq:kmem_bee}
\end{equation}
where $v_i|_{\pi_i}^{\pi_{-i}}$ is the value function evaluated using $\pi_i$ against $\pi_{-i}$.
One should first note that there is a unique solution satisfying Equation~(\ref{eq:kmem_bee}) simultaneously for all $i \in \mathcal{N}$. 
Please refer to Appendix~\ref{app:proof:thm_ne} for this omitted proof.
We then define the advantage as
\begin{equation}
\phi_{i,a_i}(\pi_i, H, S) = \max\{0, Q_i|_{\pi_i}^{\pi_{-i}}(H,S, a_i) - v_i|_{\pi_i}^{\pi_{-i}}(H,S)\}
\label{eq:refine}
\end{equation}
A refinement mapping $\Gamma: \{\Pi_i^K\}_{i\in \mathcal{N}} \mapsto \{\Pi_i^K\}_{i\in \mathcal{N}}$ is constructed
for each $i \in \mathcal{N}$,
\begin{equation}
\pi_i(a_i|H, S) \mapsto
\frac{\pi_i(a_i|H, S) + \phi_{i,a_i}(\pi_i, H, S)}{\sum_{b_i\in \mathcal{A}_i}\pi_i(b_i|H, S) + \phi_{i,b_i}(\pi_i, H, S)}
\label{eq:ref_mapping}
\end{equation}
By Lemma~\ref{lm:fixpoint}, $\Gamma$ has at least one fix point, as each state-action mapping is a simplex $\Delta_{|\mathcal{A}_i| - 1}$ and $\Gamma$ is continuous.

If $\{\pi_i\}_{i \in \mathcal{N}}$ is already an NE, then all $\phi$'s will be zeros, making it a fixed point of $\Gamma$.

Conversely, we can show that any arbitrary fixed point of $\Gamma$, say $\{\hat \pi_i\}_{i \in \mathcal{N}}$, is also an NE.
As $v$-functions are averaging over $Q$-functions, there must exist an $a_i'\in \mathcal{A}_i$, such that (fixing an $(H,S)$) $$\hat\pi_i(a_i'|H,S) > 0,\text{ and }Q_i|_{\hat\pi_i}^{\hat\pi_{-i}}(H,S, a_i') - v_i|_{\hat\pi_i}^{\hat\pi_{-i}}(H,S)\leq 0$$
By Equation~(\ref{eq:refine}), we have $\phi_{i,a_i'}(\hat\pi_i, H, S) = 0$. Given that $\{\hat \pi_i\}_{i \in \mathcal{N}}$ is already a fixed point, by definition $\{\hat \pi_i\}_{i \in \mathcal{N}} = \Gamma(\{\hat \pi_i\}_{i \in \mathcal{N}})$, and therefore, the normalization term (the denominator) must be exactly one. Due to the fact that $\phi$'s are always non-negative, then we can conclude that for all $b_i\in \mathcal{A}_i$, it must be the case $\phi_{i,b_i}(\hat \pi_i, H, S) =~0$.
Hence, $v_i|_{\hat\pi_i}^{\hat\pi_{-i}}(H,S) \geq \max_{a_i \in \mathcal{A}_i}{Q_i|_{\hat\pi_i}^{\hat\pi_{-i}}(H,S, a_i')}$. Consequently, the equality shall hold. One can then see it it exactly the case when the aforementioned Equation~(\ref{eq:ne_boe}) is satisfied. 
%
\end{proof}



The above existence result indicates the following.
Consider two agents playing a stochastic game, where agent 1 employs a two-memory strategy and agent 2 uses a three-memory strategy.
If agent~1 asserts that it will adhere to its two-memory strategy, agent~2 may identify another two-memory strategy as a best response, potentially yielding the same payoff but allowing for memory saving.
Conversely, if agent 2 can convince agent 1 that it will maintain its three-memory strategy, agent 1 may find it advantageous to switch to a three-memory strategy as a better response.

Note that the above theorem is not a direct consequence of Nash's existence theorem, as we only discuss randomizing actions within a single strategy, rather than randomizing across multiple strategies, which we will refer to as \textit{mixed strategies} in the next section.

Another benefit of constant-memory strategies is that any $K$-memory strategy can be implemented using a $K'$-memory strategy, provided that $K' \geq K$, by simply utilizing the most recent $K$ historical records. Thus, we arrive at the following corollary.

\begin{corollary}
	Any payoff profile that is reached by an NE under a $K$-memory strategy profile can also be realized by another NE under a $K'$-memory strategy profile, as long as $K' \geq K$.
\end{corollary}


\section{Best Responses to Mixed Strategies: A Tournament Perspective }
\label{sec:mixed_str}



We have established that the best response to a single (possibly randomized) constant-memory strategy results in another constant-memory strategy. The next natural question is: what is the best response to a set of constant-memory strategies played according to a specified distribution? A further related question is: can a mixed strategy be converted into a singleton constant-memory strategy? If this is feasible, then the best response must also be a constant-memory strategy.

In the following two subsections, we will show:
\begin{enumerate}
	\item In repeated games, if an agent encounters an opponent using a mixed zero-memory strategy, it will yield the same expected utility for this agent to play against an opponent with a transformed singleton zero-memory strategy.
	\item In general, when the game involves multiple states or the opponent employs a non-zero-memory strategy, then the best response will be hard to compute (possibly even not computable) and time-dependent, which may not be encoded as a finite-memory strategy. Consequently, it implies that the opponent's mixed strategy cannot be equivalently transformed into a singleton constant-memory strategy.
\end{enumerate}


\subsection{Mixed Strategies vs Behavioral Strategies}


We first emphasize the notion of \textit{match}.
When we say an agent $i$ adopts a $K$-memory strategy, it means that agent $i$ will select one strategy $\pi_i \in \Pi_i^K$ just before a match begins.
Once the agent has ``confirmed'' its strategy, it will \textbf{not} deviate to any other strategies during the match until the termination.
Note that some strategies, especially those in $\Pi_i^\infty$, may be semantically interpreted as ``learning'' or ``evolving'' strategies, as they gradually modify the decisions based on accumulated observations; however, each of them remains a singleton strategy within the strategy space $\Pi_i^\infty$.
From the perspective of a single agent, we may also use the term \textit{episode} interchangeably with \textit{match}, as is commonly done in the context of MDPs.
The \textit{overall utility} will be calculated as the expectation over all possible matches.

Now we are ready to explain the difference between a \textit{behavioral} strategy and a \textit{mixed} strategy.
Recall that a $K$-memory strategy of agent $i$ is defined as $\pi_i: \mathcal{H}^{\leq K} \times \mathcal{S} \mapsto \Delta(\mathcal{A}_i)$; it is also referred to as a \textit{behavioral} strategy as it can randomize over actions. By definition, a pure strategy that performs deterministic actions is also considered a behavioral strategy.
A \textit{mixed} strategy (for agent $i$ and of $K$-memory) first specifies its support set $\Pi^{K+}_i \subseteq \Pi_i^K$, where each behavioral strategy $\pi_i^\iota\in \Pi^{K+}_i$ will be selected with a positive probability $p_\iota$, before each match begins.
Thus, we use a tuple $(\Pi^{K+}_i, \vec p)$ to denote a mixed strategy for agent $i$.
Intuitively, when an agent is playing against a mixed strategy $(\Pi^{K+}_i, \vec p)$, it simply means this particular agent will encounter an opponent using the behavioral strategy $\pi_i^\iota \in \Pi^{K+}_i$ for a fraction $p_\iota$ of the whole time.

One may be particularly interested in a specific type of strategies, namely the behavioral strategy obtained by state-wise randomization over the actions according to the probability distribution provided by the mixed strategy.

\begin{definition}[Mixed-Strategy-Induced Behavioral Strategy]
	Given a mixed strategy $(\Pi^{K+}_i, \vec p)$, we define $\omega_{(\Pi^{K+}_i, \vec p)}$ as the behavior strategy induced by this mixed strategy.
	Mathematically, for each $(H,S) \in \mathcal{H}^{\leq K}\times \mathcal{S} $, 
	$
	\omega_{(\Pi^{K+}_i, \vec p)}(a_i|H, S) \triangleq \sum_\iota p_\iota \cdot \pi_i^\iota(a_i|H, S)
	$.
\label{def:mix2beh}
\end{definition}

The underlying intuition is that, instead of randomly selecting one of the support strategies at the beginning and sticking to it, we also allow an agent to switch to another strategy within the same probability distribution at each timestep during play, resulting in a single behavioral strategy that randomizes over each support strategy at every state.
One can see that if the original strategy is a mixed one over a set of $K$-memory support strategies,
then its induced behavioral strategy, according to Definition~\ref{def:mix2beh}, will still be a 
$K$-memory strategy, and its best response will also be a 
$K$-memory strategy, as stated in Theorem~\ref{thm:kmem_br}.



We will first demonstrate that in a special case where a stochastic game is reduced to a repeated game and the agents use stationary  strategies, a mixed strategy has the same effect as its induced behavioral strategy.
However, in general, if a game involves transitions across multiple states or the opponents adopt non-zero-memory strategies, such equivalence does not necessarily hold.

\begin{theorem}[Utility equivalence for repeated games]
\label{thm:equiv_repeated}
If the stochastic game is merely a repeated game, i.e. $\mathcal{S}$ is a singleton,
then an agent $i$'s overall utility when it plays against a mixed strategy $(\Pi^{0+}_{-i}, \vec p)$ will be the same as that when it plays against the induced behavioral strategy $\omega_{(\Pi^{0+}_{-i}, \vec p)}$.
\end{theorem}

\begin{proof}
Assume agent $i$ is performing any arbitrary strategy $\pi_i$.
To compute her expected return against the mixed strategy $(\Pi^{0+}_{-i}, \vec p)$, one needs to establish the Bellman expectation equation for each MDP $\mathcal{M}^0(\pi^\iota_{-i})$ induced by the opponent strategy $\pi_{-i}^\iota \in \Pi^{0+}_{-i}$,
\[
\begin{split}
V_\iota
& = \sum_{a_i \in \mathcal{A}_i} \pi_i(a_i)\cdot [R_{\pi_{-i}^\iota}(a_i) + \gamma V_\iota] \\
& = \sum_{a_i \in \mathcal{A}_i} \pi_i(a_i)\cdot [\sum_{a_{-i}\in \mathcal{A}_{-i}}\pi_{-i}^\iota(a_{-i})\cdot R_i(a_i, a_{-i}) + \gamma V_\iota]
\end{split}
\]
where $V_\iota$, as a shorthand,
denotes the expected return for agent~$i$ when it is playing $\pi_i$ against $\pi_{-i}^\iota$. Note that each $\mathcal{M}^0(\pi^\iota_{-i})$ is simply a one-state MDP. Then, one can get the following
\[
V_\iota = \frac{\sum_{a_i \in \mathcal{A}_i} \pi_i(a_i) \sum_{a_{-i}\in \mathcal{A}_{-i}}\pi_{-i}^\iota(a_{-i}) R_i(a_i, a_{-i})}{1 - \gamma}
\]
The overall expected utility against this mixed strategy is therefore
\begin{equation}
V_{mix} = \sum_\iota p_\iota \frac{\sum_{a_i \in \mathcal{A}_i} \pi_i(a_i) \sum_{a_{-i}\in \mathcal{A}_{-i}}\pi_{-i}^\iota(a_{-i}) R_i(a_i, a_{-i})}{1 - \gamma}
\label{eq:mixed_oppo}
\end{equation}
When this agent is instead playing against the mixed-strategy-induced behavioral strategy $\omega_{(\Pi^{0+}_{-i}, \vec p)}$, the consequent Bellman equation is
\[
\begin{split}
V_{beh}
& = \sum_{a_i \in \mathcal{A}_i} \pi_i(a_i) [R_{\omega_{(\Pi^{0+}_{-i}, \vec p)}}(a_i) + \gamma V_{beh}] \\
& = \sum_{a_i \in \mathcal{A}_i} \pi_i(a_i) [\sum_{a_{-i}\in \mathcal{A}_{-i}} R_i(a_i, a_{-i})\cdot \omega_{(\Pi^{0+}_{-i}, \vec p)}(a_{-i})+ \gamma V_{beh}] \\
& = \sum_{a_i \in \mathcal{A}_i} \pi_i(a_i) [\sum_{a_{-i}\in \mathcal{A}_{-i}} R_i(a_i, a_{-i})  \cdot (\sum_{\iota}p_\iota\cdot \pi^\iota_{-i}(a_{-i})) + \gamma V_{beh}] \\
\end{split}
\]
Thus, solving the equation yields
\begin{equation}
\begin{split}
V_{beh} = \frac{\sum_{a_i \in \mathcal{A}_i} \pi_i(a_i) \sum_{a_{-i}\in \mathcal{A}_{-i}} R_i(a_i, a_{-i})\sum_{\iota}p_\iota\cdot \pi^\iota_{-i}(a_{-i})}{1 - \gamma} 	
\end{split}
\label{eq:mixed_pi}
\end{equation}
By comparing Equation (\ref{eq:mixed_oppo}) and (\ref{eq:mixed_pi}), it is clear that $V_{mix} = V_{beh}$, up to different orders of summation.	
\end{proof}



\begin{theorem}
[Utility Equivalence Does Not Hold for General Stochastic Games]
In general, when a stochastic game involves multiple states,
an agent $i$'s overall utility when playing against a mixed stationary strategy $(\Pi^{0+}_{-i}, \vec p)$ is not necessarily the same as when playing against the induced behavioral strategy $\omega_{(\Pi^{0+}_{-i}, \vec p)}$.
\label{thm:noneq_sg}
\end{theorem}

\begin{proofsk}
We here provide some intuitions,
while the full proof is deferred to
Appendix~\ref{app:proof:thm_noneq_sg}.
The key issue is as follows.
Even when an agent plays against a mixed stationary strategy, 
her overall utility is the expectation of the returns of playing against each of the support strategies (each corresponding to a multi-step MDP), involving $|\Pi^{0+}_{-i}|$ contraction mappings.
However, when it plays against the induced behavioral strategy, its utility is computed by evaluating only one MDP induced by $\omega_{(\Pi^{0+}_{-i}, \vec p)}$, involving a contraction mapping that differs from any of the aforementioned~$|\Pi^{0+}_{-i}|$.
\end{proofsk}

As increasing either the length of memory or the number of environmental states results in a multi-state MDP from an individual agent's perspective,
a natural implication is that utility equivalence between a mixed strategy and its induced behavioral strategy does not necessarily hold for $K$-memory strategies once $K$ is positive, even in repeated games.	

One may further wonder whether a group of agents can form some equilibrium if all of them play mixed strategies, i.e.,
\[
\forall i \in \mathcal{N},\ (\Pi^{K+}_{i}, \vec p_i) \in BR((\Pi^{K+}_{-i}, \vec p_{-i})).
\]
With some additional assumptions, one can invoke Nash's existence theorem, as the game becomes finite.
Due to the page limit, we defer detailed formalization to Appendix~\ref{app:ne_mixed} for interested readers, while the application of such theoretic results remains an open problem.

\subsection{Computing BR to Mixed Strategies is Hard}


We will first show that computing the best response against a mixed $K$-memory strategy can be reduced to optimally solving an infinite-horizon \textit{partially observable MDPs} (POMDPs)~\cite{sondik1978optimal, kaelbling1998planning}.
It turns out the reduced ones belong to a subclass of generic POMDPs, namely \textit{Contextual MDP} (CMDPs), although it may not necessarily imply less challenging computation.
To show that this reduction does not complicate the original problem, we also construct a reduction from the problem of optimally solving CMDPs back to that of computing best responses against mixed strategies in stochastic games.



\begin{theorem}
Given a mixed strategy profile $(\Pi^{K+}_{-i}, \vec p)$ of the opponents, computing the best response for agent $i$ can be reduced to optimally solving an infinite-horizon POMDP.
\label{thm:memk_br_pomdp}
\end{theorem}

\begin{proofsk}
Here, we only provide the reduction to the corresponding POMDP, while the correctness of this reduction is left to
Appendix~\ref{app:proof:thm_memk_br_pomdp}. 
	The POMDP is given as the following tuple
	$$\langle \mathcal{H}^K\times\mathcal{S}\times\Pi^{K+}_{-i}, \mathcal{A}_i, \mathcal{H}^K\times\mathcal{S}, \mathbf{T}, \mathbf{O}, \mathbf{R}, \gamma\rangle$$
	\begin{enumerate}
		\item The set of underlying states is denoted by $\mathcal{H}^K\times\mathcal{S}\times\Pi^{K+}_{-i}$. That is, a state in this POMDP is the history segment and environment state of the completely observable stochastic game, along with the unobservable opponent strategies.
		\item As previously, $\mathcal{A}_i$ is the set of available  actions of agent $i$, while $\gamma$ is the discount factor.
		\item The set of observations that can be made by agent $i$ is denoted as $\mathcal{H}^K\times\mathcal{S}$.
		\item $\mathbf{T}: (\mathcal{H}^K\times\mathcal{S}\times\Pi^{K+}_{-i}) \times \mathcal{A}_i \mapsto \Delta(\mathcal{H}^K\times\mathcal{S}\times\Pi^{K+}_{-i})$ denotes the  transition function, mathematically defined as
		\[
		\begin{split}
		\mathbf{T}\Big((H', S', \pi_{-i}') \Big| & (H, S, \pi_{-i}), a_i \Big) \triangleq \\
		& \begin{cases}
			T^K_{\pi_{-i}}(H', S'| H, S, a_i) & \text{, if } \pi_{-i}' = \pi_{-i} \\
			0 & \text{, otherwise} 
		\end{cases}	
		\end{split}
		\]
		\item $\mathbf{O}: (\mathcal{H}^K\times\mathcal{S}\times\Pi^{K+}_{-i}) \mapsto \mathcal{H}^K\times\mathcal{S}$ denotes the deterministic observation function, mathematically defined as
		\[
		\mathbf{O}\big((H, S, \pi_{-i})\big) \triangleq (H, S)
		\]
		\item $\mathbf{R}: (\mathcal{H}^K\times\mathcal{S}\times\Pi^{K+}_{-i}) \times \mathcal{A}_i \mapsto \mathbb{R}$ is the reward function,
		\[
		\mathbf{R}\big((H, S, \pi_{-i}), a_i\big) \triangleq R^K_{\pi_{-i}}(H, S, a_i)		\]
	\end{enumerate}
An optimal solution, in terms of maximizing infinite-horizon discounted rewards, of such a POMDP is typically obtained as a mapping from all possible histories (or equivalently, from beliefs over states) to potentially randomized actions~\cite{sondik1978optimal, kaelbling1998planning}, and therefore, may not correspond to finite-memory strategies in general.
\end{proofsk}


One can see that the constructed POMDPs in the above theorem belong to a subclass of generic POMDPs, where a state is composed of directly observable variables and other hidden ones. This subclass is specially termed as \textit{Mixed observability MDPs} (MOMPDs)~\cite{ong2010planning, araya2010closer, lee2007makes}.
Existing research has shown that planning algorithms originally developed for POMDPs are significantly faster for those factorized models like MOMDPs in practice.
In fact, our case fits an even more restricted model called \textit{Contextual MDPs} (CMDPs)~\cite{hallak2015contextual, benjamins2023contextualize}, which can be viewed as a special case of MOMDPs where there are no transitions among the hidden state variables.
While CMDPs and MOMDPs are special cases of POMDPs, the complexity/computability results for the former two remain unresolved.
So far, the common conjecture is that neither CMDP nor MOMDP is significantly easier to solve than POMDP, and it is proven that optimally solving infinite-horizon POMDPs is undecidable~\cite{madani2003undecidability}.

This result highly pertains to discussions on type-based methods for single-agent planning in the presence of multiple other agents~\cite{albrecht2015game, albrecht2019convergence, albrecht2016belief, zhu2025single}.
\citeauthor{albrecht2015game}~\cite{albrecht2015game} characterized the general problem from a conceptual standpoint, where each opponent's strategy acts as an oracle that can be queried; however, they left the specific implementation issues unresolved.
As a supplementary, \citeauthor{zhu2025single}~\cite{zhu2025single} offered a spectrum of implementable planners for the stationary base case, where each support strategy within the opponent's mixed strategy is stationary.
Here, Theorem~\ref{thm:memk_br_pomdp} further generalizes this to constant-memory strategies, enabling all the formulations in~\cite{zhu2025single} to be extended to the entire family of constant-memory strategies.



We will also present a reduction in the reversed direction.

\begin{theorem}
	Optimally solving a CMDP can be reduced to computing the best response for an agent $i$ against a profile of mixed zero-memory (i.e., stationary) strategies $(\Pi^{0+}_{-i}, \vec p)$ of its opponents.
\label{thm:cmdp_to_br}
\end{theorem}

\begin{proof}
We prove the above theorem by constructing a reduction for any given CMDP instance that requires an optimal solution, to an SG instance that requires a best response for one of the agents against its opponent's mixed strategy.
Consider a CMDP formally defined as a tuple $\langle \mathcal{C}, \mathcal{S}, \mathcal{A}, f_T, f_R, \gamma \rangle$, where
\begin{enumerate}
	\item $\mathcal{C}$ is a finite set of unobservable contexts, one of which will be selected at the beginning of each episode.
	\item $f_T$ and $f_R$ take a context $c\in \mathcal{C}$ and output a transition function $T^{c}: \mathcal{S} \times \mathcal{A} \mapsto \Delta(\mathcal{S})$ as well as a reward function $R^{c}: \mathcal{S}\times \mathcal{A} \mapsto \mathbb{R}$, respectively. The tuple $\langle \mathcal{S}, \mathcal{A}, T^c, R^c, \gamma \rangle$ then constitutes an MDP.
\end{enumerate}
One can construct an SG with two players $$\langle \{1,2\}, \mathcal{S}, \{\mathcal{A}, \mathcal{A'}\}, T_0, \{R_1, R_2\} \rangle,$$ where agent 1 with action set $\mathcal{A}$ is playing against agent 2 (as a context controller/switcher), who holds a set of stationary strategies $\{\pi_{c}: \mathcal{S} \mapsto \Delta(\mathcal{A}')\}_{c\in \mathcal{C}}$. We need to prove that there exists a $T_0: \mathcal{S} \times \mathcal{A} \times \mathcal{A}' \mapsto \Delta(\mathcal{S})$ and $\{\pi_c\}_{c\in \mathcal{C}}$, such that the following system of equations holds simultaneously,
\begin{equation}
\forall c\in \mathcal{C},\
T^{c}(S' | S, a) = \sum_{a'\in \mathcal{A}'} T_0(S' | S, a, a')\cdot \pi_c(a'|S)
\label{eq:cmdp_decomp} 	
\end{equation}
We omit the similar discussion on $R^c$.
One can see that the above equation operates independently for each $(S,a)$ pair but should hold simultaneously for all $c\in\mathcal{C}$ while fixing a pair of $(S,a)$.
For each $(S, a)$, Equation~(\ref{eq:cmdp_decomp}) can be written in matrix notation as
\[
M_C = M_\Pi \cdot M_T
\]
where $M_C[c,S'] = T^{c}(S' | S, a),  M_\Pi[c, a'] = \pi_c(a'|S), M_T[a', S'] = T_0(S' | S, a, a')$, thus, $M_C \in \mathbb{R}^{\mathcal{C}\times \mathcal{S}}, M_\Pi \in \mathbb{R}^{\mathcal{C}\times \mathcal{A}'}, M_T \in \mathbb{R}^{\mathcal{A}' \times \mathcal{S}}$.
Specifically, the $j$-th row $M_C[j] \in \mathbb{R}^{1\times \mathcal{S}} $ of $M_C$ is a linear combination of all rows in $M_T$, with the linear weights provided by the $j$-th row $M_\Pi[j] \in \mathbb{R}^{1\times \mathcal{A}'}$ of $M_\Pi$.


A natural question arises: how can we find such $M_\Pi$ and $M_T$ of minimum sizes, i.e. with smallest $|\mathcal{A}'|$. This can be further reduced to finding a minimum set of $|\mathcal{S}|$-dimensional vectors whose linear combination can represent all the row vectors in $M_C$.
We now describe a procedure that iteratively construct such $M_\Pi$ and $M_T$, formally given in Algorithm~\ref{alg:gram_schmidt}. The idea is quite clean and elegant: start with the first row of $M_C$ as the first basis vector, project the $j$-th subsequent row onto all previous $(j-1)$ basis vectors, and treat the orthogonal residual as the $j$-th basis vector if it is non-zero. 
One may note that this procedure resembles the \textit{Gram-Schmidt Orthogonalization}~\cite{cheney2009linear}, which can be done in \textit{strongly-polynomial time}. The only difference is, the standard \textit{Gram-Schmidt Orthogonalization} starts with a set of vectors that are already linearly independent, though they may not be orthogonal to each other. In contrast, here we start with a set of vectors that may be linearly dependent, and the goal is to find the minimum set of basis vectors.
\end{proof}

\begin{algorithm}[!ht]
    \caption{Find the minimum $M_\Pi$ and $M_T$}
    \label{alg:gram_schmidt}
    \begin{algorithmic}[1]
        \State \textbf{Input:} $M_C$ 
        \State \textbf{Output:} $M_\Pi$ and $M_T$ of minimum sizes
        \State Initialize: $M_T$ as an empty matrix
        \State $M_T.append\_row(M_C[1])$
        \Comment {index starts from 1}
        \For{$j = 2 \to |C|$}
        	\State $new \gets M_C[j] - \sum_{k=1}^{j-1}\frac{\langle M_T[k],M_C[j] \rangle}{\langle M_T[k],M_T[k] \rangle}M_T[k]$
        	\If{$new \neq \vec 0$}
        		\State $M_T.append\_row(new)$
        	\EndIf
        \EndFor
        \State $M_T \gets normalize\_each\_row(M_T)$
        \State $M_\Pi \gets M_C \cdot transpose(M_T)$
        \State \Return $M_\Pi$, $M_T$
    \end{algorithmic}
\end{algorithm}



Therefore, one can conclude that the theorem below directly follows from Theorem~\ref{thm:memk_br_pomdp} and Theorem~\ref{thm:cmdp_to_br}.
\begin{theorem}
The computational problem of computing the best response to a mixed constant-memory strategy is as hard as that of optimally solving CMDPs.
\label{thm:br=cmdp}
\end{theorem}


Finally, we highlight some connections to the existing literature:
\begin{enumerate}
	\item If in each turn the opponent is allowed to switch to a different support strategy independently of previous actions, which can be reduced to a mixed-strategy-induced behavioral strategy, then how the best response is computed in our work is equivalent to solving a \textit{belief-induced MDP} in~\cite{zhu2025single}.
	\item \citeauthor{wang2019pure}~\cite{wang2019pure} observed that there may not exist a pure one-memory strategy as a best response against a population of one-memory opponents, each potentially adopting a different one-memory strategy (as if in a tournament). Our work provides some formal evidence: the best response in general is not even within constant-memory; instead, it may incorporate infinite memory. 
	\item Best responses to mixed strategies here can be seen as one level of recursion in a bottom-up construction of dynamic programming in I-POMDPs~\cite{gmytrasiewicz2005framework}. Therefore, our work can serve as the missing justification for why solving POMDPs or CMPDs, rather than MDPs, is essential in I-POMDPs.
\end{enumerate}

%
%
%
%



\section{Empirical Study}

The purpose these empirical studies is not to benchmark the algorithms mentioned in this section; rather, it is to present an intuitive illustration of the effects of memory.
The tested domains include two matrix games played sequentially, and one domain borrowed from robotics with raw image inputs.


\subsection{Sequential Matrix Games}

We consider two matrix games that are played in a repeated manner, namely the \textit{Iterated Prisoner's Dilemma} (IPD), and the \textit{Iterated Traveler's Dilemma} (ITD).

\subsubsection{The Iterated Prisoner's Dilemma}

The payoff matrix is shown in the table below.
We also remind the readers of 
a library \cite{knight_2023_7861907} that implements most of the strategies from the well-known Axelrod's IPD tournament.
\textit{Our approach can compute the best response of any constant-memory strategy in this library, whether deterministic or randomized.}
In particular, we would like to highlight a family of strategies, called
\textit{$N$-Tit(s)-for-$M$-Tat(s)} (originally named by~\citeauthor{harper2017reinforcement}~\cite{harper2017reinforcement}), which is a parameterized version of the classic \textit{Tit-for-Tat}.
An agent adopting \textit{$N$-Tit(s)-for-$M$-Tat(s)} will retaliate immediately after it has been \textit{defected} $M$ times, by responding with \textit{defection} in the next $N$ rounds.
Thus, it is a $\max(N,M)$-memory strategy.

\begin{table}[ht]
\begin{tabular}{c|c|c|}
\cline{2-3}
                        & C      & D      \\ \hline
\multicolumn{1}{|c|}{C} & (1, 1) & (-1, 2) \\ \hline
\multicolumn{1}{|c|}{D} & (2, -1) & (0, 0)  \\ \hline
\end{tabular}
\label{tab:pd}
\end{table}
\vspace{-2mm}

\begin{figure}[!ht]
	\centering
	\includegraphics[width=60mm]{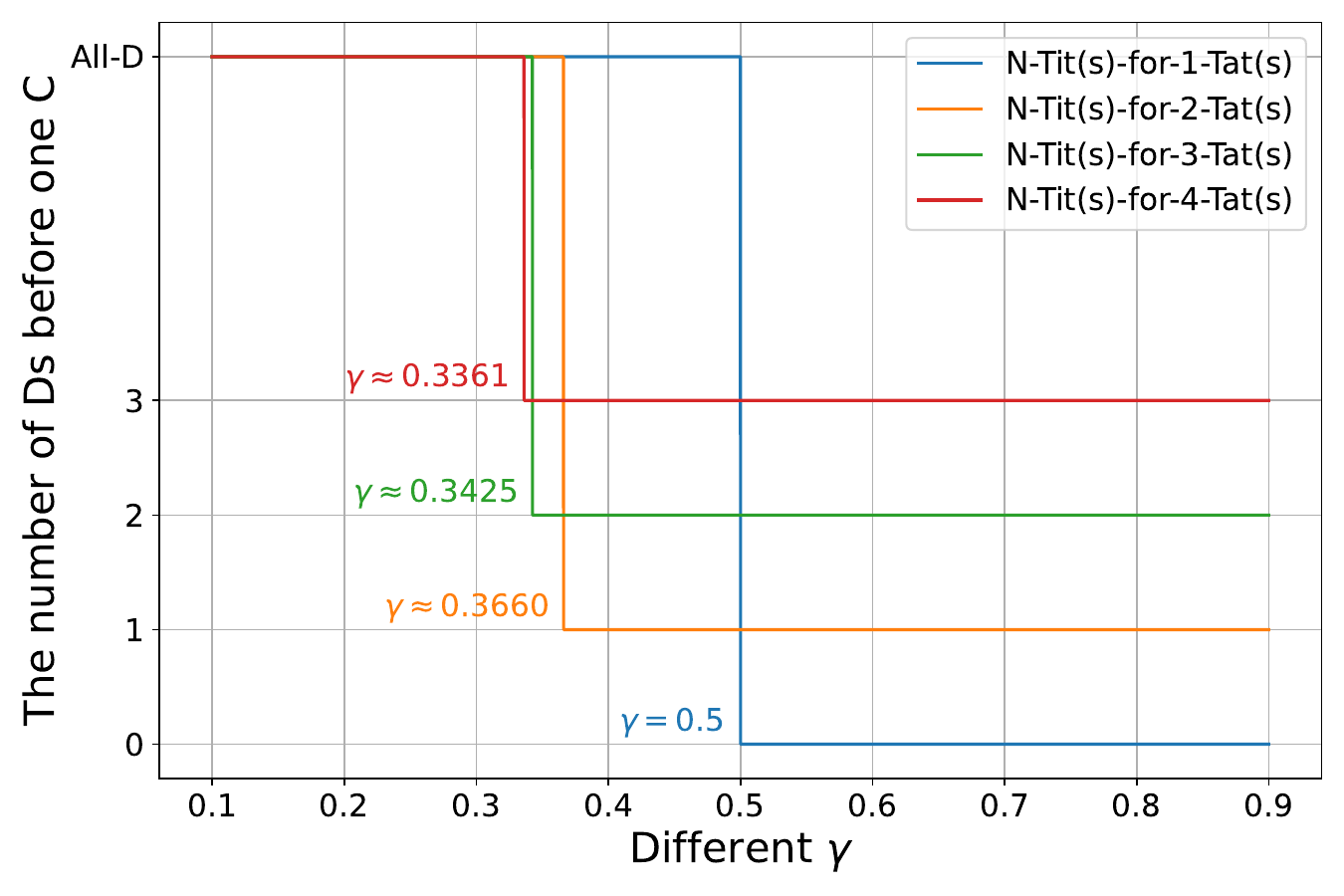}
	\vspace{-3mm}
	\caption{Experimental results of the IPD.}
	\label{fig:ipd_results}
	\Description{Experimental results of the Iterated Prisoner's Dilemma}
\end{figure}

\begin{figure}[!ht]
	\centering
	\includegraphics[width=60mm]{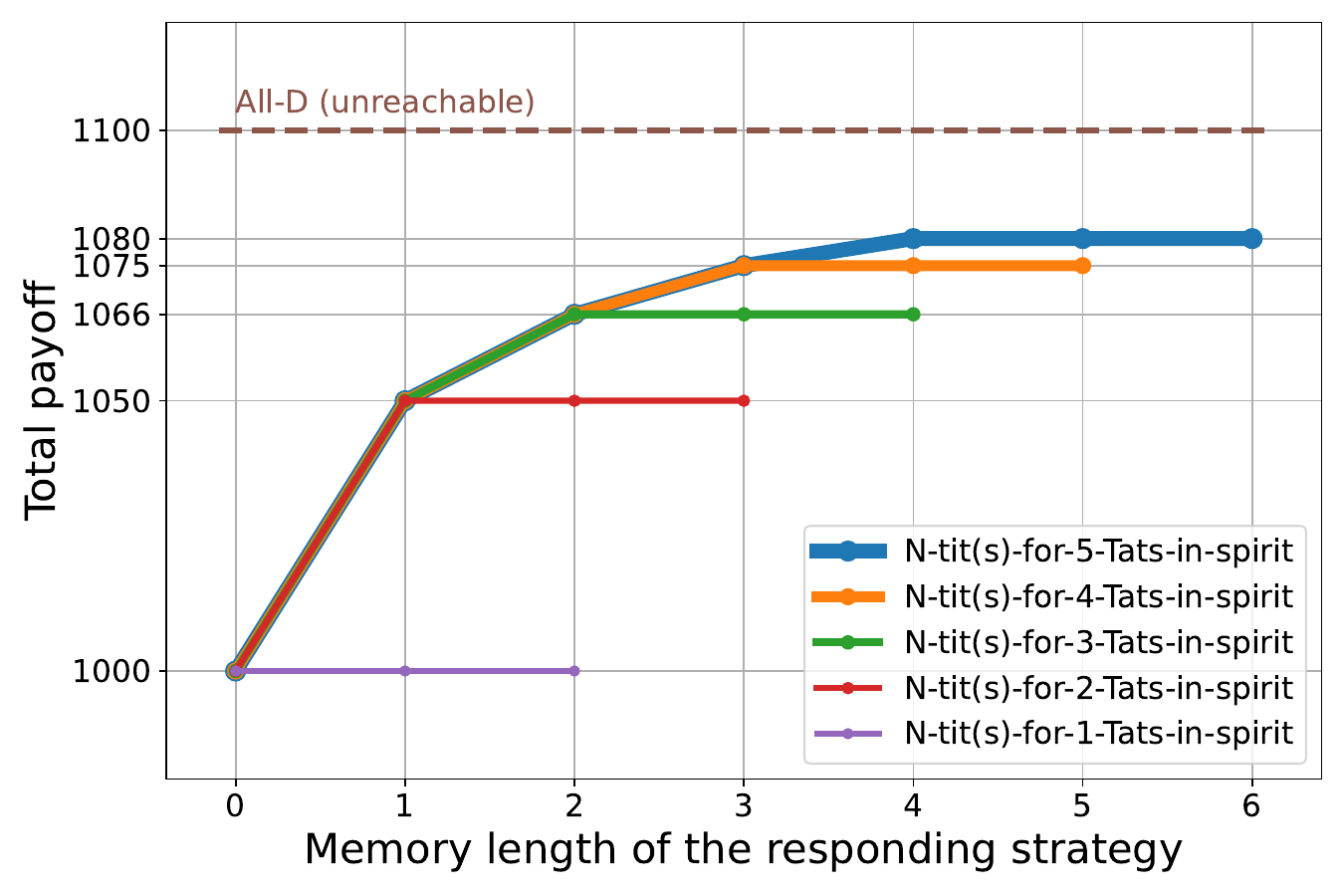}
	\vspace{-3mm}
	\caption{Experimental results of the ITD.
	}
	\label{fig:itd_results}
	\Description{Experimental results of the Iterated Traveler's Dilemma}
\end{figure}

\begin{figure*}[!ht]
	\centering
	\includegraphics[width=56mm]{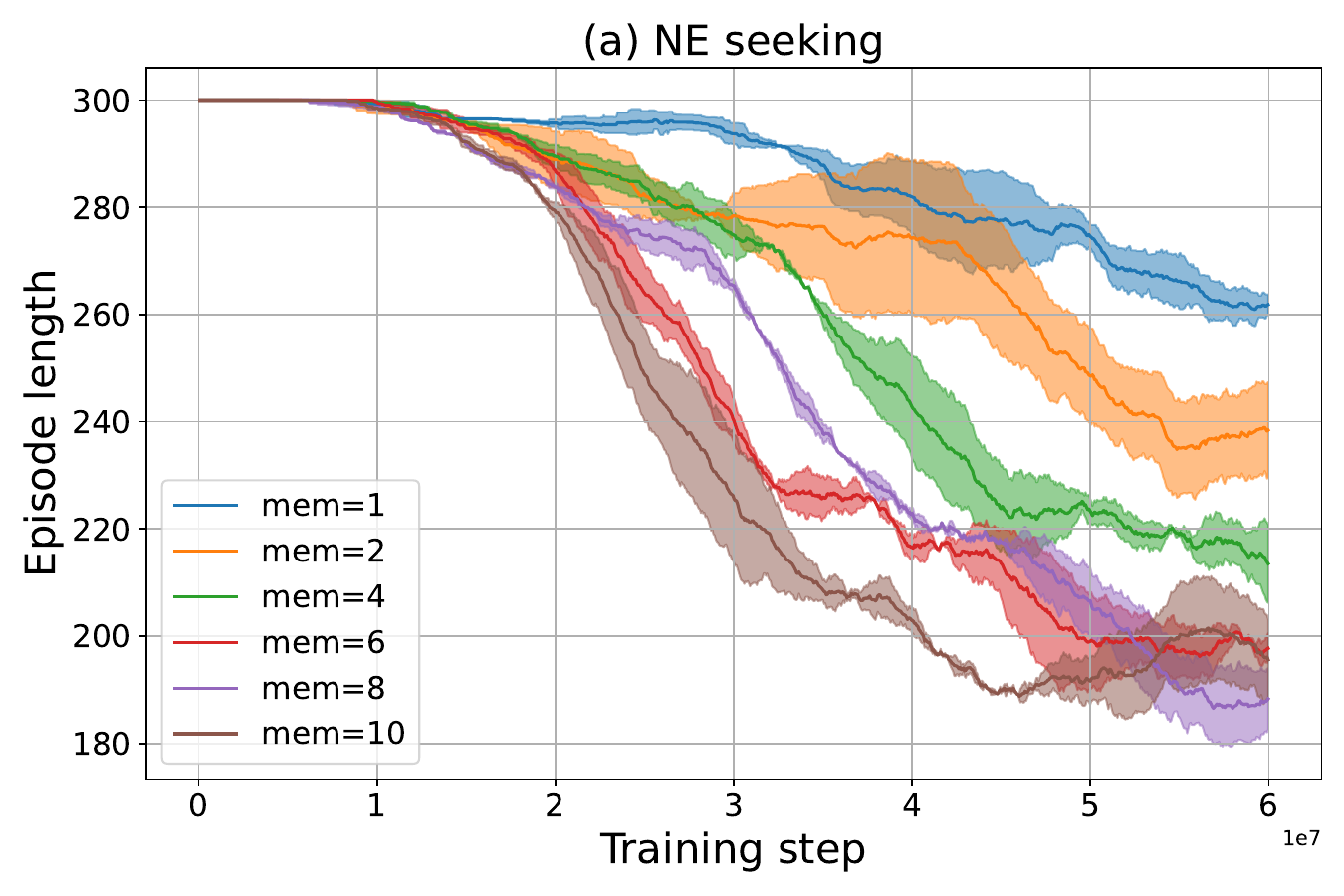}
	\includegraphics[width=56mm]{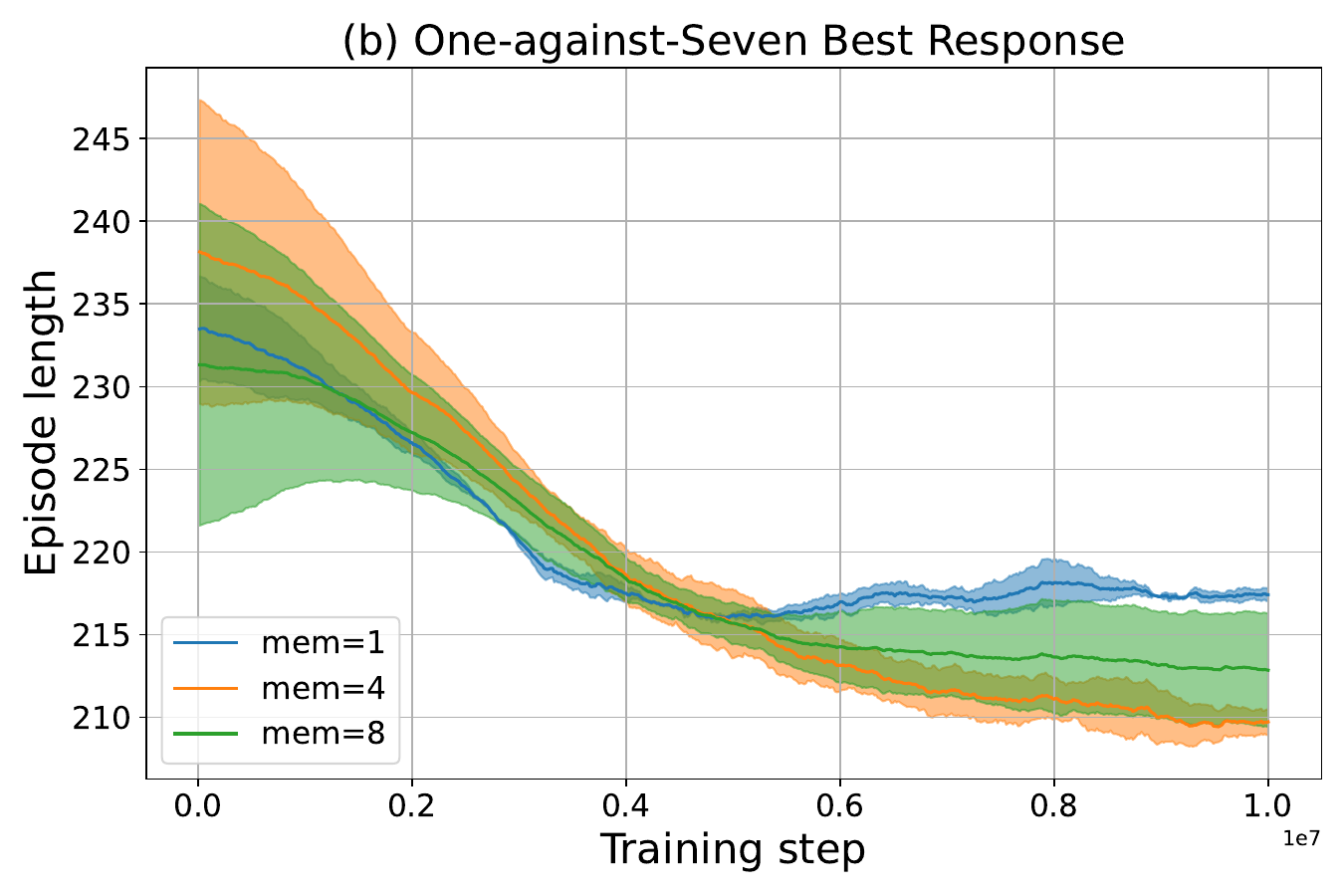}
	\includegraphics[width=56mm]{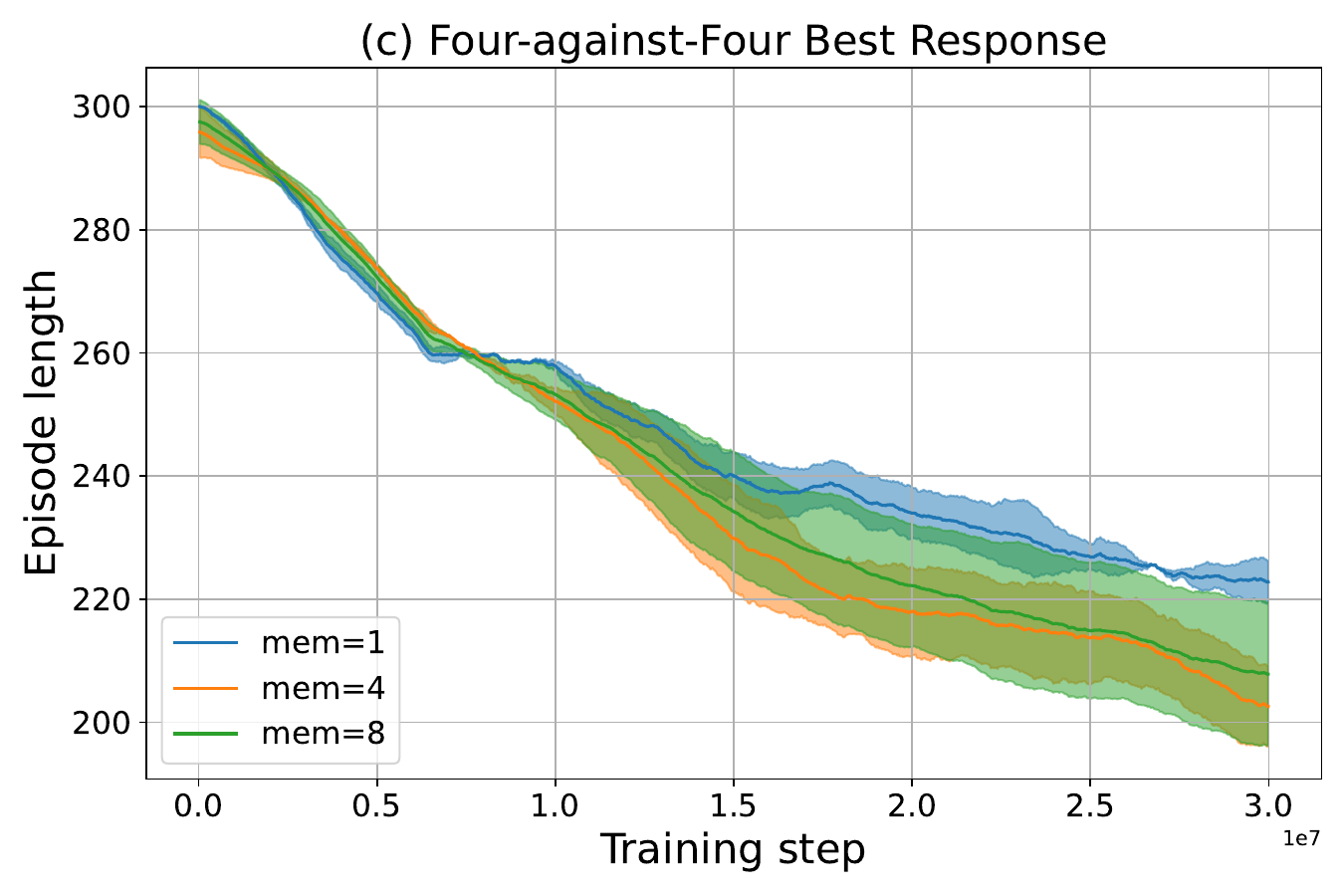}
	\caption{Experimental results of the Pursuit domain. (a) NE seeking; (b) Single agent best responding to the the rest 7 agents; (c) One team of four agents best responding to the other team of four agents.}
	\label{fig:pursuit_results}
	\Description{Experimental results of the Pursuit domain}
\end{figure*}


We compute the best responses for various settings of $(N,M)$ and the discount factor $\gamma$, and illustrate our findings in Figure~\ref{fig:ipd_results}.
We utilize 
MDPtoolbox~\cite{chades2014mdptoolbox}
to compute the exact solutions of the formulated MDPs, where the resulted policies are all deterministic ones.
One can observe clear phase transitions in the best responses.
Reading the figure from right to left, it indicates that
when the discount factor is sufficiently large, the best response to \textit{$N$-Tit(s)-for-$M$-Tat(s)} is to \textit{defect} ($M$-1) times followed by one \textit{cooperation}, and to repeat this pattern periodically, regardless of $N$.
However, when the agent is less patient by placing less value on future reciprocity, it will consider taking one additional round of \textit{defection}, leading to permanent mutual defection from the ($M$+1)-th round onwards.

This finding is summarized in a formal theorem that can be mathematically justified.
We also compute the closed-form solutions for those values of 
$\gamma$ that trigger the phase transition. \textit{Please refer to Appendix~\ref{app:prisoner} for details.}
Please note that while this paper focuses on the discounted-payoff setting, our code also includes the computation of best responses under the average-payoff setting, although a detailed discussion is omitted here.


\subsubsection{The Iterated Traveler's Dilemma}

One may notice that the aforementioned \textit{($M$-1)-D-before-One-C} strategy against \textit{$N$-Tit(s)-for-$M$-Tat(s)} can be implemented using only $(M-1)$ memory instead of $\max(N,M)$ memory. 
Specifically, if there has been no \textit{cooperation} played by itself in the previous ($M$-1) rounds, it should \textit{cooperate} in the current round; otherwise, it should continue \textit{defecting}.
However, we have not addressed the theoretic case of computing the best response 
against a $K$-memory strategy using only $K'$-memory with $K'<K$.
Note that one cannot formulate a 
$K$-memory MDP but compute its optimal policy in the form of 
$K'$-memory using dynamic programming, as this would result in inconsistent policy updates.
To circumvent this issue,
we can run model-free algorithms, e.g., Q-learning which only requires a form of the policy in advance, to see the outcome $K'$-memory strategy.

Therefore, we generalize our previous findings to a broader class of games, called
 the \textit{Iterated Traveler's Dilemma} (ITD)~\cite{dasler2010iterated,tovsic2016learning}, which repeats over the matrix game of traveler's dilemma~\cite{basu1994traveler} with payoffs given as,
\begin{equation}
u_i(a_i, a_{-i}) \triangleq \min(a_i, a_{-i}) + 2\cdot sign(a_{-i} - a_i)
\end{equation}
where each action $a_i$ is an integer variable, also known as a \textit{bid}.
Please note that when $i \in \{1,2\}$ and $a_i \in \{0, 1\}$ for each $i$, it reduces the aforementioned prisoner's dilemma.
In our study, we consider $i \in \{1,2\}$ with a more fine-grained action space $a_i \in \{0, 1, \cdots 10\}$ to render a harder computational problem.
As the investigation of ITD is a relatively new area, we also provide some justification for its importance in Appendix~\ref{app:traveler}.

We implement the \textit{N-Tit(s)-for-M-Tat(s)} in spirit, as there are no longer well-defined notions of \textit{defections} or \textit{cooperations}.
In this context, if an agent finds that its opponent's last bid is smaller than its own last bid, this will be interpreted as a \textit{defection}, while \textit{cooperation} is defined in the opposite manner.
Tabular Q-learning is leveraged to compute the best responses for various memory length against the \textit{N-Tit(s)-for-M-Tat(s)} for different values of $M$.
For a specific value of $M$, we compute the best response using memory lengths ranging from 0 to $(M+1)$.
We illustrate the total (undiscounted) payoff for 100 rounds in Figure~\ref{fig:itd_results}.
It turns out that, when restricted to $M'$-memory, the best response against \textit{N-Tit(s)-for-M-Tat(s)} is exactly \textit{$M'$-D-before-One-C}, given any $M' < M$.
For example, playing against \textit{N-Tit(s)-for-5-Tat(s)}, the best response restricted to only 3-memory will be \textit{Three-D-before-One-C}, resulting in the agent exploiting its opponent for 3/4 of the time, with a total payoff $= 10\times 25 + 11\times 75 = 1075$.

\subsection{The Pursuit Domain}

The aforementioned two games present clear social dilemmas,
such that the technique used in the proof of Theorem~\ref{thm:ne_beh} may only find NEs where both players \textit{defect} from the very beginning, regardless of the memory length utilized (cf. the aforementioned notebook\footnote{Please refer to \texttt{code/kMemNE\_full.ipynb} in the codebase.}).
Therefore, we also conduct some experiments on a more intricate testbed borrowed from the robotics community, namely the \textit{Pursuit} domain~\cite{gupta2017cooperative}.
In this task, 8 pursuer agents attempt to catch 20 random walkers (also called evaders).
Each pursuer agent can only observe a limited local range, and once 4 pursuers simultaneously overlap with the same evader, this evader will be removed from the game. An episode terminates immediately after all the evaders are removed.
Ideally, these 8 agents will devide into two teams for evader hunting.
It is actually a Partially Observable SG (POSG) rather than strictly an SG.
We aim to investigate: 1) whether longer memories will lead to improved NEs; 2) how well one agent can respond to the 7 others; 3) how well one team (four agents) can respond to the other team (the rest four).

The main results are presented in Figure~\ref{fig:pursuit_results}
focusing solely on the results obtained with DQN~\cite{mnih2013playing}, as it significantly outperforms other algorithms in this task.
\textit{Additional benchmarking results using other algorithms, e.g., A2C~\cite{mnih2016asynchronous} and PPO~\cite{schulman2017proximal}, along with relevant detailed settings, are provided in Appendix~\ref{app:pursuit} for the reader's reference.}
To increase the memory length, we simply stack the historical observations and actions.
For multi-agent learning in search for NEs, we equip each agent with an identical network and train them to learn independently.
As shown in Figure~\ref{fig:pursuit_results}(a), utilizing longer memory indeed helps the pursuers catch the evaders faster, indicating a better NE.
We extract the eventual strategy trained with 8-memory as it appears to be the best.
In the experiments depicted in Figure~\ref{fig:pursuit_results}(b), 7 agents are equipped with this pretrained 8-memory strategy, while the remaining agent learns from scratch to find the best response. 
In the experiments shown in Figure~\ref{fig:pursuit_results}(b), one team of 4 agents are equipped with this pretrained 8-memory strategy, leaving the remaining team of the rest 4 agents learning from scratch to find a best (``team'') response.
As a result, using memories of length 4 and 8 is clearly better than using memories of length one.
However, 8-memory responses are not significantly distinguishable from 4-memory responses, which may be attributed to the fact that 4-memory strategies are already sufficient to serve as the best response, or possibly due to some representation error introduced by the deep neural network.

It is also interesting to note that, the improvement, which is reflected by the episode length, made by one agent with the other 7 agents fixed (as shown in Figure~\ref{fig:pursuit_results}(b)) is clearly less then that made by a team of agents with the other team fixed (as shown in Figure~\ref{fig:pursuit_results}(c)). As we examined, in the former case, there is typically one of the 7 fixed agents who occasionally collaborates with three of them and at other times with the remaining three, creating the pseudo-effect of two four-agent teams.
This observation may potentially explains why adding one more learning agent only leads to only incremental improvement.


\section{Conclusion}

In this work, we develop a theoretic framework to study constant-memory strategies.
The notion of best responses and equilibria are well-established.
In particular, we highlight that responding to mixed constant-memory strategies may be computationally hard, possibly even not computable.
These results can be seen as an extension of both~\cite{chen2017k,wang2019pure} (from repeated games to stochastic games)
and~\cite{zhu2025single} (from stationary strategies to K-memory ones).
We also conduct experiments on well-known social dilemmas as well as a multi-robot domain to verify those theoretic insights.



\bibliographystyle{ACM-Reference-Format} 
\bibliography{sample}


\begin{thebibliography}{65}


\ifx \showCODEN    \undefined \def \showCODEN     #1{\unskip}     \fi
\ifx \showDOI      \undefined \def \showDOI       #1{#1}\fi
\ifx \showISBNx    \undefined \def \showISBNx     #1{\unskip}     \fi
\ifx \showISBNxiii \undefined \def \showISBNxiii  #1{\unskip}     \fi
\ifx \showISSN     \undefined \def \showISSN      #1{\unskip}     \fi
\ifx \showLCCN     \undefined \def \showLCCN      #1{\unskip}     \fi
\ifx \shownote     \undefined \def \shownote      #1{#1}          \fi
\ifx \showarticletitle \undefined \def \showarticletitle #1{#1}   \fi
\ifx \showURL      \undefined \def \showURL       {\relax}        \fi
\providecommand\bibfield[2]{#2}
\providecommand\bibinfo[2]{#2}
\providecommand\natexlab[1]{#1}
\providecommand\showeprint[2][]{arXiv:#2}

\bibitem[\protect\citeauthoryear{Albrecht, Crandall, and Ramamoorthy}{Albrecht
  et~al\mbox{.}}{2016}]%
        {albrecht2016belief}
\bibfield{author}{\bibinfo{person}{Stefano~V Albrecht},
  \bibinfo{person}{Jacob~W Crandall}, {and} \bibinfo{person}{Subramanian
  Ramamoorthy}.} \bibinfo{year}{2016}\natexlab{}.
\newblock \showarticletitle{Belief and truth in hypothesised behaviours}.
\newblock \bibinfo{journal}{\emph{Artificial Intelligence}}
  \bibinfo{volume}{235} (\bibinfo{year}{2016}), \bibinfo{pages}{63--94}.
\newblock


\bibitem[\protect\citeauthoryear{Albrecht and Ramamoorthy}{Albrecht and
  Ramamoorthy}{2015}]%
        {albrecht2015game}
\bibfield{author}{\bibinfo{person}{Stefano~V Albrecht} {and}
  \bibinfo{person}{Subramanian Ramamoorthy}.} \bibinfo{year}{2015}\natexlab{}.
\newblock \showarticletitle{A game-theoretic model and best-response learning
  method for ad hoc coordination in multiagent systems}.
\newblock \bibinfo{journal}{\emph{arXiv preprint arXiv:1506.01170}}
  (\bibinfo{year}{2015}).
\newblock


\bibitem[\protect\citeauthoryear{Albrecht and Ramamoorthy}{Albrecht and
  Ramamoorthy}{2019}]%
        {albrecht2019convergence}
\bibfield{author}{\bibinfo{person}{Stefano~V Albrecht} {and}
  \bibinfo{person}{Subramanian Ramamoorthy}.} \bibinfo{year}{2019}\natexlab{}.
\newblock \showarticletitle{On convergence and optimality of best-response
  learning with policy types in multiagent systems}.
\newblock \bibinfo{journal}{\emph{arXiv preprint arXiv:1907.06995}}
  (\bibinfo{year}{2019}).
\newblock


\bibitem[\protect\citeauthoryear{Albrecht and Stone}{Albrecht and
  Stone}{2018}]%
        {albrecht2018autonomous}
\bibfield{author}{\bibinfo{person}{Stefano~V Albrecht} {and}
  \bibinfo{person}{Peter Stone}.} \bibinfo{year}{2018}\natexlab{}.
\newblock \showarticletitle{Autonomous agents modelling other agents: A
  comprehensive survey and open problems}.
\newblock \bibinfo{journal}{\emph{Artificial Intelligence}}
  \bibinfo{volume}{258} (\bibinfo{year}{2018}), \bibinfo{pages}{66--95}.
\newblock


\bibitem[\protect\citeauthoryear{Araya-L{\'o}pez, Thomas, Buffet, and
  Charpillet}{Araya-L{\'o}pez et~al\mbox{.}}{2010}]%
        {araya2010closer}
\bibfield{author}{\bibinfo{person}{Mauricio Araya-L{\'o}pez},
  \bibinfo{person}{Vincent Thomas}, \bibinfo{person}{Olivier Buffet}, {and}
  \bibinfo{person}{Fran{\c{c}}ois Charpillet}.}
  \bibinfo{year}{2010}\natexlab{}.
\newblock \showarticletitle{A closer look at MOMDPs}. In
  \bibinfo{booktitle}{\emph{2010 22nd IEEE International Conference on Tools
  with Artificial Intelligence}}, Vol.~\bibinfo{volume}{2}. IEEE,
  \bibinfo{pages}{197--204}.
\newblock


\bibitem[\protect\citeauthoryear{Baarslag}{Baarslag}{2024}]%
        {baarslag2024multi}
\bibfield{author}{\bibinfo{person}{Tim Baarslag}.}
  \bibinfo{year}{2024}\natexlab{}.
\newblock \showarticletitle{Multi-deal negotiation}. In
  \bibinfo{booktitle}{\emph{Proceedings of the 23rd International Conference on
  Autonomous Agents and Multiagent Systems}}. \bibinfo{pages}{2668--2673}.
\newblock


\bibitem[\protect\citeauthoryear{Baarslag, Hendrikx, Hindriks, and
  Jonker}{Baarslag et~al\mbox{.}}{2016}]%
        {baarslag2016learning}
\bibfield{author}{\bibinfo{person}{Tim Baarslag}, \bibinfo{person}{Mark~JC
  Hendrikx}, \bibinfo{person}{Koen~V Hindriks}, {and}
  \bibinfo{person}{Catholijn~M Jonker}.} \bibinfo{year}{2016}\natexlab{}.
\newblock \showarticletitle{Learning about the opponent in automated bilateral
  negotiation: a comprehensive survey of opponent modeling techniques}.
\newblock \bibinfo{journal}{\emph{Autonomous Agents and Multi-Agent Systems}}
  \bibinfo{volume}{30} (\bibinfo{year}{2016}), \bibinfo{pages}{849--898}.
\newblock


\bibitem[\protect\citeauthoryear{Balseiro, Besbes, and Weintraub}{Balseiro
  et~al\mbox{.}}{2015}]%
        {balseiro2015repeated}
\bibfield{author}{\bibinfo{person}{Santiago~R Balseiro}, \bibinfo{person}{Omar
  Besbes}, {and} \bibinfo{person}{Gabriel~Y Weintraub}.}
  \bibinfo{year}{2015}\natexlab{}.
\newblock \showarticletitle{Repeated auctions with budgets in ad exchanges:
  Approximations and design}.
\newblock \bibinfo{journal}{\emph{Management Science}} \bibinfo{volume}{61},
  \bibinfo{number}{4} (\bibinfo{year}{2015}), \bibinfo{pages}{864--884}.
\newblock


\bibitem[\protect\citeauthoryear{Basu}{Basu}{1994}]%
        {basu1994traveler}
\bibfield{author}{\bibinfo{person}{Kaushik Basu}.}
  \bibinfo{year}{1994}\natexlab{}.
\newblock \showarticletitle{The traveler's dilemma: Paradoxes of rationality in
  game theory}.
\newblock \bibinfo{journal}{\emph{The American Economic Review}}
  \bibinfo{volume}{84}, \bibinfo{number}{2} (\bibinfo{year}{1994}),
  \bibinfo{pages}{391--395}.
\newblock


\bibitem[\protect\citeauthoryear{Ben-Porath}{Ben-Porath}{1990}]%
        {ben1990complexity}
\bibfield{author}{\bibinfo{person}{Elchanan Ben-Porath}.}
  \bibinfo{year}{1990}\natexlab{}.
\newblock \showarticletitle{The complexity of computing a best response
  automaton in repeated games with mixed strategies}.
\newblock \bibinfo{journal}{\emph{Games and Economic Behavior}}
  \bibinfo{volume}{2}, \bibinfo{number}{1} (\bibinfo{year}{1990}),
  \bibinfo{pages}{1--12}.
\newblock


\bibitem[\protect\citeauthoryear{Benjamins, Eimer, Schubert, Mohan, D{\"o}hler,
  Biedenkapp, Rosenhahn, Hutter, and Lindauer}{Benjamins et~al\mbox{.}}{2023}]%
        {benjamins2023contextualize}
\bibfield{author}{\bibinfo{person}{Carolin Benjamins}, \bibinfo{person}{Theresa
  Eimer}, \bibinfo{person}{Frederik Schubert}, \bibinfo{person}{Aditya Mohan},
  \bibinfo{person}{Sebastian D{\"o}hler}, \bibinfo{person}{Andr{\'e}
  Biedenkapp}, \bibinfo{person}{Bodo Rosenhahn}, \bibinfo{person}{Frank
  Hutter}, {and} \bibinfo{person}{Marius Lindauer}.}
  \bibinfo{year}{2023}\natexlab{}.
\newblock \showarticletitle{Contextualize Me {\textendash} The Case for Context
  in Reinforcement Learning}.
\newblock \bibinfo{journal}{\emph{Transactions on Machine Learning Research}}
  (\bibinfo{year}{2023}).
\newblock
\showISSN{2835-8856}
\urldef\tempurl%
\url{https://openreview.net/forum?id=Y42xVBQusn}
\showURL{%
\tempurl}


\bibitem[\protect\citeauthoryear{Brouwer}{Brouwer}{1911}]%
        {brouwer1911abbildung}
\bibfield{author}{\bibinfo{person}{Luitzen Egbertus~Jan Brouwer}.}
  \bibinfo{year}{1911}\natexlab{}.
\newblock \showarticletitle{{\"U}ber abbildung von mannigfaltigkeiten}.
\newblock \bibinfo{journal}{\emph{Mathematische annalen}} \bibinfo{volume}{71},
  \bibinfo{number}{1} (\bibinfo{year}{1911}), \bibinfo{pages}{97--115}.
\newblock


\bibitem[\protect\citeauthoryear{Carmel and Markovitch}{Carmel and
  Markovitch}{1998}]%
        {carmel1998explore}
\bibfield{author}{\bibinfo{person}{David Carmel} {and} \bibinfo{person}{Shaul
  Markovitch}.} \bibinfo{year}{1998}\natexlab{}.
\newblock \showarticletitle{How to explore your opponent's strategy (almost)
  optimally}. In \bibinfo{booktitle}{\emph{Proceedings International Conference
  on Multi Agent Systems (Cat. No. 98EX160)}}. IEEE, \bibinfo{pages}{64--71}.
\newblock


\bibitem[\protect\citeauthoryear{Chad{\`e}s, Chapron, Cros, Garcia, and
  Sabbadin}{Chad{\`e}s et~al\mbox{.}}{2014}]%
        {chades2014mdptoolbox}
\bibfield{author}{\bibinfo{person}{Iadine Chad{\`e}s},
  \bibinfo{person}{Guillaume Chapron}, \bibinfo{person}{Marie-Jos{\'e}e Cros},
  \bibinfo{person}{Fr{\'e}d{\'e}rick Garcia}, {and} \bibinfo{person}{R{\'e}gis
  Sabbadin}.} \bibinfo{year}{2014}\natexlab{}.
\newblock \showarticletitle{MDPtoolbox: a multi-platform toolbox to solve
  stochastic dynamic programming problems}.
\newblock \bibinfo{journal}{\emph{Ecography}} \bibinfo{volume}{37},
  \bibinfo{number}{9} (\bibinfo{year}{2014}), \bibinfo{pages}{916--920}.
\newblock


\bibitem[\protect\citeauthoryear{Chen, Lin, Tang, Wang, Wang, and Wang}{Chen
  et~al\mbox{.}}{2017a}]%
        {chen2017k}
\bibfield{author}{\bibinfo{person}{Lijie Chen}, \bibinfo{person}{Fangzhen Lin},
  \bibinfo{person}{Pingzhong Tang}, \bibinfo{person}{Kangning Wang},
  \bibinfo{person}{Ruosong Wang}, {and} \bibinfo{person}{Shiheng Wang}.}
  \bibinfo{year}{2017}\natexlab{a}.
\newblock \showarticletitle{K-memory strategies in repeated games}. In
  \bibinfo{booktitle}{\emph{Proceedings of the 16th Conference on Autonomous
  Agents and MultiAgent Systems}}. \bibinfo{pages}{1493--1498}.
\newblock


\bibitem[\protect\citeauthoryear{Chen, Tang, and Wang}{Chen
  et~al\mbox{.}}{2017b}]%
        {chen2017bounded}
\bibfield{author}{\bibinfo{person}{Lijie Chen}, \bibinfo{person}{Pingzhong
  Tang}, {and} \bibinfo{person}{Ruosong Wang}.}
  \bibinfo{year}{2017}\natexlab{b}.
\newblock \showarticletitle{Bounded rationality of restricted turing machines}.
  In \bibinfo{booktitle}{\emph{Proceedings of the AAAI Conference on Artificial
  Intelligence}}, Vol.~\bibinfo{volume}{31}.
\newblock


\bibitem[\protect\citeauthoryear{Cheney and Kincaid}{Cheney and
  Kincaid}{2009}]%
        {cheney2009linear}
\bibfield{author}{\bibinfo{person}{Elliott~Ward Cheney} {and}
  \bibinfo{person}{David~Ronald Kincaid}.} \bibinfo{year}{2009}\natexlab{}.
\newblock \bibinfo{booktitle}{\emph{Linear Algebra: Theory and Applications}}.
\newblock \bibinfo{publisher}{Jones \& Bartlett Learning}.
\newblock


\bibitem[\protect\citeauthoryear{Cho, van Merrienboer, Gulcehre, Bougares,
  Schwenk, and Bengio}{Cho et~al\mbox{.}}{2014}]%
        {cho2014learning}
\bibfield{author}{\bibinfo{person}{Kyunghyun Cho}, \bibinfo{person}{B van
  Merrienboer}, \bibinfo{person}{Caglar Gulcehre}, \bibinfo{person}{F
  Bougares}, \bibinfo{person}{H Schwenk}, {and} \bibinfo{person}{Yoshua
  Bengio}.} \bibinfo{year}{2014}\natexlab{}.
\newblock \showarticletitle{Learning phrase representations using RNN
  encoder-decoder for statistical machine translation}. In
  \bibinfo{booktitle}{\emph{Conference on Empirical Methods in Natural Language
  Processing (EMNLP 2014)}}.
\newblock


\bibitem[\protect\citeauthoryear{Dasler and Tosic}{Dasler and Tosic}{2010}]%
        {dasler2010iterated}
\bibfield{author}{\bibinfo{person}{P Dasler} {and} \bibinfo{person}{P Tosic}.}
  \bibinfo{year}{2010}\natexlab{}.
\newblock \showarticletitle{The iterated traveler's dilemma: Finding good
  strategies in games with ``bad'' structure: Preliminary results and
  analysis}. In \bibinfo{booktitle}{\emph{Proc of the 8th Euro. Workshop on
  Multi-Agent Systems, EUMAS}}, Vol.~\bibinfo{volume}{10}.
\newblock


\bibitem[\protect\citeauthoryear{De~Weerd, Verbrugge, and Verheij}{De~Weerd
  et~al\mbox{.}}{2017}]%
        {de2017negotiating}
\bibfield{author}{\bibinfo{person}{Harmen De~Weerd}, \bibinfo{person}{Rineke
  Verbrugge}, {and} \bibinfo{person}{Bart Verheij}.}
  \bibinfo{year}{2017}\natexlab{}.
\newblock \showarticletitle{Negotiating with other minds: the role of recursive
  theory of mind in negotiation with incomplete information}.
\newblock \bibinfo{journal}{\emph{Autonomous Agents and Multi-Agent Systems}}
  \bibinfo{volume}{31} (\bibinfo{year}{2017}), \bibinfo{pages}{250--287}.
\newblock


\bibitem[\protect\citeauthoryear{Fink}{Fink}{1964}]%
        {fink1964equilibrium}
\bibfield{author}{\bibinfo{person}{Arlington~M Fink}.}
  \bibinfo{year}{1964}\natexlab{}.
\newblock \showarticletitle{Equilibrium in a stochastic $ n $-person game}.
\newblock \bibinfo{journal}{\emph{Journal of science of the hiroshima
  university, series ai (mathematics)}} \bibinfo{volume}{28},
  \bibinfo{number}{1} (\bibinfo{year}{1964}), \bibinfo{pages}{89--93}.
\newblock


\bibitem[\protect\citeauthoryear{Foerster, Farquhar, Afouras, Nardelli, and
  Whiteson}{Foerster et~al\mbox{.}}{2018}]%
        {foerster2018counterfactual}
\bibfield{author}{\bibinfo{person}{Jakob Foerster}, \bibinfo{person}{Gregory
  Farquhar}, \bibinfo{person}{Triantafyllos Afouras}, \bibinfo{person}{Nantas
  Nardelli}, {and} \bibinfo{person}{Shimon Whiteson}.}
  \bibinfo{year}{2018}\natexlab{}.
\newblock \showarticletitle{Counterfactual multi-agent policy gradients}. In
  \bibinfo{booktitle}{\emph{Proceedings of the AAAI conference on artificial
  intelligence}}, Vol.~\bibinfo{volume}{32}.
\newblock


\bibitem[\protect\citeauthoryear{Gmytrasiewicz and Doshi}{Gmytrasiewicz and
  Doshi}{2005}]%
        {gmytrasiewicz2005framework}
\bibfield{author}{\bibinfo{person}{Piotr~J Gmytrasiewicz} {and}
  \bibinfo{person}{Prashant Doshi}.} \bibinfo{year}{2005}\natexlab{}.
\newblock \showarticletitle{A framework for sequential planning in multi-agent
  settings}.
\newblock \bibinfo{journal}{\emph{Journal of Artificial Intelligence Research}}
   \bibinfo{volume}{24} (\bibinfo{year}{2005}), \bibinfo{pages}{49--79}.
\newblock


\bibitem[\protect\citeauthoryear{Guo, Huo, Zhang, Wang, Yu, Xu, Zheng, and
  Zhang}{Guo et~al\mbox{.}}{2024}]%
        {guo2024generative}
\bibfield{author}{\bibinfo{person}{Jiayan Guo}, \bibinfo{person}{Yusen Huo},
  \bibinfo{person}{Zhilin Zhang}, \bibinfo{person}{Tianyu Wang},
  \bibinfo{person}{Chuan Yu}, \bibinfo{person}{Jian Xu}, \bibinfo{person}{Bo
  Zheng}, {and} \bibinfo{person}{Yan Zhang}.} \bibinfo{year}{2024}\natexlab{}.
\newblock \showarticletitle{Generative Auto-bidding via Conditional Diffusion
  Modeling}. In \bibinfo{booktitle}{\emph{Proceedings of the 30th ACM SIGKDD
  Conference on Knowledge Discovery and Data Mining}}.
  \bibinfo{pages}{5038--5049}.
\newblock


\bibitem[\protect\citeauthoryear{Guo, Hu, Xu, and Zhang}{Guo
  et~al\mbox{.}}{2019}]%
        {guo2019learning}
\bibfield{author}{\bibinfo{person}{Xin Guo}, \bibinfo{person}{Anran Hu},
  \bibinfo{person}{Renyuan Xu}, {and} \bibinfo{person}{Junzi Zhang}.}
  \bibinfo{year}{2019}\natexlab{}.
\newblock \showarticletitle{Learning mean-field games}.
\newblock \bibinfo{journal}{\emph{Advances in neural information processing
  systems}}  \bibinfo{volume}{32} (\bibinfo{year}{2019}).
\newblock


\bibitem[\protect\citeauthoryear{Gupta, Egorov, and Kochenderfer}{Gupta
  et~al\mbox{.}}{2017}]%
        {gupta2017cooperative}
\bibfield{author}{\bibinfo{person}{Jayesh~K Gupta}, \bibinfo{person}{Maxim
  Egorov}, {and} \bibinfo{person}{Mykel Kochenderfer}.}
  \bibinfo{year}{2017}\natexlab{}.
\newblock \showarticletitle{Cooperative multi-agent control using deep
  reinforcement learning}. In \bibinfo{booktitle}{\emph{International
  Conference on Autonomous Agents and Multiagent Systems}}. Springer,
  \bibinfo{pages}{66--83}.
\newblock


\bibitem[\protect\citeauthoryear{Hallak, Di~Castro, and Mannor}{Hallak
  et~al\mbox{.}}{2015}]%
        {hallak2015contextual}
\bibfield{author}{\bibinfo{person}{Assaf Hallak}, \bibinfo{person}{Dotan
  Di~Castro}, {and} \bibinfo{person}{Shie Mannor}.}
  \bibinfo{year}{2015}\natexlab{}.
\newblock \showarticletitle{Contextual markov decision processes}.
\newblock \bibinfo{journal}{\emph{arXiv preprint arXiv:1502.02259}}
  (\bibinfo{year}{2015}).
\newblock


\bibitem[\protect\citeauthoryear{Harper, Knight, Jones, Koutsovoulos, Glynatsi,
  and Campbell}{Harper et~al\mbox{.}}{2017}]%
        {harper2017reinforcement}
\bibfield{author}{\bibinfo{person}{Marc Harper}, \bibinfo{person}{Vincent
  Knight}, \bibinfo{person}{Martin Jones}, \bibinfo{person}{Georgios
  Koutsovoulos}, \bibinfo{person}{Nikoleta~E Glynatsi}, {and}
  \bibinfo{person}{Owen Campbell}.} \bibinfo{year}{2017}\natexlab{}.
\newblock \showarticletitle{Reinforcement learning produces dominant strategies
  for the iterated prisoner's dilemma}.
\newblock \bibinfo{journal}{\emph{PloS one}} \bibinfo{volume}{12},
  \bibinfo{number}{12} (\bibinfo{year}{2017}), \bibinfo{pages}{e0188046}.
\newblock


\bibitem[\protect\citeauthoryear{Hochreiter and Schmidhuber}{Hochreiter and
  Schmidhuber}{1997}]%
        {hochreiter1997long}
\bibfield{author}{\bibinfo{person}{Sepp Hochreiter} {and}
  \bibinfo{person}{J{\"u}rgen Schmidhuber}.} \bibinfo{year}{1997}\natexlab{}.
\newblock \showarticletitle{Long short-term memory}.
\newblock \bibinfo{journal}{\emph{Neural computation}} \bibinfo{volume}{9},
  \bibinfo{number}{8} (\bibinfo{year}{1997}), \bibinfo{pages}{1735--1780}.
\newblock


\bibitem[\protect\citeauthoryear{Iyer, Johari, and Sundararajan}{Iyer
  et~al\mbox{.}}{2014}]%
        {iyer2014mean}
\bibfield{author}{\bibinfo{person}{Krishnamurthy Iyer}, \bibinfo{person}{Ramesh
  Johari}, {and} \bibinfo{person}{Mukund Sundararajan}.}
  \bibinfo{year}{2014}\natexlab{}.
\newblock \showarticletitle{Mean field equilibria of dynamic auctions with
  learning}.
\newblock \bibinfo{journal}{\emph{Management Science}} \bibinfo{volume}{60},
  \bibinfo{number}{12} (\bibinfo{year}{2014}), \bibinfo{pages}{2949--2970}.
\newblock


\bibitem[\protect\citeauthoryear{Kaelbling, Littman, and Cassandra}{Kaelbling
  et~al\mbox{.}}{1998}]%
        {kaelbling1998planning}
\bibfield{author}{\bibinfo{person}{Leslie~Pack Kaelbling},
  \bibinfo{person}{Michael~L Littman}, {and} \bibinfo{person}{Anthony~R
  Cassandra}.} \bibinfo{year}{1998}\natexlab{}.
\newblock \showarticletitle{Planning and acting in partially observable
  stochastic domains}.
\newblock \bibinfo{journal}{\emph{Artificial intelligence}}
  \bibinfo{volume}{101}, \bibinfo{number}{1-2} (\bibinfo{year}{1998}),
  \bibinfo{pages}{99--134}.
\newblock


\bibitem[\protect\citeauthoryear{Knight, Campbell, Marc, Gaffney, Shaw,
  Reddy~Janga, Glynatsi, Campbell, Langner, Singh, Rymer, Campbell, Young,
  Hakem, Palmer, Glass, Mancia, Argenson, Martin, Kjurgielajtis, Murase,
  Parvatikar, Beck, Davidson-Pilon, Zoulias, Pohl, Slavin, Standen, Kratz, and
  Areeb}{Knight et~al\mbox{.}}{2023}]%
        {knight_2023_7861907}
\bibfield{author}{\bibinfo{person}{Vince Knight}, \bibinfo{person}{Owen
  Campbell}, \bibinfo{person}{Marc}, \bibinfo{person}{T.J. Gaffney},
  \bibinfo{person}{Eric Shaw}, \bibinfo{person}{VSN Reddy~Janga},
  \bibinfo{person}{Nikoleta Glynatsi}, \bibinfo{person}{James Campbell},
  \bibinfo{person}{Karol~M. Langner}, \bibinfo{person}{Sourav Singh},
  \bibinfo{person}{Julie Rymer}, \bibinfo{person}{Thomas Campbell},
  \bibinfo{person}{Jason Young}, \bibinfo{person}{M Hakem},
  \bibinfo{person}{Geraint Palmer}, \bibinfo{person}{Kristian Glass},
  \bibinfo{person}{Daniel Mancia}, \bibinfo{person}{Edouard Argenson},
  \bibinfo{person}{Jones Martin}, \bibinfo{person}{Kjurgielajtis},
  \bibinfo{person}{Yohsuke Murase}, \bibinfo{person}{Sudarshan Parvatikar},
  \bibinfo{person}{Melanie Beck}, \bibinfo{person}{Cameron Davidson-Pilon},
  \bibinfo{person}{Marios Zoulias}, \bibinfo{person}{Adam Pohl},
  \bibinfo{person}{Paul Slavin}, \bibinfo{person}{Timothy Standen},
  \bibinfo{person}{Aaron Kratz}, {and} \bibinfo{person}{Ahmed Areeb}.}
  \bibinfo{year}{2023}\natexlab{}.
\newblock \bibinfo{booktitle}{\emph{Axelrod-Python/Axelrod: v4.12.0}}.
\newblock
\urldef\tempurl%
\url{https://doi.org/10.5281/zenodo.7861907}
\showDOI{\tempurl}


\bibitem[\protect\citeauthoryear{Knoblauch}{Knoblauch}{1994}]%
        {knoblauch1994computable}
\bibfield{author}{\bibinfo{person}{Vicki Knoblauch}.}
  \bibinfo{year}{1994}\natexlab{}.
\newblock \showarticletitle{Computable strategies for repeated prisoner' s
  dilemma}.
\newblock \bibinfo{journal}{\emph{Games and Economic Behavior}}
  \bibinfo{volume}{7}, \bibinfo{number}{3} (\bibinfo{year}{1994}),
  \bibinfo{pages}{381--389}.
\newblock


\bibitem[\protect\citeauthoryear{Lee, Rong, and Hsu}{Lee et~al\mbox{.}}{2007}]%
        {lee2007makes}
\bibfield{author}{\bibinfo{person}{Wee Lee}, \bibinfo{person}{Nan Rong}, {and}
  \bibinfo{person}{David Hsu}.} \bibinfo{year}{2007}\natexlab{}.
\newblock \showarticletitle{What makes some POMDP problems easy to
  approximate?}
\newblock \bibinfo{journal}{\emph{Advances in neural information processing
  systems}}  \bibinfo{volume}{20} (\bibinfo{year}{2007}).
\newblock


\bibitem[\protect\citeauthoryear{Littman}{Littman}{1994}]%
        {littman1994markov}
\bibfield{author}{\bibinfo{person}{Michael~L Littman}.}
  \bibinfo{year}{1994}\natexlab{}.
\newblock \showarticletitle{Markov games as a framework for multi-agent
  reinforcement learning}.
\newblock In \bibinfo{booktitle}{\emph{Machine learning proceedings 1994}}.
  \bibinfo{publisher}{Elsevier}, \bibinfo{pages}{157--163}.
\newblock


\bibitem[\protect\citeauthoryear{Lowe, Wu, Tamar, Harb, Pieter~Abbeel, and
  Mordatch}{Lowe et~al\mbox{.}}{2017}]%
        {lowe2017multi}
\bibfield{author}{\bibinfo{person}{Ryan Lowe}, \bibinfo{person}{Yi~I Wu},
  \bibinfo{person}{Aviv Tamar}, \bibinfo{person}{Jean Harb},
  \bibinfo{person}{OpenAI Pieter~Abbeel}, {and} \bibinfo{person}{Igor
  Mordatch}.} \bibinfo{year}{2017}\natexlab{}.
\newblock \showarticletitle{Multi-agent actor-critic for mixed
  cooperative-competitive environments}.
\newblock \bibinfo{journal}{\emph{Advances in neural information processing
  systems}}  \bibinfo{volume}{30} (\bibinfo{year}{2017}).
\newblock


\bibitem[\protect\citeauthoryear{Madani, Hanks, and Condon}{Madani
  et~al\mbox{.}}{2003}]%
        {madani2003undecidability}
\bibfield{author}{\bibinfo{person}{Omid Madani}, \bibinfo{person}{Steve Hanks},
  {and} \bibinfo{person}{Anne Condon}.} \bibinfo{year}{2003}\natexlab{}.
\newblock \showarticletitle{On the undecidability of probabilistic planning and
  related stochastic optimization problems}.
\newblock \bibinfo{journal}{\emph{Artificial Intelligence}}
  \bibinfo{volume}{147}, \bibinfo{number}{1-2} (\bibinfo{year}{2003}),
  \bibinfo{pages}{5--34}.
\newblock


\bibitem[\protect\citeauthoryear{Megiddo and Wigderson}{Megiddo and
  Wigderson}{1986}]%
        {megiddo1986play}
\bibfield{author}{\bibinfo{person}{Nimrod Megiddo} {and} \bibinfo{person}{Avi
  Wigderson}.} \bibinfo{year}{1986}\natexlab{}.
\newblock \showarticletitle{On play by means of computing machines: preliminary
  version}. In \bibinfo{booktitle}{\emph{Theoretical aspects of reasoning about
  knowledge}}. Elsevier, \bibinfo{pages}{259--274}.
\newblock


\bibitem[\protect\citeauthoryear{Mirsky, Carlucho, Rahman, Fosong, Macke,
  Sridharan, Stone, and Albrecht}{Mirsky et~al\mbox{.}}{2022}]%
        {mirsky2022survey}
\bibfield{author}{\bibinfo{person}{Reuth Mirsky}, \bibinfo{person}{Ignacio
  Carlucho}, \bibinfo{person}{Arrasy Rahman}, \bibinfo{person}{Elliot Fosong},
  \bibinfo{person}{William Macke}, \bibinfo{person}{Mohan Sridharan},
  \bibinfo{person}{Peter Stone}, {and} \bibinfo{person}{Stefano~V. Albrecht}.}
  \bibinfo{year}{2022}\natexlab{}.
\newblock \bibinfo{title}{A Survey of Ad Hoc Teamwork Research}.
\newblock
\newblock
\showeprint[arxiv]{2202.10450}~[cs.MA]


\bibitem[\protect\citeauthoryear{Mnih, Badia, Mirza, Graves, Lillicrap, Harley,
  Silver, and Kavukcuoglu}{Mnih et~al\mbox{.}}{2016}]%
        {mnih2016asynchronous}
\bibfield{author}{\bibinfo{person}{Volodymyr Mnih},
  \bibinfo{person}{Adria~Puigdomenech Badia}, \bibinfo{person}{Mehdi Mirza},
  \bibinfo{person}{Alex Graves}, \bibinfo{person}{Timothy Lillicrap},
  \bibinfo{person}{Tim Harley}, \bibinfo{person}{David Silver}, {and}
  \bibinfo{person}{Koray Kavukcuoglu}.} \bibinfo{year}{2016}\natexlab{}.
\newblock \showarticletitle{Asynchronous methods for deep reinforcement
  learning}. In \bibinfo{booktitle}{\emph{International conference on machine
  learning}}. PmLR, \bibinfo{pages}{1928--1937}.
\newblock


\bibitem[\protect\citeauthoryear{Mnih, Kavukcuoglu, Silver, Graves, Antonoglou,
  Wierstra, and Riedmiller}{Mnih et~al\mbox{.}}{2013}]%
        {mnih2013playing}
\bibfield{author}{\bibinfo{person}{Volodymyr Mnih}, \bibinfo{person}{Koray
  Kavukcuoglu}, \bibinfo{person}{David Silver}, \bibinfo{person}{Alex Graves},
  \bibinfo{person}{Ioannis Antonoglou}, \bibinfo{person}{Daan Wierstra}, {and}
  \bibinfo{person}{Martin Riedmiller}.} \bibinfo{year}{2013}\natexlab{}.
\newblock \showarticletitle{Playing atari with deep reinforcement learning}.
\newblock \bibinfo{journal}{\emph{arXiv preprint arXiv:1312.5602}}
  (\bibinfo{year}{2013}).
\newblock


\bibitem[\protect\citeauthoryear{Nachbar and Zame}{Nachbar and Zame}{1996}]%
        {nachbar1996non}
\bibfield{author}{\bibinfo{person}{John~H Nachbar} {and}
  \bibinfo{person}{William~R Zame}.} \bibinfo{year}{1996}\natexlab{}.
\newblock \showarticletitle{Non-computable strategies and discounted repeated
  games}.
\newblock \bibinfo{journal}{\emph{Economic theory}}  \bibinfo{volume}{8}
  (\bibinfo{year}{1996}), \bibinfo{pages}{103--122}.
\newblock


\bibitem[\protect\citeauthoryear{Ong, Png, Hsu, and Lee}{Ong
  et~al\mbox{.}}{2010}]%
        {ong2010planning}
\bibfield{author}{\bibinfo{person}{Sylvie~CW Ong}, \bibinfo{person}{Shao~Wei
  Png}, \bibinfo{person}{David Hsu}, {and} \bibinfo{person}{Wee~Sun Lee}.}
  \bibinfo{year}{2010}\natexlab{}.
\newblock \showarticletitle{Planning under uncertainty for robotic tasks with
  mixed observability}.
\newblock \bibinfo{journal}{\emph{The International Journal of Robotics
  Research}} \bibinfo{volume}{29}, \bibinfo{number}{8} (\bibinfo{year}{2010}),
  \bibinfo{pages}{1053--1068}.
\newblock


\bibitem[\protect\citeauthoryear{Puig, Undersander, Szot, Cote, Yang, Partsey,
  Desai, Clegg, Hlavac, Min, et~al\mbox{.}}{Puig et~al\mbox{.}}{2023}]%
        {puig2023habitat}
\bibfield{author}{\bibinfo{person}{Xavier Puig}, \bibinfo{person}{Eric
  Undersander}, \bibinfo{person}{Andrew Szot}, \bibinfo{person}{Mikael~Dallaire
  Cote}, \bibinfo{person}{Tsung-Yen Yang}, \bibinfo{person}{Ruslan Partsey},
  \bibinfo{person}{Ruta Desai}, \bibinfo{person}{Alexander~William Clegg},
  \bibinfo{person}{Michal Hlavac}, \bibinfo{person}{So~Yeon Min},
  {et~al\mbox{.}}} \bibinfo{year}{2023}\natexlab{}.
\newblock \showarticletitle{Habitat 3.0: A co-habitat for humans, avatars and
  robots}.
\newblock \bibinfo{journal}{\emph{arXiv preprint arXiv:2310.13724}}
  (\bibinfo{year}{2023}).
\newblock


\bibitem[\protect\citeauthoryear{Puterman}{Puterman}{2014}]%
        {puterman2014markov}
\bibfield{author}{\bibinfo{person}{Martin~L Puterman}.}
  \bibinfo{year}{2014}\natexlab{}.
\newblock \bibinfo{booktitle}{\emph{Markov decision processes: discrete
  stochastic dynamic programming}}.
\newblock \bibinfo{publisher}{John Wiley \& Sons}.
\newblock


\bibitem[\protect\citeauthoryear{Raffin, Hill, Gleave, Kanervisto, Ernestus,
  and Dormann}{Raffin et~al\mbox{.}}{2021}]%
        {stable-baselines3}
\bibfield{author}{\bibinfo{person}{Antonin Raffin}, \bibinfo{person}{Ashley
  Hill}, \bibinfo{person}{Adam Gleave}, \bibinfo{person}{Anssi Kanervisto},
  \bibinfo{person}{Maximilian Ernestus}, {and} \bibinfo{person}{Noah Dormann}.}
  \bibinfo{year}{2021}\natexlab{}.
\newblock \showarticletitle{Stable-Baselines3: Reliable Reinforcement Learning
  Implementations}.
\newblock \bibinfo{journal}{\emph{Journal of Machine Learning Research}}
  \bibinfo{volume}{22}, \bibinfo{number}{268} (\bibinfo{year}{2021}),
  \bibinfo{pages}{1--8}.
\newblock
\urldef\tempurl%
\url{http://jmlr.org/papers/v22/20-1364.html}
\showURL{%
\tempurl}


\bibitem[\protect\citeauthoryear{Rashid, Samvelyan, De~Witt, Farquhar,
  Foerster, and Whiteson}{Rashid et~al\mbox{.}}{2020}]%
        {rashid2020monotonic}
\bibfield{author}{\bibinfo{person}{Tabish Rashid}, \bibinfo{person}{Mikayel
  Samvelyan}, \bibinfo{person}{Christian~Schroeder De~Witt},
  \bibinfo{person}{Gregory Farquhar}, \bibinfo{person}{Jakob Foerster}, {and}
  \bibinfo{person}{Shimon Whiteson}.} \bibinfo{year}{2020}\natexlab{}.
\newblock \showarticletitle{Monotonic value function factorisation for deep
  multi-agent reinforcement learning}.
\newblock \bibinfo{journal}{\emph{Journal of Machine Learning Research}}
  \bibinfo{volume}{21}, \bibinfo{number}{178} (\bibinfo{year}{2020}),
  \bibinfo{pages}{1--51}.
\newblock


\bibitem[\protect\citeauthoryear{Rubinstein}{Rubinstein}{1986}]%
        {rubinstein1986finite}
\bibfield{author}{\bibinfo{person}{Ariel Rubinstein}.}
  \bibinfo{year}{1986}\natexlab{}.
\newblock \showarticletitle{Finite automata play the repeated prisoner's
  dilemma}.
\newblock \bibinfo{journal}{\emph{Journal of economic theory}}
  \bibinfo{volume}{39}, \bibinfo{number}{1} (\bibinfo{year}{1986}),
  \bibinfo{pages}{83--96}.
\newblock


\bibitem[\protect\citeauthoryear{Schulman, Wolski, Dhariwal, Radford, and
  Klimov}{Schulman et~al\mbox{.}}{2017}]%
        {schulman2017proximal}
\bibfield{author}{\bibinfo{person}{John Schulman}, \bibinfo{person}{Filip
  Wolski}, \bibinfo{person}{Prafulla Dhariwal}, \bibinfo{person}{Alec Radford},
  {and} \bibinfo{person}{Oleg Klimov}.} \bibinfo{year}{2017}\natexlab{}.
\newblock \showarticletitle{Proximal policy optimization algorithms}.
\newblock \bibinfo{journal}{\emph{arXiv preprint arXiv:1707.06347}}
  (\bibinfo{year}{2017}).
\newblock


\bibitem[\protect\citeauthoryear{Shapley}{Shapley}{1953}]%
        {shapley1953stochastic}
\bibfield{author}{\bibinfo{person}{Lloyd~S Shapley}.}
  \bibinfo{year}{1953}\natexlab{}.
\newblock \showarticletitle{Stochastic games}.
\newblock \bibinfo{journal}{\emph{Proceedings of the national academy of
  sciences}} \bibinfo{volume}{39}, \bibinfo{number}{10} (\bibinfo{year}{1953}),
  \bibinfo{pages}{1095--1100}.
\newblock


\bibitem[\protect\citeauthoryear{Sharon, Stern, Felner, and Sturtevant}{Sharon
  et~al\mbox{.}}{2015}]%
        {sharon2015conflict}
\bibfield{author}{\bibinfo{person}{Guni Sharon}, \bibinfo{person}{Roni Stern},
  \bibinfo{person}{Ariel Felner}, {and} \bibinfo{person}{Nathan~R Sturtevant}.}
  \bibinfo{year}{2015}\natexlab{}.
\newblock \showarticletitle{Conflict-based search for optimal multi-agent
  pathfinding}.
\newblock \bibinfo{journal}{\emph{Artificial intelligence}}
  \bibinfo{volume}{219} (\bibinfo{year}{2015}), \bibinfo{pages}{40--66}.
\newblock


\bibitem[\protect\citeauthoryear{Shen, Peng, Liu, Zhang, Qian, Hong, Guo, Ding,
  Lu, and Tang}{Shen et~al\mbox{.}}{2020}]%
        {shen2020reinforcement}
\bibfield{author}{\bibinfo{person}{Weiran Shen}, \bibinfo{person}{Binghui
  Peng}, \bibinfo{person}{Hanpeng Liu}, \bibinfo{person}{Michael Zhang},
  \bibinfo{person}{Ruohan Qian}, \bibinfo{person}{Yan Hong},
  \bibinfo{person}{Zhi Guo}, \bibinfo{person}{Zongyao Ding},
  \bibinfo{person}{Pengjun Lu}, {and} \bibinfo{person}{Pingzhong Tang}.}
  \bibinfo{year}{2020}\natexlab{}.
\newblock \showarticletitle{Reinforcement mechanism design: With applications
  to dynamic pricing in sponsored search auctions}. In
  \bibinfo{booktitle}{\emph{Proceedings of the AAAI conference on artificial
  intelligence}}, Vol.~\bibinfo{volume}{34}. \bibinfo{pages}{2236--2243}.
\newblock


\bibitem[\protect\citeauthoryear{Simon}{Simon}{1990}]%
        {simon1990bounded}
\bibfield{author}{\bibinfo{person}{Herbert~A Simon}.}
  \bibinfo{year}{1990}\natexlab{}.
\newblock \showarticletitle{Bounded rationality}.
\newblock \bibinfo{journal}{\emph{Utility and probability}}
  (\bibinfo{year}{1990}), \bibinfo{pages}{15--18}.
\newblock


\bibitem[\protect\citeauthoryear{Solan and Vieille}{Solan and Vieille}{2015}]%
        {solan2015stochastic}
\bibfield{author}{\bibinfo{person}{Eilon Solan} {and} \bibinfo{person}{Nicolas
  Vieille}.} \bibinfo{year}{2015}\natexlab{}.
\newblock \showarticletitle{Stochastic games}.
\newblock \bibinfo{journal}{\emph{Proceedings of the National Academy of
  Sciences}} \bibinfo{volume}{112}, \bibinfo{number}{45}
  (\bibinfo{year}{2015}), \bibinfo{pages}{13743--13746}.
\newblock


\bibitem[\protect\citeauthoryear{Sondik}{Sondik}{1978}]%
        {sondik1978optimal}
\bibfield{author}{\bibinfo{person}{Edward~J Sondik}.}
  \bibinfo{year}{1978}\natexlab{}.
\newblock \showarticletitle{The optimal control of partially observable Markov
  processes over the infinite horizon: Discounted costs}.
\newblock \bibinfo{journal}{\emph{Operations research}} \bibinfo{volume}{26},
  \bibinfo{number}{2} (\bibinfo{year}{1978}), \bibinfo{pages}{282--304}.
\newblock


\bibitem[\protect\citeauthoryear{Stern}{Stern}{2019}]%
        {stern2019multi-overview}
\bibfield{author}{\bibinfo{person}{Roni Stern}.}
  \bibinfo{year}{2019}\natexlab{}.
\newblock \showarticletitle{Multi-agent path finding--an overview}.
\newblock \bibinfo{journal}{\emph{Artificial Intelligence}}
  (\bibinfo{year}{2019}), \bibinfo{pages}{96--115}.
\newblock


\bibitem[\protect\citeauthoryear{Su, Huo, Zhang, Dou, Yu, Xu, Lu, and Zheng}{Su
  et~al\mbox{.}}{2024}]%
        {su2024auctionnet}
\bibfield{author}{\bibinfo{person}{Kefan Su}, \bibinfo{person}{Yusen Huo},
  \bibinfo{person}{Zhilin Zhang}, \bibinfo{person}{Shuai Dou},
  \bibinfo{person}{Chuan Yu}, \bibinfo{person}{Jian Xu},
  \bibinfo{person}{Zongqing Lu}, {and} \bibinfo{person}{Bo Zheng}.}
  \bibinfo{year}{2024}\natexlab{}.
\newblock \showarticletitle{AuctionNet: A Novel Benchmark for Decision-Making
  in Large-Scale Games}.
\newblock \bibinfo{journal}{\emph{Advances in Neural Information Processing
  Systems}}  \bibinfo{volume}{37} (\bibinfo{year}{2024}),
  \bibinfo{pages}{94428--94452}.
\newblock


\bibitem[\protect\citeauthoryear{Takahashi}{Takahashi}{1964}]%
        {takahashi1964equilibrium}
\bibfield{author}{\bibinfo{person}{Masayuki Takahashi}.}
  \bibinfo{year}{1964}\natexlab{}.
\newblock \showarticletitle{Equilibrium points of stochastic non-cooperative $
  n $-person games}.
\newblock \bibinfo{journal}{\emph{Journal of Science of the Hiroshima
  University, Series AI (Mathematics)}} \bibinfo{volume}{28},
  \bibinfo{number}{1} (\bibinfo{year}{1964}), \bibinfo{pages}{95--99}.
\newblock


\bibitem[\protect\citeauthoryear{To{\v{s}}ic}{To{\v{s}}ic}{2016}]%
        {tovsic2016learning}
\bibfield{author}{\bibinfo{person}{Predrag~T To{\v{s}}ic}.}
  \bibinfo{year}{2016}\natexlab{}.
\newblock \showarticletitle{On Learning and Co-learning Effective Strategies in
  Iterated Travelers' Dilemma}. In \bibinfo{booktitle}{\emph{2016 IEEE/WIC/ACM
  International Conference on Web Intelligence (WI)}}. IEEE,
  \bibinfo{pages}{674--677}.
\newblock


\bibitem[\protect\citeauthoryear{Vaswani, Shazeer, Parmar, Uszkoreit, Jones,
  Gomez, Kaiser, and Polosukhin}{Vaswani et~al\mbox{.}}{2017}]%
        {vaswani2017attention}
\bibfield{author}{\bibinfo{person}{Ashish Vaswani}, \bibinfo{person}{Noam
  Shazeer}, \bibinfo{person}{Niki Parmar}, \bibinfo{person}{Jakob Uszkoreit},
  \bibinfo{person}{Llion Jones}, \bibinfo{person}{Aidan~N Gomez},
  \bibinfo{person}{{\L}ukasz Kaiser}, {and} \bibinfo{person}{Illia
  Polosukhin}.} \bibinfo{year}{2017}\natexlab{}.
\newblock \showarticletitle{Attention is all you need}.
\newblock \bibinfo{journal}{\emph{Advances in neural information processing
  systems}}  \bibinfo{volume}{30} (\bibinfo{year}{2017}).
\newblock


\bibitem[\protect\citeauthoryear{Wang and Lin}{Wang and Lin}{2019}]%
        {wang2019pure}
\bibfield{author}{\bibinfo{person}{Shiheng Wang} {and}
  \bibinfo{person}{Fangzhen Lin}.} \bibinfo{year}{2019}\natexlab{}.
\newblock \showarticletitle{Pure Strategy Best Responses to Mixed Strategies in
  Repeated Games}.
\newblock \bibinfo{journal}{\emph{arXiv preprint arXiv:1902.09066}}
  (\bibinfo{year}{2019}).
\newblock


\bibitem[\protect\citeauthoryear{Yu, Velu, Vinitsky, Gao, Wang, Bayen, and
  Wu}{Yu et~al\mbox{.}}{2022}]%
        {yu2022the}
\bibfield{author}{\bibinfo{person}{Chao Yu}, \bibinfo{person}{Akash Velu},
  \bibinfo{person}{Eugene Vinitsky}, \bibinfo{person}{Jiaxuan Gao},
  \bibinfo{person}{Yu Wang}, \bibinfo{person}{Alexandre Bayen}, {and}
  \bibinfo{person}{Yi Wu}.} \bibinfo{year}{2022}\natexlab{}.
\newblock \showarticletitle{The Surprising Effectiveness of {PPO} in
  Cooperative Multi-Agent Games}. In \bibinfo{booktitle}{\emph{Thirty-sixth
  Conference on Neural Information Processing Systems Datasets and Benchmarks
  Track}}.
\newblock
\urldef\tempurl%
\url{https://openreview.net/forum?id=YVXaxB6L2Pl}
\showURL{%
\tempurl}


\bibitem[\protect\citeauthoryear{Zhang, Du, Shan, Zhou, Du, Tenenbaum, Shu, and
  Gan}{Zhang et~al\mbox{.}}{2023}]%
        {zhang2023building}
\bibfield{author}{\bibinfo{person}{Hongxin Zhang}, \bibinfo{person}{Weihua Du},
  \bibinfo{person}{Jiaming Shan}, \bibinfo{person}{Qinhong Zhou},
  \bibinfo{person}{Yilun Du}, \bibinfo{person}{Joshua~B Tenenbaum},
  \bibinfo{person}{Tianmin Shu}, {and} \bibinfo{person}{Chuang Gan}.}
  \bibinfo{year}{2023}\natexlab{}.
\newblock \showarticletitle{Building cooperative embodied agents modularly with
  large language models}.
\newblock \bibinfo{journal}{\emph{arXiv preprint arXiv:2307.02485}}
  (\bibinfo{year}{2023}).
\newblock


\bibitem[\protect\citeauthoryear{Zhu and Lin}{Zhu and Lin}{2025}]%
        {zhu2025single}
\bibfield{author}{\bibinfo{person}{Fengming Zhu} {and}
  \bibinfo{person}{Fangzhen Lin}.} \bibinfo{year}{2025}\natexlab{}.
\newblock \showarticletitle{Single-Agent Planning in a Multi-Agent System: A
  Unified Framework for Type-Based Planners}. In
  \bibinfo{booktitle}{\emph{Proceedings of the 24th International Conference on
  Autonomous Agents and Multiagent Systems}}. \bibinfo{pages}{2382--2391}.
\newblock


\bibitem[\protect\citeauthoryear{Zuo and Tang}{Zuo and Tang}{2015}]%
        {zuo2015optimal}
\bibfield{author}{\bibinfo{person}{Song Zuo} {and} \bibinfo{person}{Pingzhong
  Tang}.} \bibinfo{year}{2015}\natexlab{}.
\newblock \showarticletitle{Optimal machine strategies to commit to in
  two-person repeated games}. In \bibinfo{booktitle}{\emph{Proceedings of the
  AAAI Conference on Artificial Intelligence}}, Vol.~\bibinfo{volume}{29}.
\newblock


\end{thebibliography}


\clearpage
\onecolumn
\appendix


\section{Missing Proofs for the Theoretic Results}

\subsection{For Theorem~\ref{thm:kmem_br}}
\label{app:proof:thm_kmem_br}

\begin{proof}
Our proof by induction is inspired by~\cite{puterman2014markov} (cf. Chapter 5.5).

	Given an SG, and an opponent strategy profile $\pi_{-i}^\infty \in \Pi_{-i}^\infty$, the induced MDP in general is
	$\mathcal{M}^\infty(\pi_{-i}^\infty) = \langle \mathcal{H}^\infty\times \mathcal{S}, \mathcal{A}_i, T^\infty_{\pi_{-i}}, R^\infty_{\pi_{-i}}, \gamma \rangle$,
	\begin{itemize}
		\item $\mathcal{A}_i$ and $\gamma$ inherit from the previous setup,
		\item A state is now consisting the whole history plus the current environment state, i.e. $\mathcal{H}^\infty\times \mathcal{S}$,
		\item Transitions are now made also for the complete histories, as we have
		\[Pr(a_{-i}, S' | H, S, a_i) = T(S'|S,a)\pi_{-i}^\infty(a_{-i}|H, S)\]
		Therefore, for $(H', S')$, $(H, S) \in \mathcal{H}^\infty\times \mathcal{S}$, 
		\[
		T^\infty(H', S'| H, S, a_i) \triangleq
		\begin{cases}
			T(S'|S,a)\pi_{-i}^\infty(a_{-i}|H, S),
			& \text{ if } H' = [H, S, (a_i, a_{-i})] \\
			0,
			& \text{ otherwise}\\
		\end{cases}
		\]
		where $[H, S, (a_i, a_{-i})]$ means to concatenate the existing history and the latest state-action tuple, which is a deterministic operation.
		\item Rewards on the complete histories: $R^\infty(H, S, a_i) \triangleq \sum_{a_{-i} \in \mathcal{A}_{-i}}R_i(S,a)\pi_{-i}^\infty(a_{-i}|H, S)$,
	\end{itemize}
The above $\mathcal{M}^\infty$ is trivially a valid MDP because transitions are made among complete state trajectories where the Markov property must hold.
	
Now we will show that if there exists a $\pi_{-i}^K \in \Pi_{-i}^K$, such that
\begin{equation}
\pi_{-i}^\infty (a_i | (H^K, H^-), S) = \pi_{-i}^K(a_i | H^K, S)	
\label{eq:pi_proj}
\end{equation}
where $H^K = H[-\min\{K, len(H)\}:]$ and $H^- = H[:-\min\{K, len(H)\}]$ (the latest $K$ historical records and the remaining prefix),
then for the control policy of this MDP, it is sufficient for agent $i$ to restrict the attention to $\Pi_i^K$ instead of general $\Pi_i^\infty$. More specifically, given an $\pi_i^\infty \in \Pi_i^\infty$, it is possible to construct a memory-restricted alternative $\pi_i^K$ such that the following target equation holds
\begin{equation}
Pr^{\pi_i^K}(H^K,S,a_i) = Pr^{\pi_i^\infty}(H^K,S,a_i)
\label{eq:jointSA}
\end{equation}
where $Pr^{\pi}$ means the probability under the particular policy $\pi$. The above proof target is sufficient in terms of seeking for an equivalent solution because it directly pertains to the reward function.
We will show that such a strategy for agent $i$ can be constructed by the following, i.e. by marginalizing over histories happened earlier then $K$ steps ago,
\begin{equation}
\pi_i^K (a_i | H^K, S) = \sum_{H^-} \pi_i^\infty(a_i|H^K, H^-, S) Pr(H^-)
\label{eq:infty2k}
\end{equation}

We will prove this equation by induction.

\textit{For the base case}, when $|H| = 0$ which simply means $S$ is the initial state, then Equation~(\ref{eq:jointSA}) obviously holds.

\textit{For the inductive case}, we hypothesize that the following holds for all possible $(\hat H, \hat S)$ with $|\hat H| = t-1$,
\[
Pr^{\pi_i^K}(\hat H^K,\hat S,a_i) = Pr^{\pi_i^\infty}(\hat H^K,\hat S,a_i)
\]
Because of Equation~(\ref{eq:pi_proj}), we have
\begin{equation}
\begin{split}
Pr(a_{-i}, S' | H, S, a_i)
& = T(S'|S,a)\pi_{-i}^\infty(a_{-i}|(H^K, H^-), S) \\
& = T(S'|S,a)\pi_{-i}^K(a_{-i}|H^K, S) \\
& = Pr(a_{-i}, S' | H^K, S, a_i)
\label{eq:HtoHk_prob}
\end{split}
\end{equation}
Then for $|H| = t$, we have
\begin{equation}
\begin{split}
Pr^{\pi_i^K}(H^K, S)
& = \sum_{(\hat H^K, \hat S)} \sum_{a_i'} Pr^{\pi_i^K}(\hat H^K, \hat S, a_i')
\textcolor{teal}{Pr^{\pi_i^K}(H^K, S |\hat H^K, \hat S, a_i')}\\
& = \sum_{(\hat H^K, \hat S)} \sum_{a_i'} Pr^{\pi_i^\infty}(\hat H^K, \hat S, a_i') 
\textcolor{teal}{Pr^{\pi_i^\infty}(H^K, S |\hat H^K, \hat S, a_i')}\\
& = Pr^{\pi_i^\infty}(H^K, S)
\end{split}
\label{eq:hs_prob}
\end{equation}
The second equality directly follows from the inductive hypothesis and Equation~(\ref{eq:HtoHk_prob}).
Note that, for the terms in \textcolor{teal}{teal}, it does not matter which rollout policy is used, as $a_i'$ is conditioned.

Finally, we have
\[
\begin{split}
Pr^{\pi_i^K}(H^K, S, a_i)
& = Pr^{\pi_i^K}(H^K, S) \times Pr^{\pi_i^K}(a_i | H^K, S)	 \\
& = Pr^{\pi_i^K}(H^K, S) \times \pi_i^K (a_i | H^K, S) \\
& = Pr^{\pi_i^K}(H^K, S) \times \sum_{H^-} \pi_i^\infty(a_i|H^K, H^-, S) Pr(H^-)\\
& = Pr^{\pi_i^\infty}(H^K, S) \times Pr^{\pi_i^\infty}(a_i | H^K, S)	 \\
& = Pr^{\pi_i^\infty}(H^K, S, a_i)
\end{split}
\]
The third equality holds according to Equation~(\ref{eq:infty2k}), and the fourth equality directly follows from Equation~(\ref{eq:hs_prob}).

\end{proof}

\subsection{For Theorem~\ref{thm:ne_beh} (the contraction mapping part)}
\label{app:proof:thm_ne}

We will show that, given any strategy profile $\{\pi_i\}_{i\in \mathcal{N}}$, a unique solution, i.e., a set of values $\{v_i(\cdot,\cdot)\}_{i\in \mathcal{N}}$, for Equation~(\ref{app:eq:bee}) is guaranteed to exist. For simplicity, we use $v_i$ as a shorthand for $v_i|_{\pi_i}^{\pi_{-i}}$.

\begin{equation}
\begin{split}
& v_1(H,S) = 
 \sum_{a_1\in\mathcal{A}_1} \pi_1(a_1|H, S) 
\Big[ R^K_{\pi_{-1}}(H, S, a_1)
 + \gamma \sum_{H',S'} T^K_{\pi_{-1}}(H', S'| H, S, a_1) v_1|_{\pi_1}^{\pi_{-1}}(H',S') \Big] \\
& \cdots \\
& v_n(H,S) = 
 \sum_{a_n\in\mathcal{A}_n} \pi_n(a_n|H, S) 
\Big[ R^K_{\pi_{-n}}(H, S, a_n)
 + \gamma \sum_{H',S'} T^K_{\pi_{-n}}(H', S'| H, S, a_n) v_n|_{\pi_n}^{\pi_{-n}}(H',S') \Big]	
\end{split}
\label{app:eq:bee}	
\end{equation}

\begin{proof}

Let $\mathcal{V}$ denote the vector space of all possible value functions, where each $v\in \mathcal{V}$ is a function $\mathcal{N}\times\mathcal{H}^{\leq K} \times \mathcal{S} \mapsto \mathbb{R}$ (slightly reloading the notation $v_i(H,S)$). Let $\Xi: \mathcal{V} \mapsto \mathcal{V}$ denote the (multi-agent) Bellman optimality operator given as follows,
	\[
	\Xi(v)(i,H,S) = \sum_{a_i\in\mathcal{A}_i} \pi_i(a_i|H, S) \Big[ R^K_{\pi_{-i}}(H, S, a_i) + \gamma \sum_{H',S'} T^K_{\pi_{-i}}(H', S'| H, S, a_i) v(i, H',S') \Big]
	\]
	For the rest, we write $\Xi_v$ interchangeably with $\Xi(v)$ for better presentation.
	We use the infinity norm as the distance measure, defined as $\|v\|_\infty = \max_x|v(x)|$ for $v \in \mathcal{V}$.
	We then show for any two vectors $u, v \in \mathcal{V}$, we have $\|\Xi(u) - \Xi(v)\|_\infty \leq \gamma \|u - v\|_\infty$.
	
	\[
	\begin{split}
	|\Xi_u(i,H,S) - \Xi_v(i,H,S)|
	& = \Xi_u(i,H,S) - \Xi_v(i,H,S) \\
	& = \sum_{a_i\in\mathcal{A}_i} \pi_i(a_i|H, S) \Big[ R^K_{\pi_{-i}}(H, S, a_i) + \gamma \sum_{H',S'} T^K_{\pi_{-i}}(H', S'| H, S, a_i) u(i, H',S') \Big] \\
	& - \sum_{a_i\in\mathcal{A}_i} \pi_i(a_i|H, S) \Big[ R^K_{\pi_{-i}}(H, S, a_i) + \gamma \sum_{H',S'} T^K_{\pi_{-i}}(H', S'| H, S, a_i) v(i, H',S') \Big] \\
	& = \gamma \sum_{H',S'} T^K_{\pi_{-i}}(H', S'| H, S, a_i) \Big [ u(i, H',S') - v(i, H',S') \Big] \\
	& \leq \gamma \sum_{H',S'} T^K_{\pi_{-i}}(H', S'| H, S, a_i) | u(i, H',S') - v(i, H',S') | \\
	& \leq \gamma \sum_{H',S'} T^K_{\pi_{-i}}(H', S'| H, S, a_i) \| u - v \|_\infty \\
	& = \gamma \| u - v \|_\infty \sum_{H',S'} T^K_{\pi_{-i}}(H', S'| H, S, a_i) \\
	& = \gamma \| u - v \|_\infty
	\end{split}
	\]
	Overall, we have
	\[
	\|\Xi(u) - \Xi(v)\|_\infty = \max_{i,H,S} |\Xi_u(i,H,S) - \Xi_v(i,H,S)| \leq \gamma \|u - v\|_\infty
	\]
	Thus, $\Xi$ is a contraction mapping, and it naturally follows that $\Xi$ has only one unique fixed point.
\end{proof}

\subsection{For Theorem~\ref{thm:noneq_sg}}

\label{app:proof:thm_noneq_sg}

\begin{proof} Similarly as before, to evaluate an arbitrary $\pi_i$ under $\mathcal{M}(\pi^\iota_{-i})$ , one can establish the following
\[
\begin{split}
V_{mix}(S)
& = \sum_{a_i \in \mathcal{A}_i} \pi_i(a_i|S) \sum_{\iota} p_\iota \cdot Q_{\mathcal{M}(\pi^\iota_{-i})}(S, a_i) \\
& = \sum_{a_i \in \mathcal{A}_i} \pi_i(a_i|S)  \sum_{\iota} p_\iota \cdot \Big [ R_{\pi^\iota_{-i}}(S, a_i) + \gamma \sum_{S'} T_{\pi^\iota_{-i}}(S'|S, a_i)V_{\mathcal{M}(\pi^\iota_{-i})}(S')\Big ] \\
& = \sum_{a_i \in \mathcal{A}_i} \pi_i(a_i|S)  \sum_{\iota} p_\iota \cdot \Big [\sum_{a_i \in \mathcal{A}_{-i}}R_i(S,a)\pi^\iota_{-i}(a_{-i}|S) + \gamma \sum_{S'} \sum_{a_{-i} \in \mathcal{A}_{-i}}T(S'|S,a)\pi^\iota_{-i}(a_{-i}|S) V_{\mathcal{M}(\pi^\iota_{-i})}(S')\Big ] \\
& = \sum_{a_i \in \mathcal{A}_i} \pi_i(a_i|S) \Big [ \sum_{a_i \in \mathcal{A}_{-i}}R_i(S,a) \sum_{\iota} p_\iota\pi^\iota_{-i}(a_{-i}|S) + \gamma  \sum_{S'} \sum_{a_{-i} \in \mathcal{A}_{-i}}T(S'|S,a)\sum_{\iota} p_\iota\pi^\iota_{-i}(a_{-i}|S) V_{\mathcal{M}(\pi^\iota_{-i})}(S')\Big ] \\
&\textcolor{teal}{
\ \neq  \sum_{a_i \in \mathcal{A}_i} \pi_i(a_i|S)\Big [ \sum_{a_i \in \mathcal{A}_{-i}}R_i(S,a)
\underbrace{\sum_{\iota} p_\iota\pi^\iota_{-i}(a_{-i}|S)}_{\omega_{(\Pi^{0+}_{-i}, \vec p)}}
+ \gamma  \sum_{S'} \sum_{a_{-i} \in \mathcal{A}_{-i}}T(S'|S,a)
\underbrace{\sum_{\iota} p_\iota\pi^\iota_{-i}(a_{-i}|S)}_{\omega_{(\Pi^{0+}_{-i}, \vec p)}}
V_{mix}(S')\Big ]} \\
\end{split}
\]
The last equation does not necessarily hold as one cannot simply replace $ V_{\mathcal{M}(\pi^\iota_{-i})}$ with $V_{mix}$, as it will require to solve another totally different equation.
However, this particular equation is by definition the one that $V_{beh}$ should satisfy, i.e.,
\begin{equation}
\begin{split}
V_{beh}(S) = \sum_{a_i \in \mathcal{A}_i} \pi_i(a_i|S)\Big 
[  \sum_{a_i \in \mathcal{A}_{-i}}R_i(S,a)
\underbrace{\sum_{\iota} p_\iota\pi^\iota_{-i}(a_{-i}|S)}_{\omega_{(\Pi^{0+}_{-i}, \vec p)}} 
+  \gamma  \sum_{S'} \sum_{a_{-i} \in \mathcal{A}_{-i}}T(S'|S,a)
\underbrace{\sum_{\iota} p_\iota\pi^\iota_{-i}(a_{-i}|S)}_{\omega_{(\Pi^{0+}_{-i}, \vec p)}}
V_{beh}(S')\Big ]	
\end{split}
\label{eq:belief_induced_mdp}
\end{equation}
Hence, it is not necessarily the case that $V_{mix} = V_{beh}$.
\end{proof}

\subsection{For Theorem~\ref{thm:memk_br_pomdp}}

\label{app:proof:thm_memk_br_pomdp}

\begin{proof}
	Recall the corresponding POMDP
	is given as the tuple
	$\langle \mathcal{H}^K\times\mathcal{S}\times\Pi^{K+}_{-i}, \mathcal{A}_i, \mathcal{H}^K\times\mathcal{S}, \mathbf{T}, \mathbf{O}, \mathbf{R}, \gamma\rangle$,
	\begin{enumerate}
		\item States: $\mathcal{H}^K\times\mathcal{S}\times\Pi^{K+}_{-i}$ denote the set of underlying states. That is, a state in this POMDP is the history segment and environment state of the completely observable stochastic game augmented by the unobservable opponent strategies.
		\item As previously, $\mathcal{A}_i$ is the set of available control actions of agent $i$, and $\gamma$ the discount factor.
		\item Observations: $\mathcal{H}^K\times\mathcal{S}$ denote the set of observations that can be made by agent $i$.
		\item $\mathbf{T}: (\mathcal{H}^K\times\mathcal{S}\times\Pi^{K+}_{-i}) \times \mathcal{A}_i \mapsto \Delta(\mathcal{H}^K\times\mathcal{S}\times\Pi^{K+}_{-i})$ denote transition function, mathematically defined as
		\[ \mathbf{T}\Big((H', S', \pi_{-i}') \Big| (H, S, \pi_{-i}), a_i \Big) =
		\begin{cases}
			T^K_{\pi_{-i}}(H', S'| H, S, a_i) & \text{, if } \pi_{-i}' = \pi_{-i} \\
			0 & \text{, otherwise} 
		\end{cases}
		\]
		\item $\mathbf{O}: (\mathcal{H}^K\times\mathcal{S}\times\Pi^{K+}_{-i}) \mapsto \mathcal{H}^K\times\mathcal{S}$ denote the deterministic observation function, mathematically defined as
		\[
		\mathbf{O}\Big((H, S, \pi_{-i}')\Big) = (H, S)
		\]
		\item $\mathbf{R}: (\mathcal{H}^K\times\mathcal{S}\times\Pi^{K+}_{-i}) \times \mathcal{A}_i \mapsto \mathbb{R}$, mathematically defined as
		\[
		\mathbf{R}\Big((H, S, \pi_{-i}'), a_i\Big) = R^K_{\pi_{-i}}(H, S, a_i)		\]
	\end{enumerate}

Then, we need to show that such a reduction is correct, i.e., a solution maximizes agent $i$' expected payoff under the stochastic game w.r.t. the opponents' mixed strategy iff it maximizes the expected return in this reduced POMDP.
The argument is made by three steps:
\begin{enumerate}
	\item Given any initial state $S \in \mathcal{S}$, and any sequence of joint actions, the amount of historic information that an agent with perfect recall can possibly obtain will be the same at each timestep under both models.
	{\it
	\begin{itemize}
		\item The accumulated information that agent $i$ in the stochastic game can gather is the following set
		\[
		\{\vec q, S_0, a_{i,0}, a_{-i, 0}, S_1, a_{i,1}, a_{-i, 1}, \cdots, S_t\}
		\]
		and that in the reduced POMDP is all the historic observations
		\[
		\begin{split}
		     & \{\vec q, S_0\} \\
		\cup & \{(S_0, a_{i,0}, a_{-i, 0}), S_1\} \\
		\cup & \{(S_0, a_{i,0}, a_{-i, 0}), (S_1, a_{i,1}, a_{-i, 1}), S_2\} \\
		     & \cdots \\
		\cup & \{(S_{t-K}, a_{i,t-K}, a_{-i, t-K}), \cdots, (S_{t-1}, a_{i,t-1}, a_{-i, t-1}), S_{t}\} \\
		=    & \{\vec q, S_0, a_{i,0}, a_{-i, 0}, S_1, a_{i,1}, a_{-i, 1}, \cdots, S_t\}
		\end{split} 
		\]
	\end{itemize}
	}
	\item Given any initial state $S \in \mathcal{S}$, and any sequence of agent $i$' actions, the probability of reaching the same trajectory will be the same.
	{\it
	\begin{itemize}
		\item Because the opponents' actions are merely sampled from a constant-memory strategy.
	\end{itemize}
	}
	\item Given any initial state $S \in \mathcal{S}$, and policy that maps from all possible historic information to actions will result in the same payoff under both models.
	{\it
	\begin{itemize}
		\item Note that in an episode (or a match), the opponents will not switch to another strategy profile, therefore, the total return/payoff will solely depend on the probabilities of each possible trajectory under the two models, which is ensured to be the same by the aforementioned two points. 
	\end{itemize}
	}
\end{enumerate}

\end{proof}

\section{Nash Equilibria for Mixed Constant-memory Strategies}

\label{app:ne_mixed}
We examine whether a group of agents can form some equilibrium if all of them play mixed strategies, i.e., $\{(\Pi^{K+}_{i}, \vec p_i)\}_{i \in \mathcal{N}}$.
The story is that, if the support $\{\Pi^{K+}_{i}\}_{i \in \mathcal{N}}$ and a distribution over initial states $d_0\in \Delta(\mathcal{S})$ can be specified in the first place, the stochastic game can be further reduced to a normal-form game $\langle \mathcal{N}, \{\Pi^{K+}_{i}\}_{i \in \mathcal{N}}, \{u_i\}_{i \in \mathcal{N}} \rangle$,
\begin{enumerate}
	\item The game contains all agents $\mathcal{N}$,
	\item The action set of agent $i$ is $\Pi^{K+}_{i}$, i.e. to select a behavioral strategy therein,
	\item The payoff of agent $i$ is
	\[
	u_i(\pi_i, \pi_{-i}) = \sum_{S\in \mathcal{S}} d_0(S) \cdot \mathbb{E}_{(\pi_i, \pi_{-i})}\Big[ \sum_{t=0}^\infty \gamma^t R_{i,t} \Big| S_0=S \Big]
	\]
\end{enumerate}
Under this sense and provided that the reduced game is finite, invoking Nash's well-known existence theorem, we can conclude that there must exists a mixed strategy NE $\{\vec{p_i}^*\}_{i \in \mathcal{N}}$. That is, given fixed supports $\{\Pi^{K+}_{i}\}_{i \in \mathcal{N}}$, no one will be strictly better off by unitarily deviating from $\{\vec{p_i}^*\}_{i \in \mathcal{N}}$ to another distribution for mixing over its support strategies.
However, the application of this result remains an open (and perhaps even unjustified) problem.
One idea might be promising:
as we will see later, finding a behavioral strategy best response to a mixed strategy is computationally hard, but will it helps if it allows for finding a mixed strategy best response instead?

\section{Details for the Iterated Prisoner's Dilemma}

\label{app:prisoner}

For the readers' convenience, we first echo the payoff matrix of the Prisoner's Dilemma in Table~\ref{app:tab:pd}, with $T>R>P>S$.

\begin{table}[ht]
\begin{tabular}{c|c|c|}
\cline{2-3}
                        & C      & D      \\ \hline
\multicolumn{1}{|c|}{C} & (R, R) & (S, T) \\ \hline
\multicolumn{1}{|c|}{D} & (T, S) & (P, P)  \\ \hline
\end{tabular}
\caption{The payoff matrix of the Prisoner's Dilemma}
\label{app:tab:pd}
\vspace{-5mm}
\end{table}

\subsection{Some Formal Results}

\begin{theorem}
	Given the opponent playing a ``$N$-Tit(s)-for-$M$-Tat(s)'' strategy, there exists a best response strategy that can be implemented with $(M-1)$-memory.
\end{theorem}

\begin{proof}
	As a ``$N$-Tit(s)-for-$M$-Tat(s)'' strategy is a $\max(N,M)$-memory strategy, then Theorem~\ref{thm:kmem_br} implies that there must exist a $\max(N,M)$-memory strategy serving as a best response.
	\textit{
	We then show that 1) there exists a strategy within $\max(N,M)$-memory that can result in the following payoff sequence; and 2) any strategy that can result in the following payoff sequence is a best response,
	}
	\[
\underbrace{(T, S), (T, S), \cdots, (T, S)}_{M - 1}, (R, R),
\underbrace{(T, S), (T, S), \cdots, (T, S)}_{M - 1}, (R, R), \cdots
	\]
One can construct a $M$-memory strategy as ``$X$-D(s)-before-one-C'', with $X=M-1$ here. By its name, it means to start with defection, and then cooperate only after $(M-1)$ defections.
As the opponent will retaliate only when being defected $M$ times, therefore, the constructed strategy will make the opponent cooperate all the time, hence the above payoff sequence.
In fact, implementing such a ``$(M-1)$-D(s)-before-one-C'' strategy only requires the agent to keep track of the past $(M-1)$ actions of its own: if it has played one defection in the past $(M-1)$ rounds, then keep cooperating, otherwise defect for one round. \textit{Hence, it is actually a $(M-1)$-memory strategy.}

Also, to have a better sequence, one should note that, it is impossible to ``flip'' every $(R,R)$ to $(T,S)$, as the opponent will definitely defect after the $M$-the defection.
The only way to better off is to flip some of the $(R,R)$'s to $(T,S)$'s without sacrificing too much of the future return.
In fact, it is only possible to flip the first $(R, R)$ to $(T,S)$ (i.e., by playing an ``All-D'' strategy), and the rest all will be changed to $(P, P)$. As we will show detailed calculations in the next subsection, when the agent is patient enough (i.e., with high discount factor), such a deviation is not profitable. Nevertheless, one should note that even when the agent is impatient, and therefore, adopts the ``All-D'' strategy, the strategy can be implemented with 0-memory.
\end{proof}


\subsection{Phase Transition}

We also mentioned that the best response will transit from a ``$X$-D(s)-before-one-C'' strategy to a ``All-D'' one, when the discounted factor keeps decreasing (i.e., the agents become less patient and more myopic). Now we formally derive the critical point of the discount factor that triggers such phase transition.

Assume the column player is playing a ``$N$-Tit(s)-for-$M$-Tat(s)'' strategy, we compute the discounted accumulated payoffs of the row player performing different responding strategies.
\begin{enumerate}
	\item When the row player plays a ``$(M-1)$-D(s)-before-one-C'' strategy, the payoff sequence will be
\[
\underbrace{(T, S), (T, S), \cdots, (T, S)}_{M - 1}, (R, R),
\underbrace{(T, S), (T, S), \cdots, (T, S)}_{M - 1}, (R, R), \cdots
\]
The discounted accumulated payoff will be
\begin{equation}
\begin{split}
& \big ( T + \gamma T + \cdots + \gamma^{M-2}T + \gamma^{M-1}R \big )
 +  \big ( \gamma^M T + \gamma^{M+1} T + \cdots + \gamma^{2M-2}T + \gamma^{2M-1}R  \big ) + \cdots\\
= & \frac{T}{1-\gamma^M} + \frac{\gamma T}{1-\gamma^M} + \cdots + \frac{\gamma^{M-2} T}{1-\gamma^M} + \frac{\gamma^{M-1}R}{1-\gamma^M} \\
= & \frac{\frac{T}{1-\gamma^M}(1-\gamma^{M-1})}{1-\gamma}
\end{split}
\label{app:eq:dbeforec}
\end{equation}

\item When the row play plays an ``All-D'' strategy, the payoff sequence will be
\[
\underbrace{(T, S), (T, S), \cdots, (T, S)}_{M}, \underbrace{(P, P), \cdots}_{\text{forever}}
\]
The discounted accumulated payoff will be
\begin{equation}
\begin{split}
& \big ( T + \gamma T + \cdots + \gamma^{M-2}T + \gamma^{M-1}T \big )
 +  \big ( \gamma^M P + \gamma^{M+1} P + \cdots \big )\\
= & \frac{T(1-\gamma^M)}{1-\gamma} + \frac{\gamma^M P}{1 - \gamma}
\end{split}
\end{equation}
\label{app:eq:alld}
\end{enumerate}

Compare Eq~(\ref{app:eq:dbeforec}) and Eq~(\ref{app:eq:alld}), we first simplify it to the following, and then solve $\gamma$ in terms of the other constants.
\begin{equation}
	T(1-\gamma^M)(1-\gamma^M) + P\gamma^M(1-\gamma^M) = T(1-\gamma^{M-1}) + R\gamma^{M-1}(1-\gamma)
\end{equation}
It is intractable to solve it manually. In fact, such an equation with its order being a variable is infeasible to solve even resorting to sophisticated libraries like \textit{SymPy}.\footnote{https://www.sympy.org/en/index.html}. Therefore, we substitute $M$ with concrete values first and then solve $\gamma$ using \textit{SymPy}.

We list the closed-form solutions for $M$ up to 3, and substitute \{T=2, R=1, P=0, S=-1\} to the final expression.
\begin{enumerate}
	\item When $M = 1$.
	\[
	\begin{split}
	\gamma & = \frac{R - T}{P - T} = \frac{1}{2}
	\end{split}
	\]
	The other solution is 1.
	\item When $M = 2$.
	\[
	\gamma = \frac{-P + T - \sqrt{(P^2 + 4PR - 6PT - 4RT + 5T^2)}}{2P - 2T}
	= -\frac{1}{2} + \frac{\sqrt{3}}{2} \approx 0.366025
	\]
	The other three solutions are 0, 1, and a negative real number.
	\item When $M = 3$.
	\[
	\begin{split}
	\gamma & = \frac{-\Big(\frac{1}{2}\sqrt{(-7 + \frac{27(-R + T)}{(P - T)})^2 + 32} - \frac{7}{2} + \frac{27(-R + T)}{2(P - T)}\Big)^{\frac{1}{3}}}{3} - \frac{1}{3} + \frac{2}{3\Big(\frac{1}{2}\sqrt{(-7 + \frac{27(-R + T)}{(P - T)} )^2 + 32} - \frac{7}{2} + \frac{27(-R + T)}{2(P - T)}\Big)^\frac{1}{3}} \\
	& = -1/3 - (-41/4 + 3*\sqrt{201}/4)^{1/3}/3 + 2/(3*(-41/4 + 3*\sqrt{201}/4)^{1/3}) \\
	& \approx 0.342508
	\end{split}
	\]
	The other four solutions are 0, 1, and two complex numbers.
\end{enumerate} 

\subsection{Sample Outputs of the Computed Best Responses}

We here present the computed best response for Player 1 against
Player 2 who plays a ``2-Tits-For-2-Tats'' strategy.

\begin{verbatim}
=======================================+
BR to [2 Tits For 2 Tats], val = 15.26 | 
+--------------+-------------+-------------+----------+
| Histories    | P1 action   | P2 action   |   P1 val |
+==============+=============+=============+==========+
| []           | D           | C           |  15.2632 |
+--------------+-------------+-------------+----------+
| [D, D]       | C           | C           |  14.7368 |
+--------------+-------------+-------------+----------+
| [D, C]       | C           | C           |  14.7368 |
+--------------+-------------+-------------+----------+
| [C, D]       | D           | C           |  15.2632 |
+--------------+-------------+-------------+----------+
| [C, C]       | D           | C           |  15.2632 |
+--------------+-------------+-------------+----------+
| [D, D, D, D] | C           | D           |  11.7368 |
+--------------+-------------+-------------+----------+
| [D, D, D, C] | C           | D           |  11.7368 |
+--------------+-------------+-------------+----------+
| [D, D, C, D] | D           | C           |  15.2632 |
+--------------+-------------+-------------+----------+
| [D, D, C, C] | D           | C           |  15.2632 |
+--------------+-------------+-------------+----------+
| [D, C, D, D] | C           | D           |  11.7368 |
+--------------+-------------+-------------+----------+
| [D, C, D, C] | C           | D           |  11.7368 |
+--------------+-------------+-------------+----------+
| [D, C, C, D] | D           | C           |  15.2632 |
+--------------+-------------+-------------+----------+
| [D, C, C, C] | D           | C           |  15.2632 |
+--------------+-------------+-------------+----------+
| [C, D, D, D] | C           | C           |  14.7368 |
+--------------+-------------+-------------+----------+
| [C, D, D, C] | C           | C           |  14.7368 |
+--------------+-------------+-------------+----------+
| [C, D, C, D] | D           | C           |  15.2632 |
+--------------+-------------+-------------+----------+
| [C, D, C, C] | D           | C           |  15.2632 |
+--------------+-------------+-------------+----------+
| [C, C, D, D] | C           | C           |  14.7368 |
+--------------+-------------+-------------+----------+
| [C, C, D, C] | C           | C           |  14.7368 |
+--------------+-------------+-------------+----------+
| [C, C, C, D] | D           | C           |  15.2632 |
+--------------+-------------+-------------+----------+
| [C, C, C, C] | D           | C           |  15.2632 |
+--------------+-------------+-------------+----------+
\end{verbatim}

\section{Details for the Iterated Traveler's Dilemma}

\label{app:traveler}

\subsection{Generalized (One-Shot) Traveler's Dilemma}

It is conventionally a two-player game, but here we introduce a generalized multi-player version and then present the two-player version as a special case. 
This domain is of great significance as one can see its connection with PD, auction with the same common value and negotiation.
A one-shot multiple-player Traveller's Dilemma~(TD) consists of three parameters, denoted as $TD(N,k,A)$, where $N = [1..n]$ is the set of $n$ participating agents, $k > 1$ is a constant coefficient, and $A \subseteq \mathbb{N}$ is a finite set of possible (non-negative) biddings.
Given a bidding profile $\vec{a}=(a_i, a_{-i})\in A^N$ that is simultaneously reported from all agents, the utility for agent $i$ is calculated as
\[
u_i(a_i, a_{-i}) = \min(a_i, a_{-i}) + k\cdot sign(\min(a_{-i}) - a_i)
\]
where we slightly abuse $\min$ by allowing it to first flatten all its arguments which might be a scaler or a vector and then return whichever element that is the minimum.

\begin{figure}[htbp]
	\centering
	\includegraphics[width=50mm]{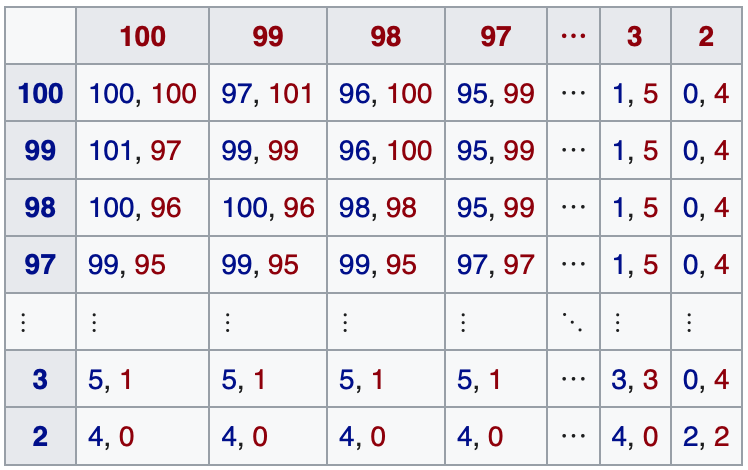}
	\caption{The payoff matrix for $TD(2,2,[2..100])$. The figure is borrowed from its wiki page.}
	\label{fig:td100}
\end{figure}

\paragraph{Profitable deviation.}
Suppose $a_i = \alpha$ and $a_{-i} = \alpha \cdot \mathbf{1}$, then
$u_i(a_i, a_{-i}) = u_i(\alpha, \alpha\cdot \mathbf{1}) = \alpha$.
Consider possible deviation as $a_i' \gets \alpha-x$,
the utility will be changed to
$u_i(a_i', a_{-i}) = \alpha - x + k$.
As long as there exists an $x$ such that $\alpha-x \in A$ and $k > x$, this will be a profitable deviation.
Note that one would never deviate to higher bids, say $a_i'' \gets\alpha + y \in A$, $u_i(a_i'', a_{-i}) = \alpha - k < \alpha = u_i(a_i, a_{-i})$.
In a nutshell, everyone is incentivized to bid slightly lower than the current lowest one.

\paragraph{Nash equilibrium.}
There is a unique NE, $a_1 = \cdots = a_n = \min(A)$, as none can bid even lower. Note that in terms of better response instead of best response, one does not have to bid lower than the current lowest one, she can just bid the same as the current lowest one. For example, $u_i(\alpha, (\alpha - 1) \cdot \mathbf{1}) = \alpha - 1 - k$, and $u_i(\alpha - 1, (\alpha - 1) \cdot \mathbf{1}) = \alpha - 1$.

\paragraph{Two-player traveller's dilemma.}
Figure~\ref{fig:td100} is an instance of $TD(N=2, k=2, A=[2..100])$.
For a detailed story for this game,
please refer to the wiki page\footnote{https://en.wikipedia.org/wiki/Traveler\%27s\_dilemma}.
In this $TD(2, 2, [2..100])$ game, each agent required to bid an integer in the interval of $[2..100]$.
The utility can be rewritten as
\[
\begin{cases}
	u_1(a_1, a_2) := \min(a_1, a_2) + 2\cdot sign(a_2 - a_1) \\
	u_2(a_1, a_2) := \min(a_1, a_2) + 2\cdot sign(a_1 - a_2)
\end{cases}
\]

\paragraph{Potentials.} We show that TD has a \textit{generalized ordinal potential} by introducing the following definitions first.
\begin{definition}[Generalized ordinal potential]
A function $P: A^N \mapsto \mathbb{R}$ is called a generalized ordinal potential, if
\[
\forall i\in N, \forall a_{-i} \in A^{N-1}, \forall a_i, a_i' \in A,
u_i(a_i, a_{-i}) - u_i(a_i', a_{-i}) > 0 \implies P(a_i, a_{-i}) - P(a_i', a_{-i}) > 0
\]
\end{definition}

We first show the above $TD(2, 2, [2..100])$ has a generalized ordinal potential as a warming up example, and then extend it to the general case.

\begin{theorem}
	$TD(2,2,[2..100])$ has a generalized ordinal potential.
\label{thm:TD2_potential} 
\end{theorem}

\begin{proof}
	We prove it by construction. We will construct such a $P$ divided by cases.
\begin{enumerate}
	\item If $a_1 = a_2 = \alpha$, then $P(\alpha, \alpha) = 2\times (101 - \alpha) - 1$
	\item If $a_1 = a_2 - 1 = \alpha - 1$, then $P(\alpha-1, \alpha) = 2\times (101 - \alpha)$. Same for $P(\alpha, \alpha-1)$.
	\item If $a_1 = a_2 - 1 - x = \alpha - 1 - x$, then $P(\alpha-1-x, \alpha) = 2\times (101 - \alpha) - x $. Also, same for $P(\alpha, \alpha-1-x)$.
\end{enumerate}
\begin{figure}[htbp]
	\centering
	\includegraphics[width=40mm]{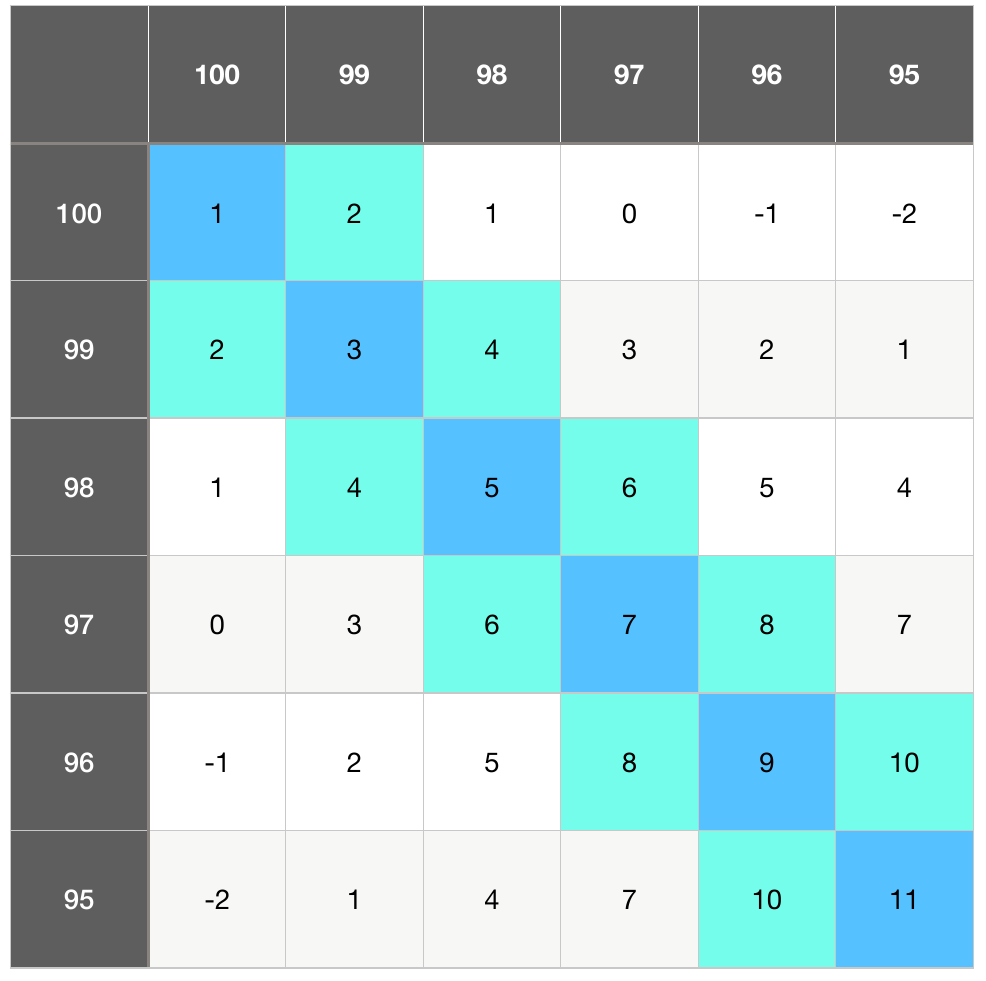}
	\caption{A general ordinal potential of $TD(2,2,[2..100])$.}
	\label{fig:potential}
\end{figure}

One can then examine the definition. W.l.o.g., we might as well consider fixing $a_2$, for two possible bids of agent 1, $a_1$ and $a_1'$,
\begin{enumerate}
	\item If $a_1 > a_2 + 1$ and $a_1' > a_2 + 1$, $u_1(a_1, a_2) = u_1(a_1', a_2) = a_2 - 2$. The premise of the implication is not satisfied, hence the implication is vacuously true.
	\item If $a_1 = a_1'$, the definition is also vacuously true. 
	\item If $a_1 = a_2 - 1$, then $u_1(a_1, a_2)= a_2 - 1 + 2 = a_2 + 1$, and $P(a_1, a_2)= 2 \times (101 - a_2)$, by the following sub-cases,
	\begin{enumerate}
		\item If $a_1' > a_2 + 1$, then
		\[
		\begin{cases}
		u_1(a_1', a_2)  = a_2 - 2  \\
		P(a_1', a_2) =2\times(101-a_1') - (a_1'-1-a_2) = 202 - 3a_1'+a_2+1 < 202-2a_2-2	
		\end{cases}
		\]
		\item If $a_1' = a_2 + 1$, then
		\[
		\begin{cases}
		u_1(a_1', a_2)= a_2 - 2 \\
		P(a_1', a_2)=2\times(101-a_1') = 202-2a_2-2
		\end{cases}
		\]
		\item If $a_1' = a_2 $, then
		\[
		\begin{cases}
		u_1(a_1', a_2)= a_2 \\
		P(a_1', a_2)=2\times(101-a_2) - 1 = 202-2a_2-1
		\end{cases} 
		\]
		\item If $a_1' = a_2 - 1$, then $a_1 = a_1'$, and the definition is vacuously true.
		\item If $a_1' < a_2 - 1 $, then
		\[
		\begin{cases}
		u_1(a_1', a_2)= a_1' + 2 < a_2 + 1 \\
		P(a_1', a_2)= 2\times (101 - a_2) - (a_2-1-a_1') = 202-3a_2+a_1'+1 < 202-2a_2
		\end{cases}
		\]
	\end{enumerate}
	All the above sub-cases satisfy $P(a_1, a_2) > P(a_1', a_2)$ whenever $u_1(a_1, a_2) > u_1(a_1', a_2)$.
\end{enumerate}
\end{proof}

The proof of Theorem~\ref{thm:TD2_potential} provides one with an illustrative example where a generalized ordinal potential function $P$ can be directly devised along with its rough structure.
Further, we show in general the multi-player TD game has a generalized ordinal potential, however, by a powerful but rather intuitive lemma so as to make our lives much easier, instead of explicitly coming up with a desired function $P$.

\begin{definition}[Improvement path]
An improvement path with respect to a bidding profile $\vec a$ is a maximal sequence of profitable deviations starting with $\vec a$. 
\end{definition}

\begin{lemma}[Finite improvement property] A game is said to have the finite improvement property~(FIP) if every improvement path is finite. Every finite game has a generalized ordinal potential iff it has the FIP.	
\end{lemma}

\begin{theorem}
	Any $TD(N,k,A)$ has a generalized ordinal potential.
\end{theorem}

\begin{proof}
	Starting from any bidding profile, a profitable deviation is for a bidder that is not currently the lowest bidder (including tied bids) to bid less than or equal to the currently lowest one. The improvement path cannot be any longer if every one reaches the lowest possible bid.
\end{proof}

\subsection{The Iterated Traveler's Dilemma}

An iterated traveler's dilemma (ITD) consists of a Markov chain of TDs.
More specifically, an ITD is a 5-tuple $\langle N, k, [\underline{A}..\overline{A}] , p, \gamma \rangle$, given as follows,
\begin{enumerate}
	\item For all $A\subseteq [\underline{A}..\overline{A}]$, $\langle N,k,A \rangle $ is a valid TD defined as previously, i.e. $TD(N,k,A)$.
	\item Given the current dilemma $TD(N,k,[A_1..A_2])$, and a bidding profile $\vec{a} \in [A_1..A_2]^N$, the successor dilemma can be changed to $TD(N,k,[op_1(\vec a)..op_2(\vec a)])$ w.p. $p$, or stay the same otherwise. In particular, $op_1$ and $op_2$ are the two operators that aggregates the last bidding profiles. For example, $op_1(\cdot)$ can a constant function that always outputs $\underline{A}$, and $op_2(\cdot)=\max(\cdot)$, implying an ITD whose available bidding space may shrink. Formally, a transition function is defined as
	\[
	T\Big(S' \Big| S=TD(N,k,[A_1..A_2]), \vec a \Big)=
	\begin{cases}
		p,& \text{ if } S'=TD(N,k,[op_1(\vec a)..op_2(\vec a)]) \\
		1-p,& \text{ if } S'=TD(N,k,[A_1..A_2]) \\
		0,& \text{ otherwise}
	\end{cases}
	\]
	\item A reward function that assigns each agent an immediate signal is defined as
	\[
	R_i\Big( TD(N,k,A), \vec a \Big) = 
	\begin{cases}
		u_i(\vec a;TD(N,k,A)),& \text{ if } a_1 = a_2 = \cdots = a_n \\
		0,& \text{ otherwise}
	\end{cases}
	\]
	\item The game ends immediately after when $a_1 = a_2 = \cdots = a_n$, and the rewards of this final round will also be collected. The overall objective for agent $i$ is to maximize the discounted accumulated rewards $\sum_{t=0}^{T-1}\gamma^t R_i^t$, where $T$ is the number of transitions in this episode.
\end{enumerate}

The game will end according to some rules, e.g., when the number of rounds exceeds a specified number (as in the main body of this paper), or when all agents bid the same (a way harder version).

\subsection{Sample Curves of the Training Phase}

With some main results summarized in Figure~\ref{fig:itd_results}, we here also provide the detailed training curve, as shown in Figure~\ref{app:fig:pd_training}.

\begin{figure}[htbp]
	\centering
	\includegraphics[width=56mm]{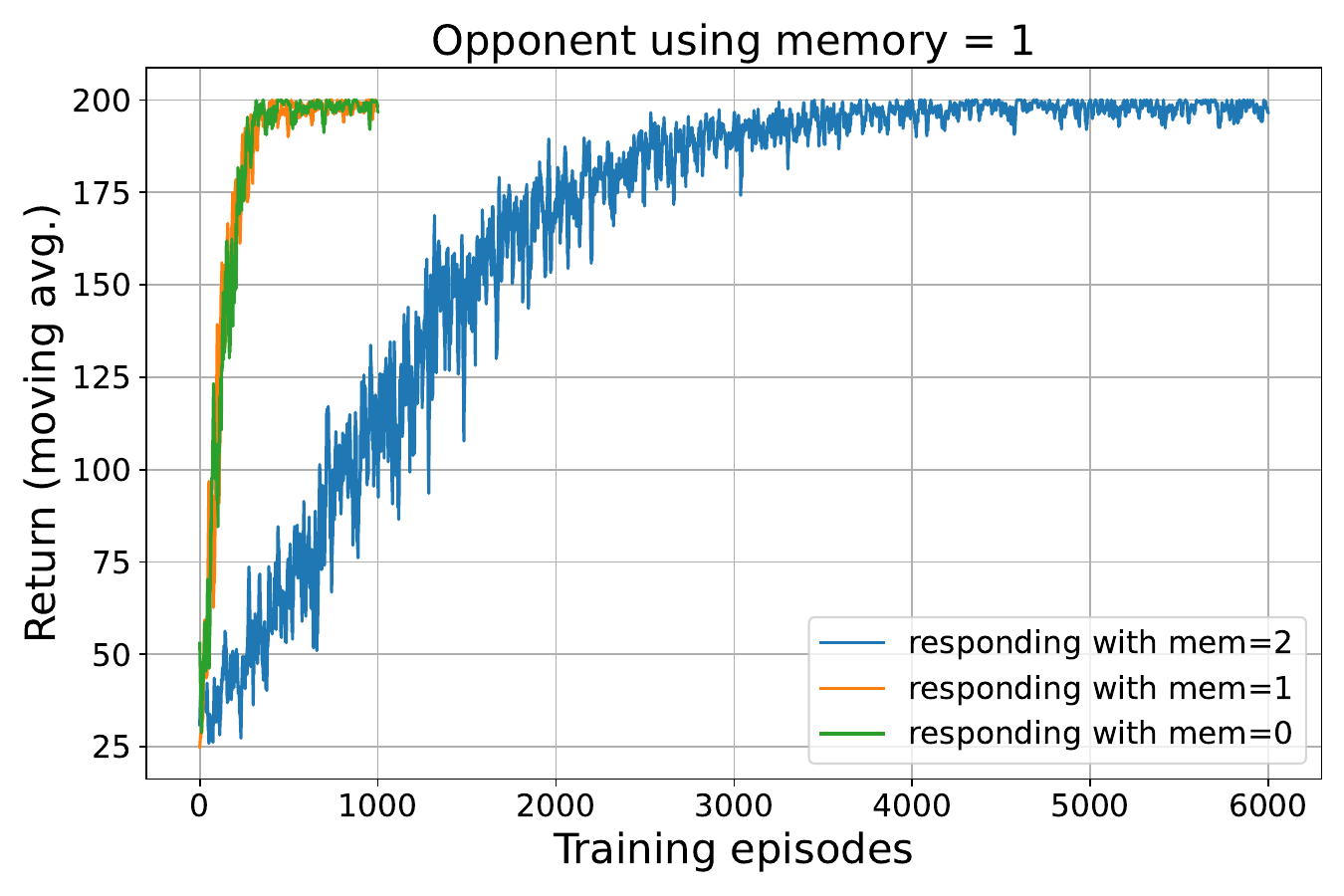}
	\includegraphics[width=56mm]{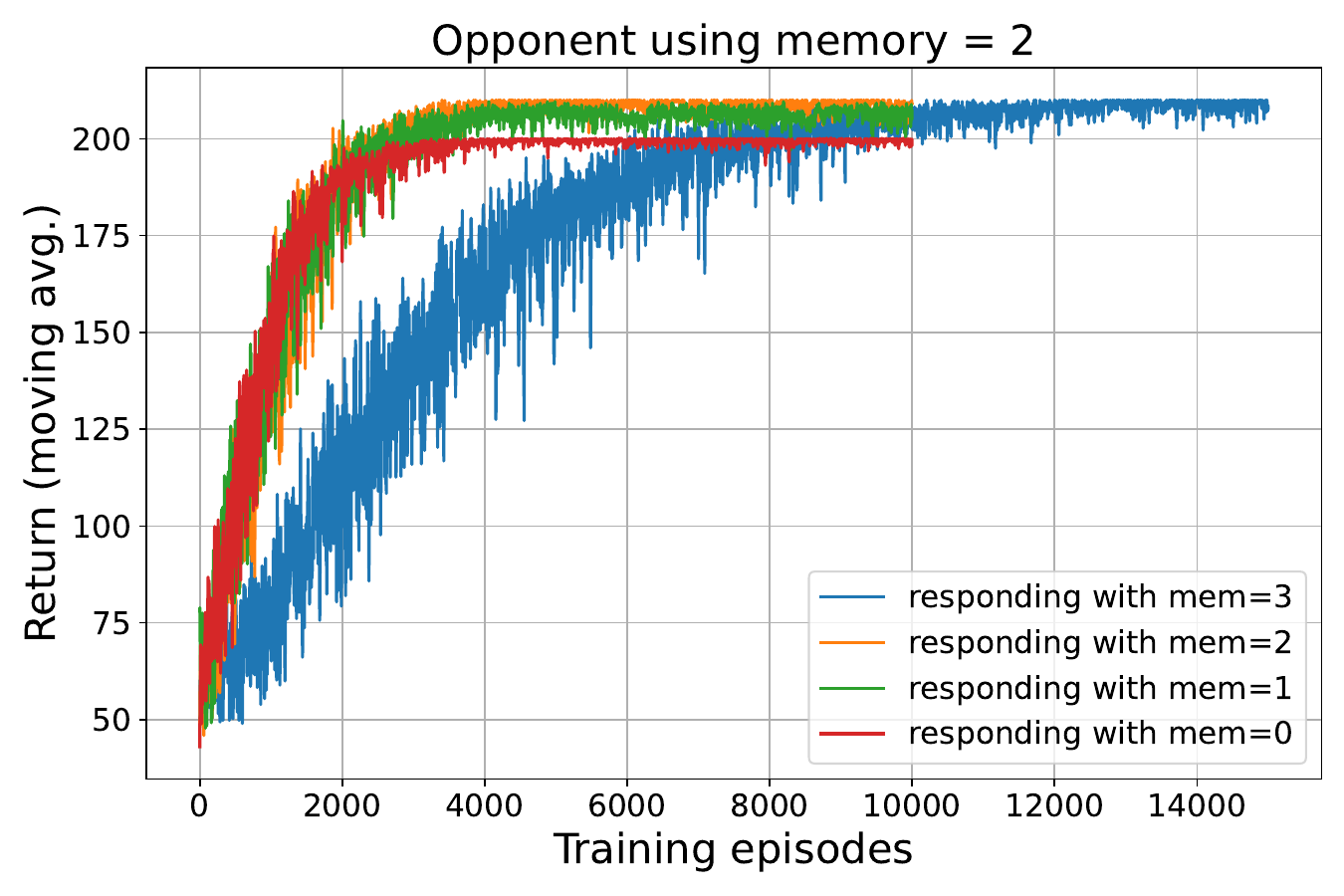}
	\includegraphics[width=56mm]{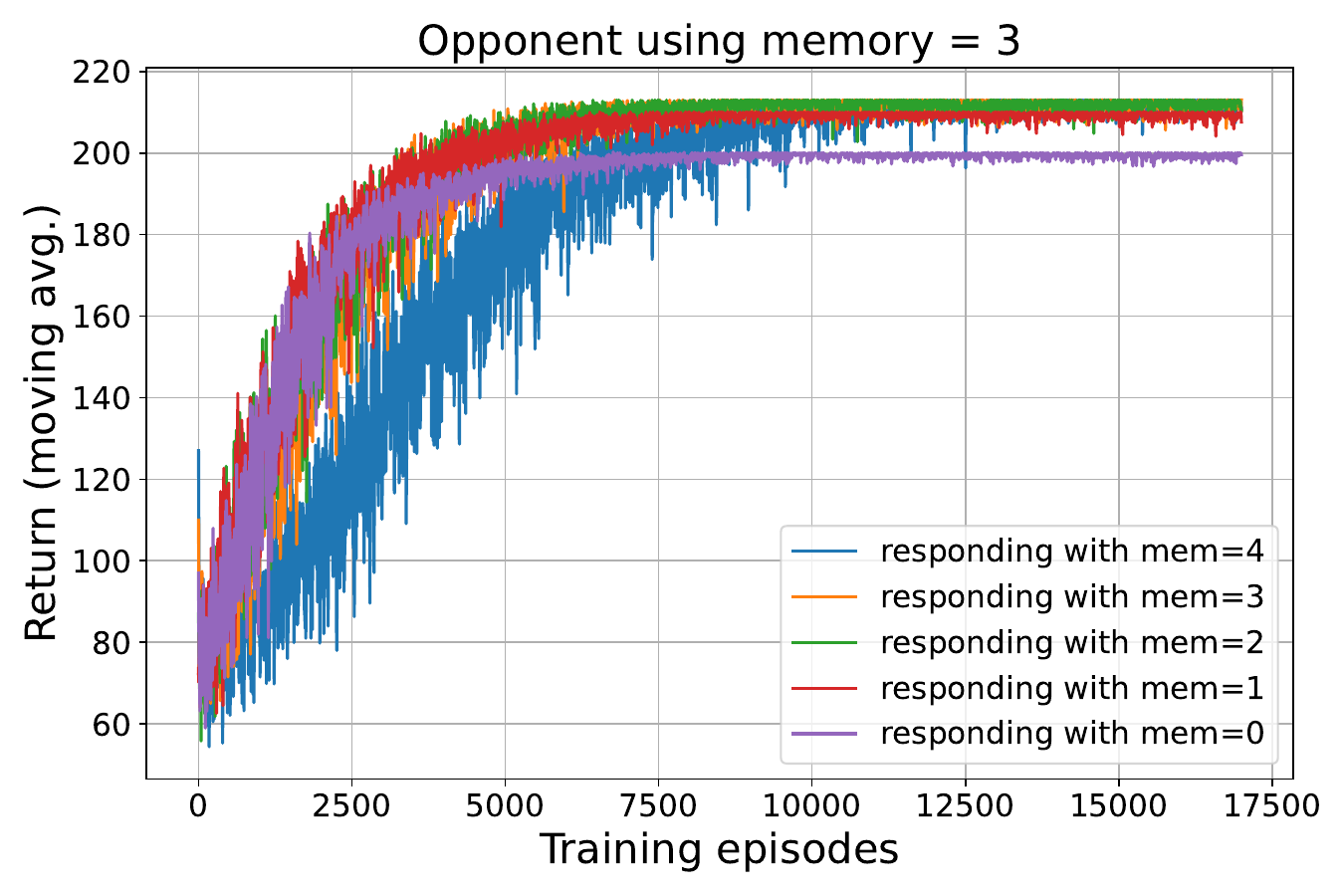}
	\includegraphics[width=56mm]{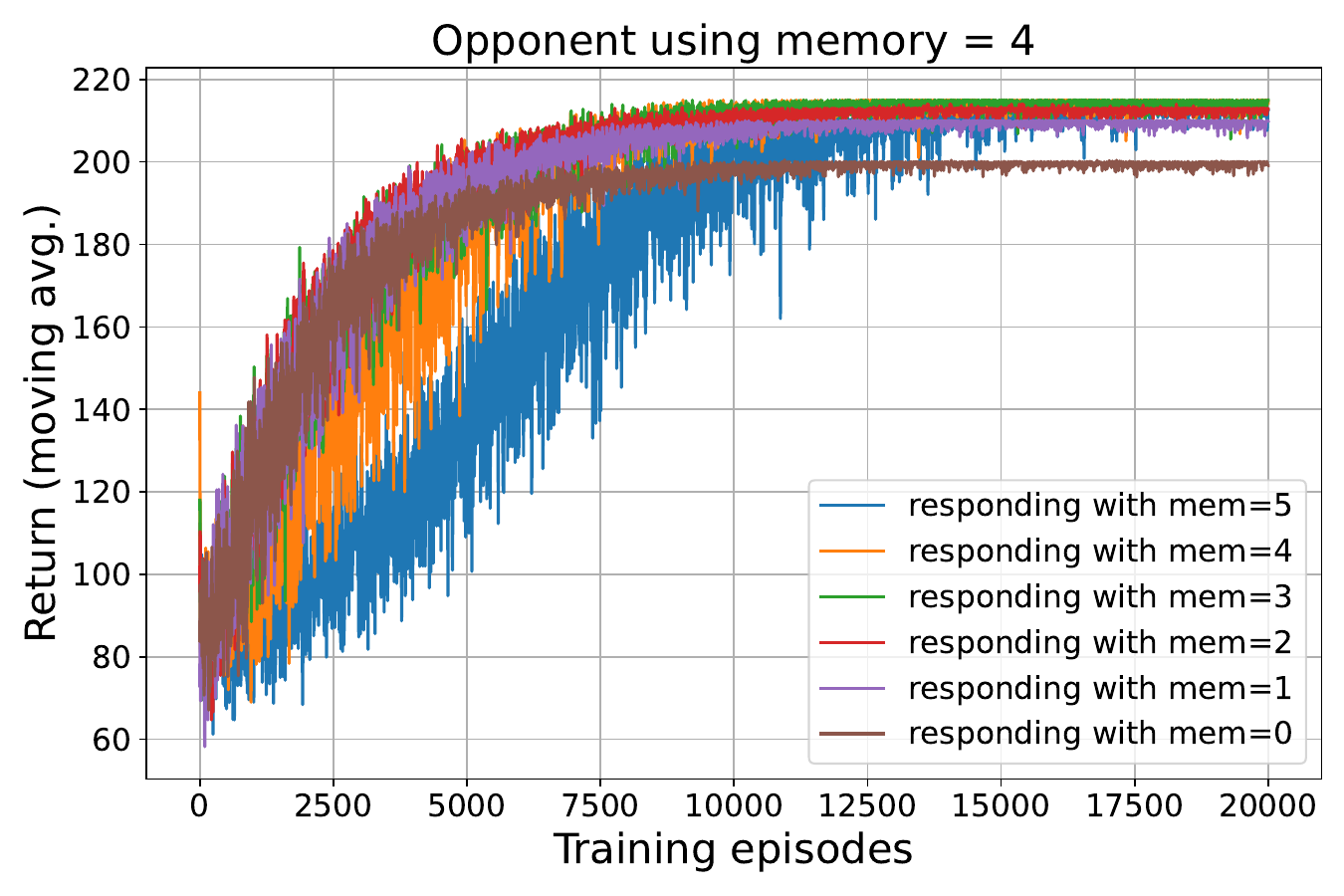}
	\includegraphics[width=56mm]{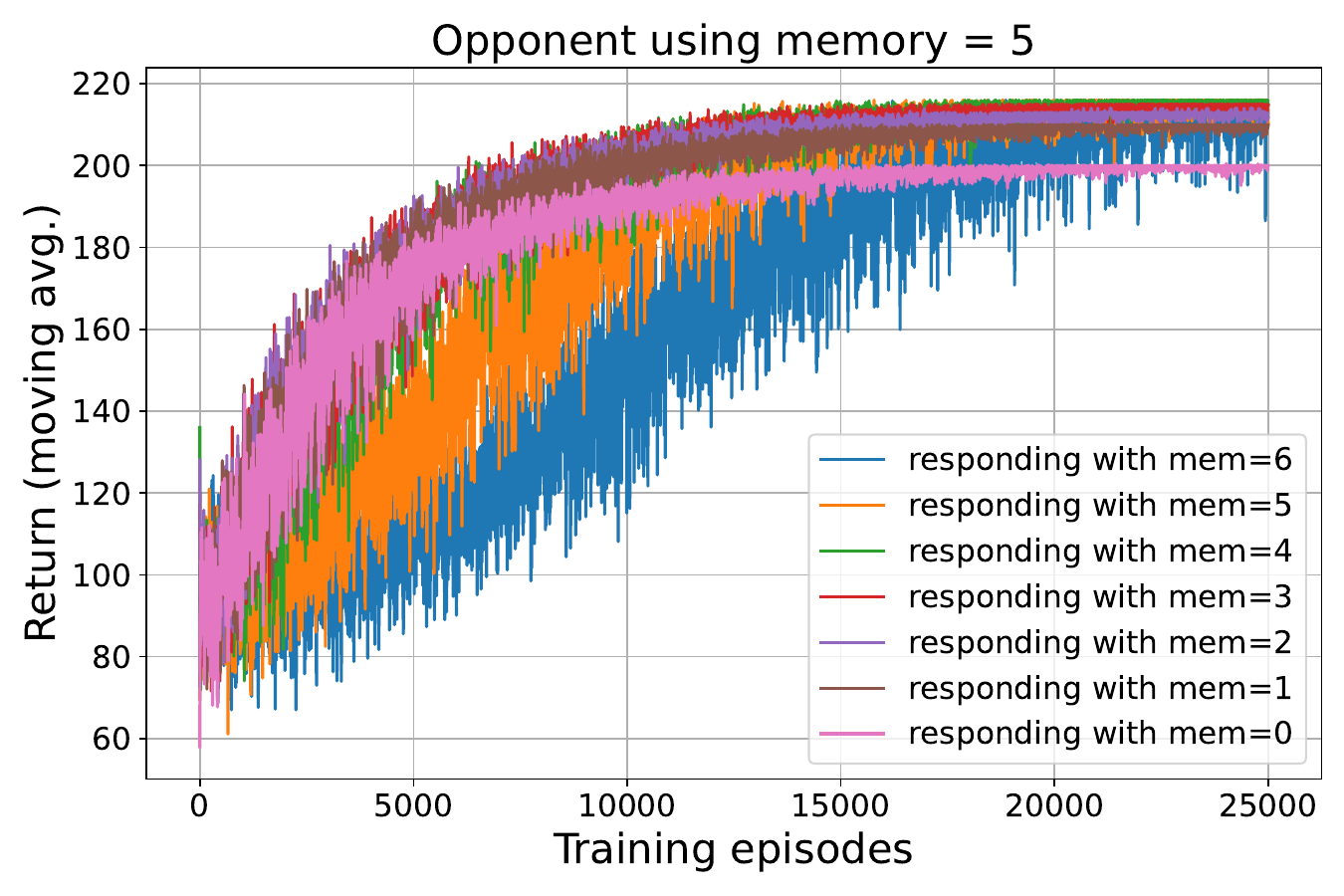}
	\caption{Training process for the Iterated Traveler's Dilemma. Each sub-figure denotes the best response (using various memory lengths) against a particular length of the memory used by the opponent.}
	\label{app:fig:pd_training}
\end{figure}

\section{Details for the Pursuit Domain}

\label{app:pursuit}

\begin{figure}[htbp]
	\centering
	\includegraphics[width=40mm]{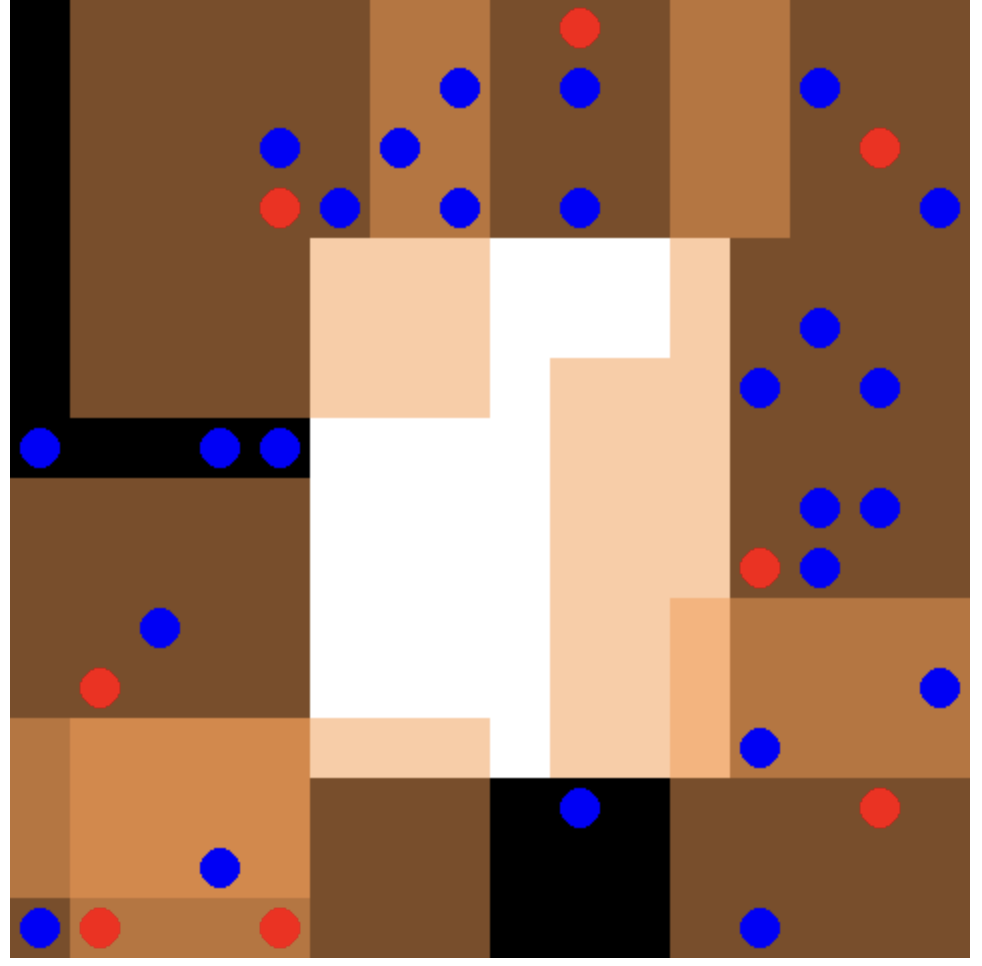}
	\caption{An illustration of the \textit{Pursuit} domain, where red dots are pursuers and blue dots are evaders. The orange squared centered at each red dot is the observed local area of that agent.}
	\label{app:fig:pursuit}
\end{figure}

\subsection{Detailed Experimental Settings}

\subsubsection{Environment Setup}
The Pursuit testbed~\cite{gupta2017cooperative}, illustrated in Figure~\ref{app:fig:pursuit},  allows the users to specify a few parameters, in order to deliver a customized environment.
We have involved four configurations, as listed in Table~\ref{app:tab:config}.
Specifically, we use the first three configurations to benchmark different RL algorithms in the next subsection, while the fourth configuration is used in the experiment in the main text.

\begin{table}[!ht]
\begin{tabular}{@{}lrrrrrrrrr@{}}
\toprule
Config            & \multicolumn{1}{l}{max\_cycles} & \multicolumn{1}{l}{width} & \multicolumn{1}{l}{height} & \multicolumn{1}{l}{\#evaders} & \multicolumn{1}{l}{\#pursuers} & \multicolumn{1}{l}{obs\_range} & \multicolumn{1}{l}{tag\_reward} & \multicolumn{1}{l}{catch\_reward} & \multicolumn{1}{l}{urgency\_reward} \\ \midrule
\texttt{HighCatch}         & 300                             & 10                          & 10                          & 10                             & 8                               & 7                              & 0.1                             & 5                                 & -0.1                                \\
\texttt{HighTag}           &                                 &                             &                             &                                &                                 &                                & 2                               & 0.1                               & -0.1                                \\
\texttt{SameTagCatch}      &                                 &                             &                             &                                &                                 &                                & 1                               & 1                                 & -0.1                                \\
\texttt{HighCatchLarge} &                                 & 16                          & 16                          & 20                             & 8                               & 7                              & 0.1                             & 3                                 & -0.1                                \\ \bottomrule
\end{tabular}
\vspace{2mm}
\caption{The detailed parameters of each configurations involved}
\label{app:tab:config}
\end{table}

\subsubsection{Hardware}
We use Linux Servers with NVIDIA GeForce RTX 3090 GPUs.

\subsubsection{Network Architecture}

All RL policies are equipped with a convolutional preprocessing network.
For this preprocessing network, we use three sequential Convolution layers, namely
\begin{enumerate}
	\item
	\texttt{Conv2d(input\_channels=3, output\_channels=32, kernel\_size=4, stride=1, padding=0)},
	\item 
    \texttt{Conv2d(input\_channels=32, output\_channels=64, kernel\_size=2, stride=1, padding=0)},
    \item
    \texttt{Conv2d(input\_channels=64, output\_channels=64, kernel\_size=2, stride=1, padding=0)},
\end{enumerate}
with a \texttt{ReLU} activation followed after each layer. Finally all the features are flattened and projected into a 512-dimensional vector.
For main body of the policy/value network, we use three-layer MLPs of the hidden dimension [512, 256, 256].

\subsubsection{RL Algorithm Parameters} We mainly adopt the implementation by~\cite{stable-baselines3}. For learning to find NEs under multi-agent settings, we equip each agent with the same network and RL algorithm, and make them learn independently. Specially, for the detail algorithmic parameters,
\begin{enumerate}
	\item In DQN, we have \texttt{batch\_size=256, exploration\_fraction=0.2};
	\item In PPO, we have \texttt{batch\_size=256}, making the policy/value network of the same architecture but updated independently;
	\item In A2C, we have \texttt{ent\_coef=0.01, vf\_coef=0.5,  n\_steps=400}, and same network setup with PPO.
\end{enumerate}
Other parameters that are not mentioned are set to be their default values.

\subsection{Benchmarking Different Deep RL Algorithms}

\begin{figure}[htbp]
	\centering
	\includegraphics[width=57mm]{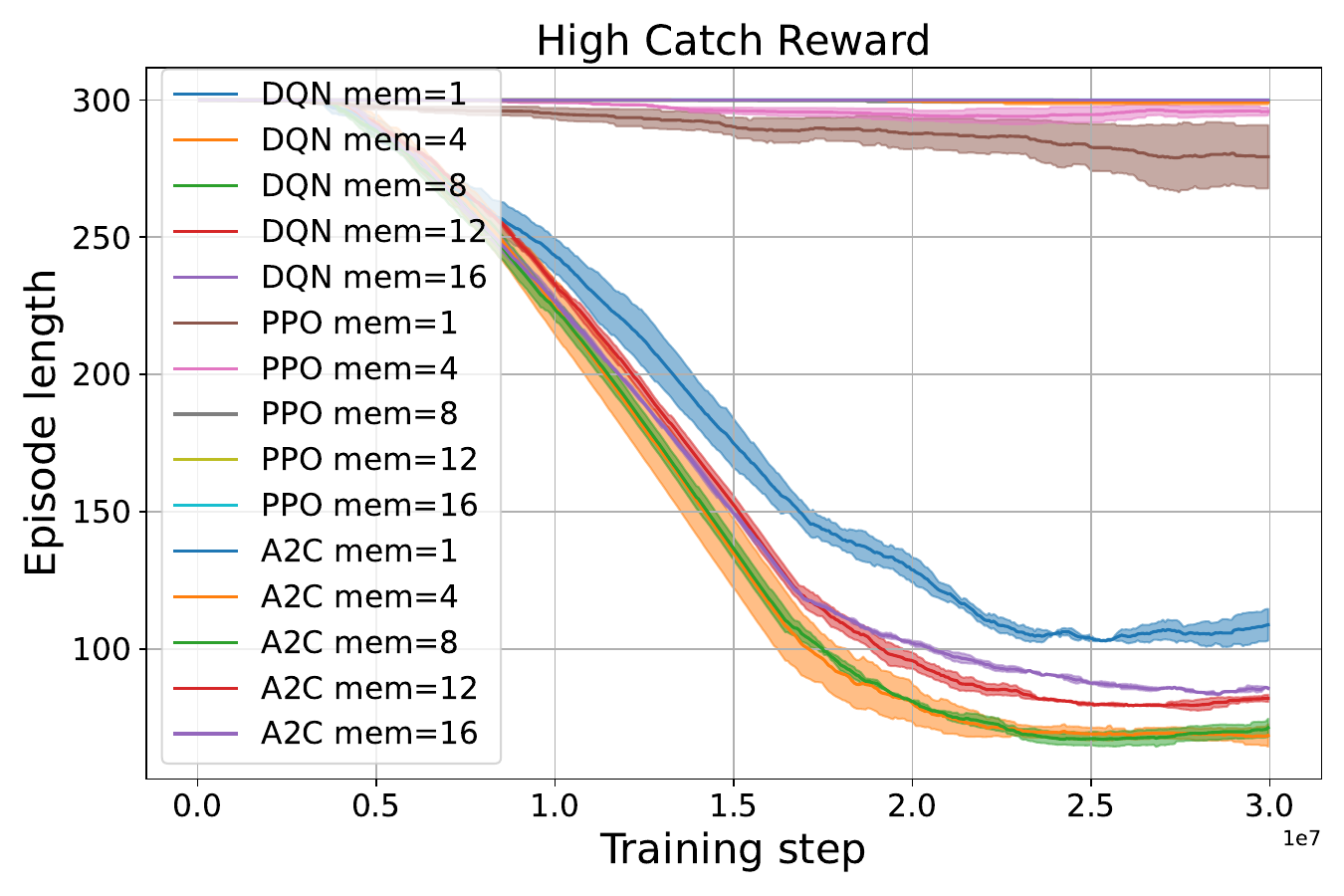}
	\includegraphics[width=57mm]{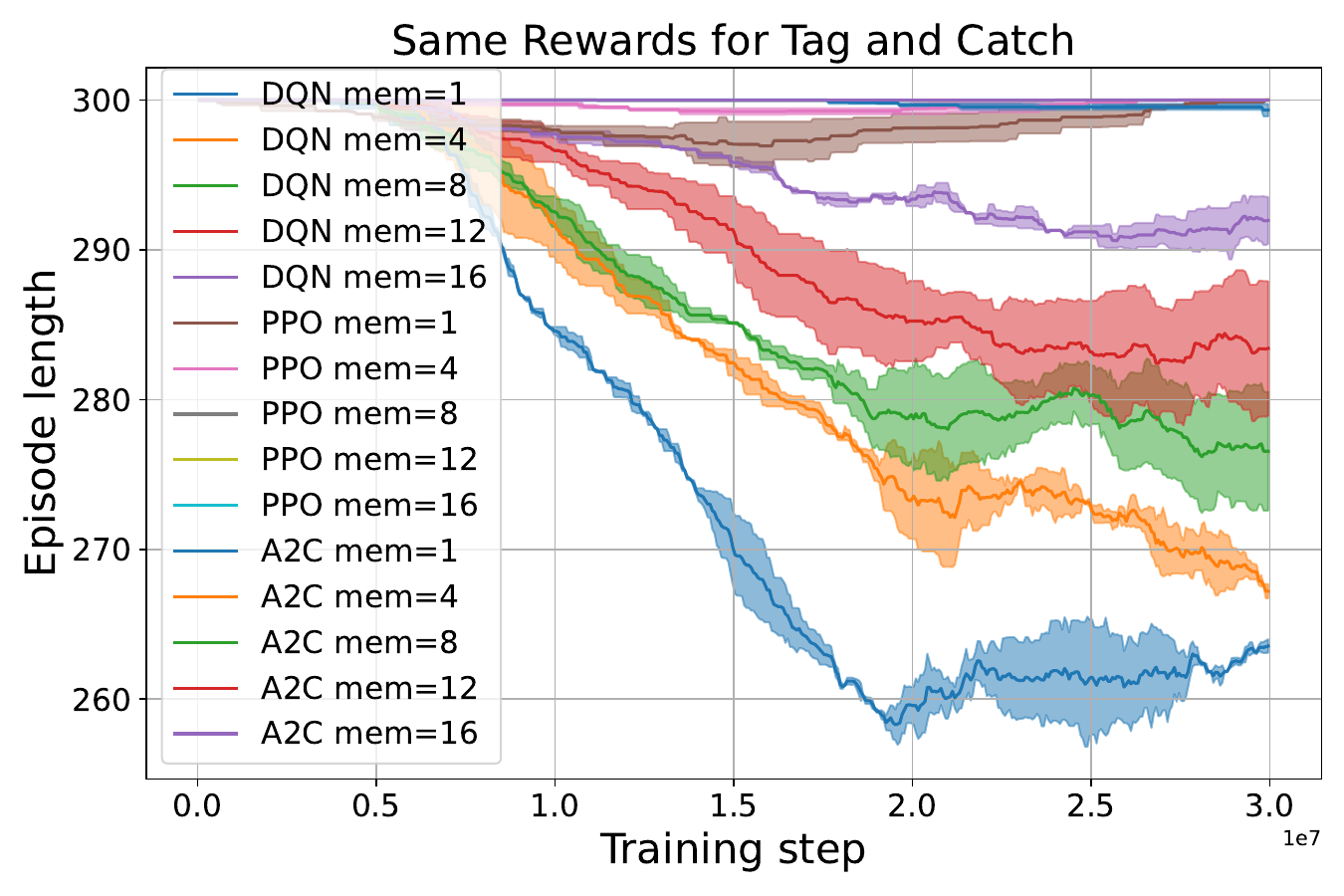}
	\includegraphics[width=57mm]{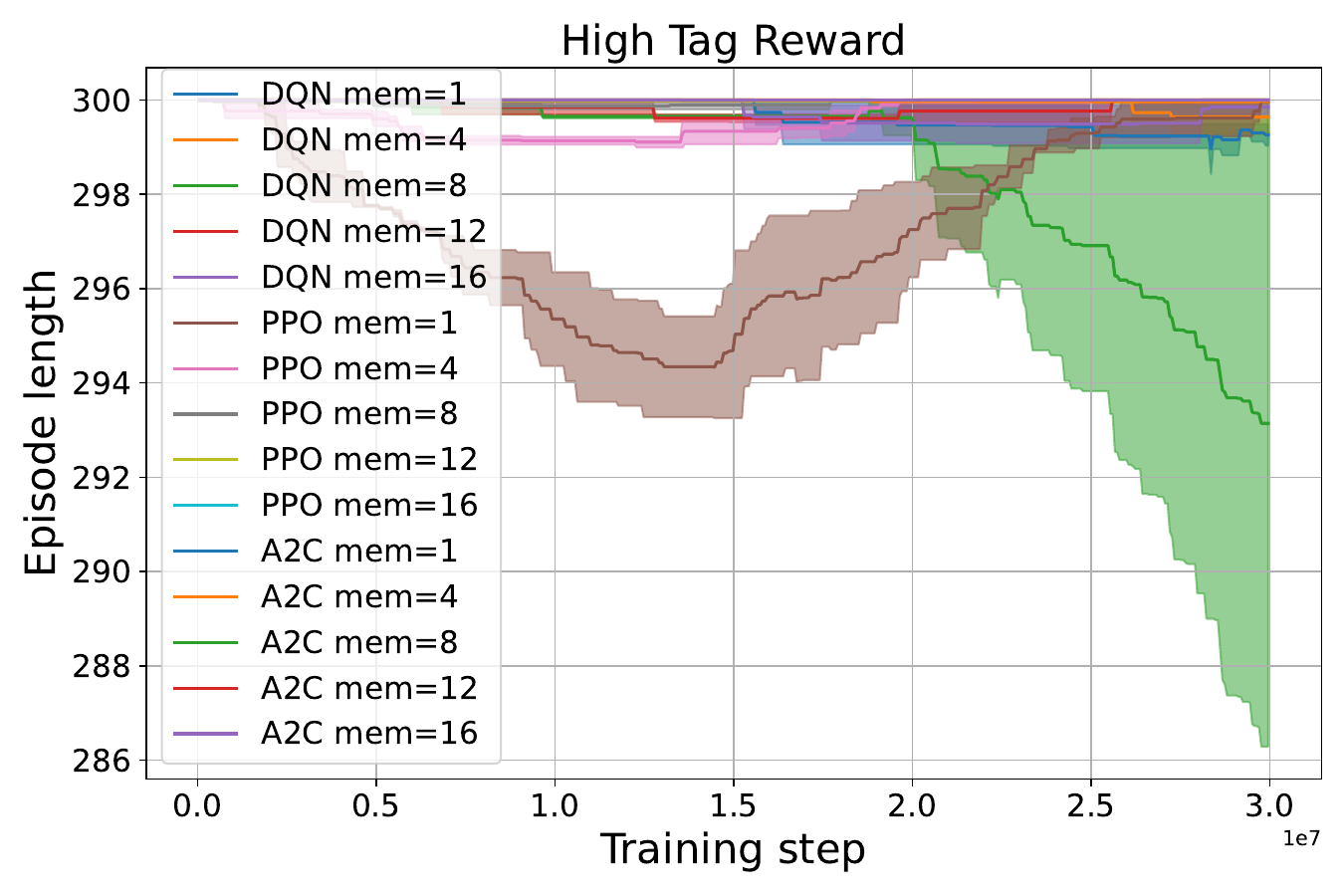}
	\caption{Benchmarking results using different algorithms and memory of varying lengths for the Pursuit domain instantiated under three configurations. From left to right, each figure represents the experimental results for the configuration \texttt{HighCatch}, \texttt{SameTagCatch}, and \texttt{HighTag}, respectively.}
	\label{app:fig:bm_pursuit}
\end{figure}

Additional to the results presented in the main text that are obtained using DQN.
We benchmark all three algorithms, including the other two, namely PPO and A2C, with the results summarized in Figure~\ref{app:fig:bm_pursuit}.
Among these three, DQN performs always the most effectively:
\begin{enumerate}
	\item For \texttt{HighCatch}, DQN can always effectively make the hunting process shorter. A longer-memory setting leads to a shorter period of hunting. In contrast, PPO and A2C do not result in successful cooperative hunting strategies.
	\item For \texttt{SameTagCatch}, the way we set-up the rewards shall lead to a solution where each agents is supposed to first tag evaders on its own, but ends up catching and removing them under cooperation with the pursuers.
	Therefore, DQN agents equipped with long memories tend to hunt those evaders more ``slowly'', leaving more time for themselves to tag the evaders to obtain sufficient rewards before the game terminates.
	In contrast, we found that although the episode lengths of PPO and A2C remain high, they are not learning to tag agents, indicated by their low returns during training.
	\item For \texttt{HighTag}, agents shall figure out that cooperatively catching evaders is not a desired strategy; rather, independently tagging the evaders without removing them from the game is supposed to be the best strategy. Thus, in most cases, the agents operates for the full episode. Except for one case of DQN with 8-memory, it seems that the agent is still confused about whether to adopt a tagging-without-hunting strategy.
	   
\end{enumerate}

\end{document}